\documentclass[a4paper,oneside,reqno,10pt]{amsart}

\usepackage[english]{babel}
\usepackage[utf8]{inputenc}		
\usepackage[T1]{fontenc}
\usepackage[scaled]{helvet}		
\usepackage{courier}			
\usepackage{eulervm}			
\usepackage[bb=boondox,bbscaled=1.05,scr=dutchcal]{mathalfa}	%
\usepackage{dsfont}
\usepackage{yfonts}
\normalfont

\usepackage{titlesec}
\titleformat{\section}
  {\centering\LARGE\normalfont\bfseries\scshape}{\thesection.}{1em}{}
\titleformat{\subsection}
  {\centering\Large\normalfont\bfseries}{\thesubsection.}{1em}{}
  \titleformat{\subsubsection}
  {\normalfont\large\scshape}{\thesubsubsection.}{1em}{}

\usepackage{tabularx,booktabs}			
\usepackage{caption,subcaption}			
\usepackage{fullpage}				    
\usepackage[english]{varioref}			
\usepackage[dvipsnames]{xcolor}			
\usepackage{hyperref}					
\usepackage{bookmark}					
\usepackage{textcomp}
\usepackage{enumerate}
\usepackage{leftidx,tensor}
\usepackage{url}
\usepackage{geometry}
\usepackage[capitalize]{cleveref}

\usepackage{parskip}

\definecolor{webbrown}{rgb}{0.65, 0.16, 0.16}

\hypersetup{
colorlinks=true, linktocpage=true, pdfstartpage=1, pdfstartview=FitV,breaklinks=true, pdfpagemode=UseNone, pageanchor=true, pdfpagemode=UseOutlines,plainpages=false, bookmarksnumbered, bookmarksopen=true, bookmarksopenlevel=1,hypertexnames=true, pdfhighlight=/O,urlcolor=webbrown, linkcolor=RoyalBlue, citecolor=ForestGreen}

\usepackage{graphicx}		
\usepackage{float}
\usepackage{tikz}			
\usepackage{tikz-cd}		

\usepackage[textsize=footnotesize,textwidth=0.85in]{todonotes}
\presetkeys%
    {todonotes}%
    {color=Dandelion}{}%
\usepackage{amsmath,amssymb,amsthm,mathtools,amsfonts}
\usepackage{bm,braket,faktor,stmaryrd}
\usepackage{nicematrix}
\numberwithin{equation}{section}

\renewenvironment{abstract}{%
\small\begin{center}
\begin{minipage}{.8\textwidth}
}
{\par\noindent\end{minipage}\end{center}\vspace{3 em}}
\makeatletter
\renewcommand\@maketitle{%
\hfill
\begin{center}\begin{minipage}{0.9 	\textwidth}
\centering
\vskip 6em
\let\footnote\thanks 
{\LARGE\bfseries \@title \par }
\vspace{1 em}
\vskip 1 em
{\large \@author \par}
\vspace{3.5 em}

\end{minipage}\end{center}
\par
}
\makeatother

\usepackage{afterpage}

\usepackage{titletoc}
\patchcmd{\section}{\normalfont}{\normalfont\Large}{}{}

\titlecontents{chapter}
[0.2em] %
{\bigskip}
{\makebox[2em][r]{\thecontentslabel.}\hspace{0.333em}}
{\hspace*{-2em}}
{\hfill\contentspage}[\smallskip]

\titlecontents{section}
[0.2em]
{\large}
{\thecontentslabel.\hspace{3pt}}
{}
{\enspace\titlerule*[0.5pc]{.}\contentspage}

\titlecontents*{subsection}
[1em]
{}
{\thecontentslabel. \hspace{3pt}}
{}
{\enspace\titlerule*[0.35pc]{--}\contentspage \hspace{1cm}}
[]
\newcommand{\dd}{\mathrm{d}}

\DeclareMathOperator*{\Res}{Res}
\DeclareMathOperator{\Aut}{Aut}

\newcommand{\BM}{\big[\begin{smallmatrix}}
\newcommand{\EM}{\end{smallmatrix}\big]}

\theoremstyle{plain}
\newtheorem{thm}{Theorem}[section]

\newtheorem{prop}[thm]{Proposition}
\newtheorem{conj}[thm]{Conjecture}
\newtheorem{lem}[thm]{Lemma}
\newtheorem{cor}[thm]{Corollary}

\theoremstyle{definition}
\newtheorem{defn}[thm]{Definition}
\newtheorem{rem}[thm]{Remark}

\title{Whittaker vectors at finite energy scale, topological recursion and Hurwitz numbers}

\usepackage{authblk}

\address[\textsc{Ga\"etan Borot}]{ \hfill\hfill \linebreak
	\textup{Institut für Mathematik und Institut für Physik, Humboldt-Universität zu Berlin, \hfill\hfill \linebreak Unter den Linden 6, 10099 Berlin, Germany.}  \hfill \hfill \linebreak  \textit{Visiting at:} \textup{Institut des Hautes \'Etudes Scientifiques \hfill\hfill \linebreak Le Bois-Marie, 35 route de Chartres, 91893 Bures-sur-Yvette, France.} \hfill\hfill \linebreak \textit{Email:} \textup{gaetan.borot@hu-berlin.de}}

\address[\textsc{Nitin Kumar Chidambaram}]{\hfill\hfill\linebreak
	\textup{School of Mathematics, University of Edinburgh \hfill\hfill\linebreak James Clerk Maxwell Building, Peter Guthrie Tait Rd, Edinburgh EH9 3FD, United Kingdom. \hfill \hfill \linebreak \textit{Email:} \textup{nitin.chidambaram@ed.ac.uk}}}

\address[\textsc{Giacomo Umer}]{ \hfill\hfill\linebreak
	\textup{Institut für Mathematik, Humboldt-Universität zu Berlin \hfill\hfill\linebreak Unter den Linden 6, 10099 Berlin, Germany. }\hfill\hfill\linebreak \textit{Email:} \textup{umergiac@hu-berlin.de}}

\author{Ga\"etan Borot, Nitin Kumar Chidambaram, Giacomo Umer}

	\thanks{blblba}

\begin{document}

\maketitle

\begin{abstract}
	We upgrade the results of Borot--Bouchard--Chidambaram--Creutzig \cite{BBCC} to show that the Gaiotto vector in $4d$ $\mathcal{N} = 2$ pure supersymmetric gauge theory admits an analytic continuation with respect to the energy scale (which can therefore be taken to be finite, instead of infinitesimal), and is computed by topological recursion on the (ramified) half Seiberg--Witten spectral curve. This has a number of interesting consequences for the Gaiotto vector: relations to intersection theory on $\overline{\mathcal{M}}_{g,n}$ in at least two different ways, Hurwitz numbers, quantum curves, and (almost complete) description of the correlators as analytic functions of $\hslash$ (instead of formal series). The same method is used to establish analogous results for the more general Whittaker vector constructed in the recent work of Chidambaram--Do{\l}{\k{e}}ga--Osuga \cite{CDO24}.
\end{abstract}

	\vspace{-.2 cm}
\date{}

\tableofcontents

\newpage

\section{Motivation and overview}
\label{sec:intro}

\medskip

Schiffman--Vasserot \cite{SV} and Maulik--Okounkov \cite{MO} pro\-ved that the equivariant cohomology (suitably interpreted) of the moduli space of rank $r$ torsion free sheaves on $\mathbb{P}^2$  framed at infinity (a.k.a. the moduli space of instantons) is a Verma module for the principal $\mathcal{W}(\mathfrak{gl}_r)$-algebra. We consider the Whittaker vector in  the   $\mathcal{W}(\mathfrak{gl}_r)$-algebra module $\mathcal V_\Lambda $ (where we adjoin a formal variable $\Lambda$ to the Verma module) satisfying
\begin{equation}
	\label{Whittakercondition} \forall (i,k) \in [r] \times \mathbb{Z}_{> 0} \qquad W_k^i\ket{\mathfrak{G}} = \delta_{i,r}\delta_{k,1}\Lambda^r \ket{\mathfrak{G}}.
\end{equation} 
The equivariant parameters $(\epsilon_1,\epsilon_2)$ for the action of the torus $\left(\mathbb C^*\right)^2$ on $ \mathbb{P}^2$ specify the level of $\mathcal{W}(\mathfrak{gl}_r)$, and
\[
\epsilon_1 = - \epsilon_2 = \hslash^{\frac{1}{2}}
\]
corresponds to the self-dual level. The equivariant parameters $\mathbf{Q} = (Q_1,\ldots,Q_r)$ for the action of the Cartan of $\mathfrak{gl}_r$ specify the highest weight $(\mathbf{Q}/\epsilon_1 - \text{Weyl vector})$. The moduli space of sheaves with second Chern class $d$ carries a fundamental class $\ket{1^d}$ in equivariant cohomology, and it is encoded precisely in the Whittaker vector characterised by \eqref{Whittakercondition}
\[
\ket{\mathfrak{G}} = \sum_{d \geq 0} \Lambda^{rd} \ket{1^d}.
\]
In particular, the Whittaker vector exists in $\mathcal V_\Lambda $ and is unique. These results provide a mathematical ground for the celebrated Alday--Gaiotto--Tachikawa conjecture \cite{AGT}, and $\ket{\mathfrak{G}}$ is sometimes called the Gaiotto vector. The Poincar\'e pairing in cohomology matches the Kac--Shapovalov form on the Verma module, and the squared-norm reconstructs the  (instanton part of the) Nekrasov partition function
\[
Z_{\text{Nek}} = \langle \mathfrak{G}\ket{\mathfrak{G}} = \sum_{d \geq 0} \Lambda^{2rd} \langle 1^{d}\ket{1^{d}},
\]
which counts instantons in $\mathcal{N} = 2$ pure supersymmetric gauge theory on $\mathbb{R}^4$ with gauge group $\text{U}_r$. The $\text{SU}_r$ theory can be retrieved from the $\text{U}_r$ theory by specialising to $\sum_{a = 1}^{r} Q_a = 0$ and removing an explicit $\text{U}_1$ contribution. We prefer to work with $\text{U}_r$. In particular, we work with the $\mathcal{W}(\mathfrak{gl}_r)$-algebra  whose generators include one of conformal weight $i = 1$ in \eqref{Whittakercondition} and we do not impose that the sum of the $Q_a$ vanishes.

In gauge theory the parameter $\Lambda$ is interpreted as an energy scale. The Gaiotto vector $\ket{\mathfrak{G}}$ and the Nekrasov partition function are thus proved to exist at least if $\Lambda$ is considered as a formal parameter near $0$. Based on the well-known free field presentation of $\mathcal{W}(\mathfrak{gl}_r)$,  the Whittaker constraints \eqref{Whittakercondition} were realised as Airy structures in \cite{BBCC}, permitting the reconstruction of the all-order $\hslash = (-\epsilon_1\epsilon_2) \rightarrow 0$ expansion of the Gaiotto vector via a topological recursion. In the self-dual case, this is a variant of the Chekhov--Eynard--Orantin topological recursion on the unramified spectral curve defined by
\[
\prod_{a = 1}^{r} \bigg(y - \frac{Q_a}{x}\bigg) = 0,
\]
see Section~\ref{S2}.

The purpose of this article is to extend these results in the self-dual case to understand the analytic properties of the Gaiotto vector in the parameter $\Lambda$. In particular we want to be able to set $\Lambda$ to a finite non-zero value. We show in \cref{th:main} (proved in Section~\ref{S3}) that the coefficients in the $\hslash$-expansion of the Gaiotto vector are the all-order series expansion of meromorphic multi-differentials on the algebraic curve
\begin{equation}
	\label{finiteLambdacurve}
	S_{\Lambda} : \qquad \prod_{a = 1}^{r} \bigg(y - \frac{Q_a}{x}\bigg) + \frac{(-\Lambda)^{r}}{x^{r + 1}} = 0.
\end{equation}
which depends analytically on $\Lambda \in \mathbb{C}^*$.   Moreover, we show that these multi-differentials are computed by the Chekhov--Eynard--Orantin topological recursion on the ramified spectral curve defined by $S_{\Lambda}$. We will refer to this  curve $S_\Lambda$ as the half Seiberg--Witten curve, following the terminology used in \cite{DHS09}.\footnote{We note that the curve $S_\Lambda$ is neither the UV curve nor the Seiberg--Witten curve. Rather, $S_\Lambda$ is a degenerate version of the Seiberg--Witten curve as we explain further in \cref{rem:SW}.} A word of warning: 
in the gauge theory literature, our Coulomb branch parameters $Q_1,\cdots, Q_r$ are usually denoted $\mathrm{a}_1,\ldots,\mathrm{a}_r$.

A technical novelty is that we show directly that the $\mathcal{W}$-constraints at finite $\Lambda$, although they do not form an Airy structure, can still be solved by topological recursion \emph{after analytic continuation}. While the Whittaker constraints for formal $\Lambda$ led to topological recursion on an unramified spectral curve $S_0$, at small enough non-zero $\Lambda$ they become $\mathcal{W}$-constraints around $x = \infty$ on a ramified spectral curve $S_{\Lambda}$ which imply $(r -1)$ copies of Virasoro constraints (one copy at each ramification point). The latter form an Airy structure and its solution is precisely given by topological recursion on $S_{\Lambda}$. As the family $S_{\Lambda}$ depends analytically on $\Lambda$, we can then take $\Lambda$ to be any value in $\mathbb{C}^*$. We comment more precisely on the role of analytic continuation in Section~\ref{S2ana}.

The method of our proof can be adapted without difficulty for the other gauge groups treated in \cite{BBCC}. The non self-dual case, more precisely for $\alpha = \epsilon_1 + \epsilon_2$ finite and generic while $\hslash = - \epsilon_1\epsilon_2$ is kept formal, should also be governed by topological recursion where the spectral curve \eqref{finiteLambdacurve} is replaced by the D-module on $\mathbb{P}^1$ generated by
\begin{equation}
	\label{Dmodule}
	\left((\epsilon_1 + \epsilon_2)\partial_x - \frac{Q_1 - (\epsilon_1 + \epsilon_2)}{x}\right) \cdots \left((\epsilon_1 + \epsilon_2)\partial_x - \frac{Q_r - (\epsilon_1 + \epsilon_2)}{x}\right) + \frac{(-\Lambda)^r}{x^{r + 1}}.
\end{equation} 
This was established for formal $\Lambda$ in \cite{BBCC}, which involves the D-module \eqref{Dmodule} where $\Lambda$ is set to $0$. The study of finite $\Lambda$ for the non self-dual case is technically more demanding as we have to deal with solutions of finite-order ODEs instead of meromorphic functions on an algebraic curve. This is left for the future.

This technical novelty is interesting beyond gauge theory: it allows understanding how (simpler) constraints at the ramification points can arise from (more complicated) constraints at $\infty$ in various problems of enumerative geometry. Recently, Chidambaram, Do{\l}{\k{e}}ga and Osuga  constructed another Whittaker vector for $\mathcal{W}(\mathfrak{gl}_r)$ defined for formal $\Lambda$ (we call it the CDO vector) that encodes $b$-Hurwitz numbers, which are counts of so-called generalised branched covers of the $2$-sphere  \cite{CD22}. Generalised branched covers allow for non-orientable coverings, and the count of orientable branched covers, i.e. classical Hurwitz theory, corresponds to $b = 0$ and to the self-dual level for $\mathcal{W}(\mathfrak{gl}_r)$. In this case, \cref{th:main2} shows the existence of analytic continuation to $\Lambda \in \mathbb{C}^*$ and the computation of these Hurwitz numbers by the topological recursion on the spectral curve
\[
\prod_{a = 1}^{r} \left(\frac{P_a}{x} + y\right) + \frac{1}{\Lambda^r} \prod_{a = 1}^{r -1} \left(\frac{Q_a}{x} - y\right) = 0.
\]
This result was announced as \cite[Theorem 5.1]{CDO24} and we prove it in Section~\ref{S4}. It was used in \cite{CDO24} to give an alternative proof of the recent celebrated result of \cite{ACEH3,BDBKS} that topological recursion computes rationally weighted classical Hurwitz numbers.

For both the Gaiotto and the CDO vector, the added value of having established topological recursion on ramified spectral curves is that we can benefit from a rich and well-developed theory to derive several remarkable consequences. This is discussed in \cref{S5}. First, it allows us to give several representations of these vectors in terms of intersection numbers on $\overline{\mathcal{M}}_{g,n}$ (\cref{S51}), and establish relations to Hurwitz theory (\cref{S5H} with \cref{cor:HforG} and \cref{cor:HforZ}). In particular, for $r = 2$ the Gaiotto vector is very explicitly expressed in terms of intersection indices of triple Hodge classes, or of the deformed Theta class, see \cref{pr:intGai2}. These relations can be considered as a multifold interpretation of 4d supersymmetric gauge theory in terms of curve counting; as we comment in Section~\ref{SGT}, it is different in nature from the 2d Yang--Mills/Hurwitz theory correspondence studied by Gross and Taylor \cite{GrossTaylor,GrossTaylor2} and  more recently Novak \cite{GTNovak}. Second, we derive in \cref{pr:QGai} and \cref{pr:CDOQ} the quantum curves associated to those Whittaker vectors, and we discuss in detail the construction of a basis of analytic solutions. The only step which we do not complete is the analytic description of the connection coefficients, see Remark~\ref{connecrem}. Third, we obtain determinantal formulae for the correlation functions in Proposition~\ref{bispiq}, with kernels given in terms of the previously discussed bases of functions in Proposition~\ref{bispiq2}. Fourth, this analysis and the relation to gauge theory leads us in \cref{S54} to formulate conjectures for the topological recursion free energies of the half Seiberg--Witten and CDO spectral curves. These conjectures have been proved by Hock \cite{Hockproo} after the first version of this article was released.

\medskip

\subsubsection*{Notations}

\medskip

For any positive integer $i$, we use the notation $[i] := \{1,2,\ldots, i\}$. We also use the notation $z_{[i]}$ to denote the set $\{z_1, z_2, \ldots, z_i\}$. The symbol $\sqcup$ stands for a disjoint union. The algebra of formal Laurent series in a variable, say  $ \hslash$, is denoted $\mathbb{C}(\!(\hslash)\!)$.

\medskip

\subsubsection*{Acknowledgements}

\medskip

We thank the anonymous referee for their careful reading of the paper and the suggested improvements. We thank V.~Bouchard and T.~Creutzig for discussions on related topics, S.~Shadrin for his remarks on integrability and determinantal formulae, V.~Fantini for remarks on asymptotics of solutions of generalised hypergeometric differential equations, A.~Giacchetto and D.~Scazzuso for bringing relevant references to our attention, and R.~Kramer for his remark on the proof of the main theorem. G.B. thanks the IH\'ES for hospitality and excellent working conditions allowing the completion of this project. N.K.C. acknowledges the support of the ERC Starting Grant 948885, and the Royal Society University Research Fellowship. G.U. thanks the Berlin Mathematical School for financial support.

\vspace{1cm}

\section{Background and main results}

\label{S2}

\vspace{0.5cm}

\subsection{\texorpdfstring{$\mathcal{W}$}{W}-algebras and formal Whittaker vectors}

\label{S12}

\medskip

\subsubsection{Heisenberg algebra and \texorpdfstring{$\mathcal{W}(\mathfrak{gl}_r)$}{W(glr)}-algebra}
\label{S121}
\medskip

Let $r \geq 2$. We work over the field of Laurent series in the parameter $\hslash$. Let us consider $r$ copies of the Heisenberg algebra, generated by $(J^a_k)_{k \in \mathbb{Z}}^{a \in [r]}$ with relations
\[
[J^a_k,J^b_l] = \hslash k \delta_{a,b}\delta_{k + l,0}.
\]
We introduce the $1$-form valued fields
\[
\mathcal{J}\big(\begin{smallmatrix} a \\ x \end{smallmatrix}\big) = \sum_{k \in \mathbb{Z}} \frac{J_k^a \dd x}{x^{k + 1}}.
\]
The principal $\mathcal{W}(\mathfrak{gl}_r)$-algebra at self-dual level is a vertex operator algebra freely and strongly generated by fields $\mathcal{W}^i(x)$ with $i \in [r]$ of conformal weight $i$. The Virasoro field is $\mathcal{W}^2(x)$. By convention these fields are forms of degree $i$:
\[
\mathcal{W}^i(x) = \sum_{k \in \mathbb{Z}} \frac{W_k^i (\dd x)^{i}}{x^{k + i}}.
\] 
The generating fields  $\mathcal{W}^i(x)$ can be realised in terms of the Heisenberg fields as elementary symmetric polynomials
\[
\mathcal{W}^i(x) = \sum_{1 \leq a_1 < \cdots < a_i \leq r} \prod_{j = 1}^{i} \mathcal{J}\big(\begin{smallmatrix} a_j \\ x\end{smallmatrix}\big),
\]
or equivalently  
\begin{equation}\label{eq:Wgens}
	\prod_{a = 1}^{r} \Big(u + \mathcal{J}\big(\begin{smallmatrix} a \\ x \end{smallmatrix}\big)\Big) = \sum_{i = 0}^{r} u^{r - i} \mathcal{W}^i(x),
\end{equation}
with the convention $\mathcal{W}^0(x) = 1$.

\medskip

\subsubsection{Gaiotto vector for \texorpdfstring{$\Lambda = O(\hslash^{1/2})$}{Lambda = O(h1/2)}}
\label{S122}

\medskip

We consider the Verma module for the Heisenberg and for the \texorpdfstring{$\mathcal{W}(\mathfrak{gl}_r)$}{W(glr)}-algebra
\[
\mathcal{V} = \mathbb{C}(\mathbf Q)[T]\big[\!\big[(J_{-k}^{a})_{k \in \mathbb{Z}_{> 0}}^{a \in [r]}\big]\!\big](\!(\hslash^{\frac{1}{2}})\!)\,,
\] where we let $J_0^a$ act by the scalar $Q_a$ for each $a \in [r]$, and $\mathbf Q$ denotes the set of variables $\{Q_1,\cdots, Q_r\}$.
This is a graded vector space with $\text{deg}(J_k^a) = \text{deg}(\hslash^{\frac{1}{2}}) = 1$. We denote $\mathcal{V}_{\geq 0}$ (resp. $\mathcal{V}_{> 0}$) the subspace generated by monomials of non-negative (resp. positive) degrees. We first consider vectors in $\mathcal{V}_{\geq 0}$ of the form
\begin{equation} 
	\label{G1form}
	\ket{\mathfrak{G}_T} = \exp\left(\sum_{\substack{(g,n) \in \frac{1}{2}\mathbb{Z}_{\geq 0} \times \mathbb{Z}_{> 0} \\ 2g - 2 + n > 0}} \frac{\hslash^{g - 1}}{n!} \sum_{\substack{a_1,\ldots,a_n \in [r] \\ k_1,\ldots,k_n \in \mathbb{Z}_{> 0}}} F_{g,n}\big[\begin{smallmatrix} a_1 & \cdots & a_n \\ k_1 & \cdots & k_n \end{smallmatrix}\big] \prod_{j = 1}^{n} \frac{J_{-k_j}^{a_j}}{k_j}\right) \in 1 + \mathcal{V}_{> 0},
\end{equation}
satisfying the Whittaker condition
\begin{equation}
	\label{G1Whit} \forall (i,k) \in [r] \times \mathbb{Z}_{> 0}\qquad W_k^i \ket{\mathfrak{G}_T} = \delta_{i,r}\delta_{k,1} \hslash^{\frac{r}{2}} T \ket{\mathfrak{G}_T}.
\end{equation}
Notice that we have set the energy scale to be $\Lambda^r = \hslash^{\frac{r}{2}} T$ here. 

\medskip

\subsubsection{Correlators and unramified topological recursion}
\label{S123}

\medskip

We consider the  curve $C = \bigsqcup_{a = 1}^{r} C^a$ where $C^a = \mathbb{P}^1$, equipped with the forgetful maps $x : C \rightarrow \mathbb{P}^1$ (which forgets the label $a$) and $\mathfrak{c} : C \rightarrow [r]$ (which only remembers the label $a$). We view points $z \in C$ as pairs $\big(\begin{smallmatrix} \mathfrak{c}(z) \\ x(z) \end{smallmatrix}\big)$. We denote the point $\big(\begin{smallmatrix} a \\ \infty \end{smallmatrix}\big)$ in $ C$ by  $\infty_a$. As the fibres of $x$ play a special role, we introduce the notation
\[
\mathfrak{f}(z) = x^{-1}(x(z)),\qquad \mathfrak{f}'(z) = \mathfrak{f}(z) \setminus \{z\}.
\]
If $z_1,\ldots,z_n$ is a $n$-tuple of points in $C$ and $J \subseteq [n]$, we denote $z_{J} = (z_j)_{j \in J}$. 

The coefficients of $\ket{\mathfrak{G}_T}$ can be repackaged in terms of a collection of generating series, indexed by $(g,n) \in \mathbb{Z}_{\geq 0} \times \mathbb{Z}_{> 0}$ called correlators. They are defined as
\begin{multline}
		\label{codef}
		w_{g,n}(z_1,\ldots,z_n)   = \sum_{k_1,\ldots,k_n \in \mathbb{Z}_{> 0}} F_{g,n}\big[\begin{smallmatrix} \mathfrak{c}(z_1) & \cdots & \mathfrak{c}(z_n) \\ k_1 & \cdots & k_n \end{smallmatrix}\big] \prod_{j = 1}^{n} \frac{\dd x(z_j)}{x(z_j)^{k_j + 1}} \\
		 \quad + \delta_{g,0}\delta_{n,1} Q_{\mathfrak{c}(z_1)} \frac{\dd x(z_1)}{x(z_1)} + \delta_{g,0}\delta_{n,2} \delta_{\mathfrak{c}(z_1),\mathfrak{c}(z_2)} \frac{\dd x(z_1)\dd x(z_2)}{(x(z_1) - x(z_2))^2}. 
\end{multline}

We will also need the expressions
\begin{equation}
	\label{wprime}
	\widehat{W}_{g,i;n}(z_{[i]};v_{[n]}) = \sum_{\substack{\mathbf{L} \vdash i \\  \mathbf{N} \vdash_{\mathbf{L}} [n] \\ g : \mathbf{L} \rightarrow \frac{1}{2}\mathbb{Z}_{\geq 0}}}^{\text{no}\,\,w_{0,1}} \delta_{g,i + \sum_{L} (g_L - 1)} \prod_{L \in \mathbf{L}} w_{g_L,\# L + \# N_L}(z_{L},v_{N_L}).
\end{equation}
The notation $\mathbf{L} \vdash [i]$ means that $\mathbf{L}$ is a set of pairwise disjoint non-empty subsets of $[i]$ whose union is $[i]$. The notation $\mathbf{N} \vdash_{\mathbf{L}} [n]$ means a map associating to each $L \in \mathbf{L}$ a (possibly empty) subset $N_L \subseteq [n]$, such that the $(N_L)_{L \in \mathbf{L}}$  are pairwise disjoint and their union is $[n]$. The logic behind these expressions becomes clear by writing them for low values of $i$:
\begin{equation*}
	\begin{split}
		\widehat{W}_{g,1;n}(z;v_{[n]}) & = w_{g,1+n}(z,v_{[n]}), \\
		\widehat{W}_{g,2;n}(z_1,z_2;v_{[n]}) & = w_{g-1,2+n}(z_1,z_2,v_{[n]}) + \sum_{\substack{g_1 + g_2 = g \\ J_1 \sqcup J_2 = [n]}}^{\text{no}\,\,w_{0,1}} w_{g_1,1+\#J_1}(z_1,v_{J_1}) w_{g_2,1+\#J_2}(z_2,v_{J_2}).
	\end{split}
\end{equation*}
The exclusion of $w_{0,1}$ factors from the sum in \eqref{wprime} has the effect that $\widehat{W}_{g,i;n}$ only involves $w_{h,m}$ with $2h - 2 + m < 2g - 2 + (1 + n)$.

\begin{thm} \cite[Theorem 5.10]{BBCC} Assume that $Q_1,\ldots,Q_r \in \mathbb{C}$ are pairwise distinct and $T \in \mathbb{C}$. There exists a unique $\ket{\mathfrak{G}_T}$ of the form \eqref{G1form} satisfying the Whittaker constraints \eqref{G1Whit}. The coefficients $F_{g,n}\big[\begin{smallmatrix} a_1 & \cdots & a_n \\ k_1 & \cdots & k_n \end{smallmatrix}\big]$ vanish if $(k_1 + \cdots + k_n)r > 2g$. In particular, for any $(g,n) \in \frac{1}{2} \mathbb{Z}_{\geq 0} \times \mathbb{Z}_{> 0}$ the correlators $w_{g,n}$ defined in \eqref{codef} are meromorphic $n$-differentials on $C$. Besides, they are computed by the unramified topological recursion for $2g - 2 + n > 0$ 
	\begin{multline*}
			w_{g,n}(z_1,\ldots,z_n)  = \sum_{a = 1}^{r} \Res_{z = \infty_a} \sum_{\{z\} \subseteq Z \subseteq \mathfrak{f}(z)}  \frac{-\int^{z}_{\infty_{a}} w_{0,2}(\cdot,z_1)}{\prod_{z' \in \mathfrak{f}(z) \setminus Z} \big(w_{0,1}(z') - w_{0,1}(z)\big)}\,\widehat{W}_{g,\#Z;n - 1}(Z;z_2,\ldots,z_n), \\
			\quad + \delta_{g,\frac{r}{2}}\delta_{n,1} \frac{T}{\prod_{b \neq \mathfrak{c}(z_1)} (Q_b - Q_{\mathfrak{c}(z_1)})} \,\frac{\dd x(z_1)}{x(z_1)^2},
	\end{multline*}
	where the factors in the denominator of the first line should be understood as
	\[
	w_{0,1}(z') - w_{0,1}(z) = (Q_{\mathfrak{c}(z')} - Q_{\mathfrak{c}(z)}) \frac{\dd x(z)}{x(z)},
	\]
	since $x(z) = x(z')$ for $z' \in \mathfrak{f}(z)$. 
\end{thm} 
This formula is indeed a recursion on $2g - 2 + n >0$.

\medskip

\subsubsection{Gaiotto vector for formal \texorpdfstring{$\Lambda$}{Lambda}}
\label{S124}

\medskip

As the energy scale in \eqref{G1Whit} has been set to $\Lambda^r = \hslash^{\frac{r}{2}}T$, and  $\hslash$ is a formal parameter near $0$, we are treating $\ket{\mathfrak{G}_T}$ as a formal expansion as $\Lambda  \to 0$. In this article, we would like to understand the analytic properties of the Gaiotto vector as a function of the energy scale $\Lambda$. The first step is to consider $\Lambda$ directly as a formal parameter independently of $\hslash$. From the homogeneity of the Whittaker constraints \eqref{G1Whit} it can be inferred \cite[Lemma 4.5]{BBCC} that the coefficients $F_{g,n}\big[\begin{smallmatrix} a_1 & \cdots & a_n \\ k_1 & \cdots & k_n \end{smallmatrix}\big]$ are proportional to $T^{k_1 + \cdots + k_n}$. As they vanish for $(k_1 + \cdots + k_n)r > 2g$, we can define
\begin{equation}
	\label{homogen}
	\Phi_{g,n}\big[\begin{smallmatrix} a_1 & \cdots & a_n \\ k_1 & \cdots & k_n \end{smallmatrix}\big] = \Lambda^{r(k_1 + \cdots + k_n)}\,F_{g + (k_1 + \cdots + k_n)\frac{r}{2},n}\big[\begin{smallmatrix} a_1 & \cdots & a_n \\ k_1 & \cdots & k_n \end{smallmatrix}\big]\big|_{T = 1}.
\end{equation}
As the only possible source of half-integer powers of $\hslash$ in $\ket{\mathfrak{G}_T}$ was the prefactor $\hslash^{\frac{r}{2}}$ of $T$ in equation \eqref{G1Whit}, the $\Phi_{g,n}$ vanish for non-integer $g$. With these coefficients we can introduce the vector 
\begin{equation}
	\label{Gformalexp}
	\ket{\Gamma_{\Lambda}} = \exp\left(\sum_{(g,n) \in \mathbb{Z}_{\geq 0} \times \mathbb{Z}_{> 0}} \frac{\hslash^{g - 1}}{n!} \sum_{\substack{a_1,\ldots,a_n \in [r] \\ k_1,\ldots,k_n \in \mathbb{Z}_{> 0}}} \Phi_{g,n}\big[\begin{smallmatrix} a_1 & \cdots & a_n \\ k_1 & \cdots & k_n \end{smallmatrix}\big] \prod_{j = 1}^{n} \frac{J_{-k_j}^{a_j}}{k_j}\right)
\end{equation}
belonging to the Verma module
\begin{equation}
	\label{VLambda}
	\mathcal{V}_{\Lambda} =  \mathbb{C}(\mathbf Q) \big[\kern-2.2 pt\big[(J_{-k}^{a})_{k \in \mathbb{Z}_{> 0}}^{a \in [r]}\big]\kern-2.2 pt \big][\![\Lambda^r]\!](\!(\hslash)\!).
\end{equation}
Note that we do not require $\ket{\Gamma_{\Lambda}} $ to live in the non-negative degree part of $\mathcal{V}_{\Lambda} $. This vector satisfies the Whittaker condition
\begin{equation}
	\label{WhitG2}
	\forall (i,k) \in [r] \times \mathbb{Z}_{> 0} \qquad W_k^i \ket{\Gamma_{\Lambda}} = \delta_{i,r}\delta_{k,1} \Lambda^r \ket{\Gamma_{\Lambda}}.
\end{equation}
Two important differences with the vector $ \ket{\mathfrak{G}_T}$ from equation \eqref{G1form} are however that
\begin{itemize} 
	\item $(g,n) = (0,1)$ and $(0,2)$ do contribute to the sum, i.e., we have non-zero $\Phi_{0,1}$ and $\Phi_{0,2}$;
	\item $\Phi_{g,n}\big[\begin{smallmatrix} a_1 & \cdots & a_n \\ k_1 & \cdots & k_n \end{smallmatrix}\big]$ can be non-zero for infinitely many indices $k_1,\ldots,k_n \in \mathbb{Z}_{> 0}$.
\end{itemize} The second condition forces us to treat $\Lambda$ as a formal parameter. Despite these differences, we can still use equation \eqref{codef} to introduce the correlators
\begin{multline}
		\label{phidef} 
		\phi_{g,n}(z_1,\ldots,z_n)  = \sum_{k_1,\ldots,k_n \in \mathbb{Z}_{> 0}} \Phi_{g,n}\big[\begin{smallmatrix} \mathfrak{c}(z_1) & \cdots & \mathfrak{c}(z_n) \\ k_1 & \cdots & k_n \end{smallmatrix}\big] \prod_{j = 1}^{n} \frac{\dd x(z_j)}{x(z_j)^{k_j + 1}} \\
 	+ \delta_{g,0}\delta_{n,1} Q_{\mathfrak{c}(z_1)} \frac{\dd x(z_1)}{x(z_1)} + \delta_{g,0}\delta_{n,2} \delta_{\mathfrak{c}(z_1),\mathfrak{c}(z_2)} \frac{\dd x(z_1)\dd x(z_2)}{(x(z_1) - x(z_2))^2}.
\end{multline}
They are now defined as germs of meromorphic $n$-differentials in the $n$-th product of the formal neighbourhood of $L := \bigsqcup_{a = 1}^{r} \{\infty_a\} \subset C$. More precisely, the $\phi_{g,n}$ for $2g - 2 + n > 0$ are germs of holomorphic $n$-differentials, $\phi_{0,1}$ is the germ of a meromorphic differential having a simple pole with residue $Q_a$ at $\infty_a$, and $\phi_{0,2}$ is the germ of a meromorphic bi-differential with a double pole on the diagonal.

\vspace{0.5cm}

\subsection{Main results at finite \texorpdfstring{$\Lambda$}{Lambda}}
\label{S13}

\medskip

Our main result is that the correlators \eqref{phidef} upgrade to meromorphic multi-differentials on a ramified, genus $0$ spectral curve, which we will refer to as the half Seiberg--Witten curve. Moreover, these multi-differentials  depend analytically on $\Lambda \in \mathbb{C}^*$, and are computed by the usual Chekhov--Eynard--Orantin topological recursion. 

\medskip

\subsubsection{Ramified topological recursion}

\medskip

We give a lightning introduction to the topological recursion in the form considered by Chekhov, Eynard and Orantin --- for more details, see \cite{EORev}. In the rest of the text, topological recursion without further precision will always mean this version, and it should be distinguished from the unramified topological recursion of \cref{S123}.

The initial data is called a \textit{spectral curve}, consisting of a quadruple $(S, x, \omega_{0,1}, \omega_{0,2})$, where $S$ is a Riemann surface, $x$ is a meromorphic function on $S$ that defines a branched covering $x : S \to \mathbb P^1$, $\omega_{0,1}$ is a meromorphic differential on $S$ and $\omega_{0,2}$ is a fundamental bi-differential, i.e. a symmetric meromorphic bi-differential on $S^2$ with a double pole having bi-residue $1$ on the diagonal, and no other poles.

 We  define $ \operatorname{Ram}(S)$ to be the set of all the ramification points of the branched covering $x : S \to \mathbb P^1$ except those that are also poles of $\omega_{0,1} $. We exclude the ramification points that are also poles of $ \omega_{0,1}$ from $ \operatorname{Ram}(S)$  as they do not contribute in the topological recursion formula. We further assume that $ \operatorname{Ram}(S)$ only contains simple ramification points. Then, near any ramification point $\rho \in \operatorname{Ram}(S)$, we have the local involution which exchanges the two sheets, and we denote this by $\sigma_{\rho}$.

Given a spectral curve $(S, x, \omega_{0,1}, \omega_{0,2})$  the \textit{topological recursion}  constructs $n$-differentials on $S$ called $\omega_{g,1+n}$ for any $2g-2 + (1+n) > 0$ by the following formula 
\begin{multline}
	\label{TRome}
	\omega_{g,1+n}(\zeta_0,\ldots,\zeta_n)  = \sum_{\rho \in \operatorname{Ram}( S)} \Res_{\zeta = \rho} \frac{\frac{1}{2} \int_{\sigma_{\rho}(\zeta)}^{\zeta} \omega_{0,2}(\cdot,\zeta_0)}{\omega_{0,1}(\zeta) - \omega_{0,1}(\sigma_{\rho}(\zeta))} \bigg( \omega_{g-1,2+n}(\zeta,\sigma_{\rho}(\zeta),\zeta_{[n]})   \\ 
	\quad + \sum_{\substack{g_1 + g_2 = g \\ J_1 \sqcup J_2 = [n]}}^{\text{no}\,\,\omega_{0,1}} \omega_{g_1,1+\# J_1}(\zeta,\zeta_{J_1}) \omega_{g_2,1+\#J_2}(\sigma_{\rho}(\zeta),\zeta_{J_2})\bigg).
\end{multline}
By construction, in each variable $\zeta_i$, the  $\omega_{g,n}$ only have poles at the ramification points. 

The spectral curves considered in this article will always have genus $0$, i.e. $ S \simeq \mathbb{P}^1$ with $\zeta$ a global coordinate. Then, there is a unique fundamental bi-differential, namely
\begin{equation}
	\label{02std}
	\omega_{0,2}(\zeta_1,\zeta_2) = \frac{\dd \zeta_1 \dd \zeta_2}{(\zeta_1 - \zeta_2)^2}.
\end{equation}
The latter is invariant under changes of global coordinates on $\mathbb{P}^1$, i.e. action of $\text{PSL}_2(\mathbb{C})$ by M\"obius transformations. In this context, the data of two functions $x(\zeta),y(\zeta)$ fully specifies a spectral curve, by taking $\omega_{0,1}(\zeta) = y(\zeta)\dd x(\zeta)$ and $\omega_{0,2}$ equal to \eqref{02std}.

\medskip

\subsubsection{Gaiotto vector and topological recursion}

\medskip

Consider the analytic family\footnote{We consider $Q_1,\ldots,Q_r \in \mathbb{C}$ to be fixed pairwise distinct, but we could equally well formulate the results by letting them vary, using instead of $S$ the larger family $\widehat{S} \rightarrow \{(\Lambda,\mathbf{Q}) \in \mathbb{C}^* \times \mathbb{C}^r\,\,|\,\,\prod_{b \neq a} (Q_b - Q_a) \neq 0\big\}$.} of curves $\pi : S \rightarrow \mathbb{C}^*_\Lambda$ defined by the vanishing locus in $\mathbb P^1_x \times \mathbb P^1_y \times \mathbb C^*_\Lambda$ of
\begin{equation}\label{eq:Gaiottocurve}
	\prod_{a = 1}^{r} \bigg(y - \frac{Q_a}{x}\bigg) + \frac{(-\Lambda)^{r}}{x^{r + 1}} = 0.
\end{equation}
Here $\Lambda$ is the parameter of the base of the family. The fibre $S_{\Lambda} $ over a fixed  $\Lambda \in \mathbb{C}^*$ is a smooth genus $0$ curve, which can be uniformised by $\zeta \in \mathbb{P}^1$:
\[
\left\{\begin{array}{lll} x(\zeta) & =  - \dfrac{\Lambda^r}{\prod_{a = 1}^{r} (Q_a - \zeta)} \\[12pt] y(\zeta) & =  \dfrac{\zeta}{x(\zeta)} = -\dfrac{\zeta}{\Lambda^r} \prod_{a = 1}^{r} (Q_a - \zeta) \end{array}\right.
\]  We will refer to the fibre $S_\Lambda$ as the \textit{half Seiberg--Witten curve}.
The map $x : S_{\Lambda} \rightarrow \mathbb{P}^1$ defines a branched cover of degree $r$. Since $Q_1,\ldots,Q_r$ are pairwise distinct, $x$ has $(r - 1)$ simple ramification points and a ramification point of  index $r$ at $x = 0$, and $x = \infty$ is not a branch point. As defined earlier, $ \operatorname{Ram}(S_\Lambda)$ is the set of all ramification points except the one at $x = 0$. The formal neighbourhood $L \subset  C$ mentioned in Section~\ref{S124} is canonically identified with the formal neighbourhood of $x^{-1}(\infty) \subset S_{\Lambda}$ by considering $1/x(\zeta)$ as a local coordinate near the latter. 

\begin{rem}\label{rem:SW}
	We note that the genus $0$ half Seiberg--Witten curve $S_\Lambda$ is neither the usual Seiberg--Witten curve of pure gauge theory (which is of genus $r-1$) nor the genus $0$ UV curve one that encounters in the  class $S $ construction of gauge theories. The Seiberg--Witten curve of pure $U_r$ gauge theory in our notation is  defined by the equation
	\[
	\prod_{a=1}^r \left(y-\frac{Q_a}{x} \right) + \frac{(-\Lambda)^r}{x^{r+1}} + \frac{(-\Lambda)^r}{x^{r-1}}  = 0.
	\] From this equation, we see that the half Seiberg--Witten curve can be realised as a degenerate limit of the Seiberg--Witten curve. 
	
	The terminology half Seiberg--Witten curve can be motivated as follows. Following the philosophy of \cite{AGT}, one can view the theory associated to the Gaiotto state to be ``half of" pure gauge theory: we are looking at the theory on a genus $0$ Riemann surface with one irregular singularity (at $\infty$ say) instead of having two irregular singularities (both at $0$ and $\infty$). The half Seiberg--Witten curve appears naturally in \cite{DHS09} when trying to understand the instanton partition function in a free fermion formalism. 
\end{rem}

Let  $K_{\pi}$ be the sheaf of holomorphic differentials relative to $\pi : S \rightarrow \mathbb{C}^*_\Lambda$. Its local sections are locally defined holomorphic differentials on $S_{\Lambda}$ varying analytically in $\Lambda \in \mathbb{C}^*$. If $D \subset S$ is a divisor transverse to the fibres, then $K_{\pi}(D)$ is the sheaf of meromorphic differentials relative to $\pi$ with poles on $D$. Concretely, its global sections are meromorphic differentials on $S_{\Lambda}$ with the location and maximal order of poles specified by $D$. For instance,  $y \dd x$ defines an element of $H^0(K_{\pi}(x^{-1}(\infty)),S)$: it is indeed a meromorphic differential on $S_{\Lambda}$ with simple poles at the $r$ poles of $x$ (with residues $-Q_1,\ldots,-Q_r$) and varying analytically with $\Lambda$. If we want to allow poles on $D$ of arbitrary order, we use
\[
K_{\pi}(*D) = \lim_{d \rightarrow \infty} K_{\pi}(dD) .
\]
We are particularly interested in this sheaf when $D$ is the ramification divisor.
\[
\operatorname{Ram}(S) = \bigsqcup_{\Lambda \in \mathbb{C}^*} \operatorname{Ram}(S_{\Lambda}).
\]
Let $\pi_n : S^{[n]} \rightarrow \mathbb{C}^*_\Lambda$ be the fibre product of $n$ copies of $S$ over the base of the family, i.e.,
\[
S^{[n]} = \big\{(s_1,\ldots,s_n) \in S^n \quad \big|\quad \pi(s_1) = \cdots = \pi(s_n)\big\},
\]
with $\pi_n$ being the obvious projection to the common value $\Lambda = \pi(s_1) = \cdots = \pi(s_n)$. We define $\Delta \subset S^{[2]}$ to be the diagonal. If $\text{pr}_m : S^{[n]} \rightarrow S$ is the projection on the $m$-th factor and $\mathcal{F}$ is a sheaf on $S$, we use the notation $\mathcal{F}^{\boxtimes n} := \text{pr}_1^*(\mathcal{F}) \otimes \cdots \text{pr}_n^*(\mathcal{F})$ for its $n$-variable version.

Our central result is the following.
\begin{thm}
	\label{th:main} Assume that $Q_1,\ldots,Q_r \in \mathbb{C}$ are pairwise distinct. For any $(g,n) \in \mathbb{Z}_{\geq 0} \times \mathbb{Z}_{> 0}$, there exists $\omega_{g,n}$ which is an element of
	\begin{itemize}
		\item $H^0\big(K_{\pi}(x^{-1}(\infty),S\big)$ if $(g,n) = (0,1)$;
		\item $H^0\big(K_{\pi}^{\boxtimes 2}(2\Delta),S^{[2]}\big)$ if $(g,n) = (0,2)$;
		\item $H^0\big(K_{\pi}(*\operatorname{Ram}(S))^{\boxtimes n},S^{[n]}\big)$ if $2g - 2 + n > 0$;
	\end{itemize}
	such that $\phi_{g,n}$ in \eqref{phidef} is the all-order series expansion of $\omega_{g,n}$ as $\Lambda \rightarrow 0$ and $z_1,\ldots,z_n \rightarrow x^{-1}(\infty) \cong L$ using $1/x(\zeta_j)$ as local coordinate. Besides, for any fixed $\Lambda \in \mathbb{C}^*$ we have
	\[
	\omega_{0,1}(\zeta) = y(\zeta) \dd x(\zeta),\qquad \omega_{0,2}(\zeta_1,\zeta_2) = \frac{\dd \zeta_1 \dd \zeta_2}{(\zeta_1 - \zeta_2)^2},
	\]
	and for $2g - 2 + n > 0$, the $\omega_{g,n} $ are constructed by topological recursion  \eqref{TRome} on the half Seiberg--Witten spectral curve \eqref{eq:Gaiottocurve}.
\end{thm}

\begin{rem}
	\cref{th:main} gives a gauge theory interpretation to the $\omega_{g,n}$ obtained by  topological recursion on the half Seiberg--Witten spectral curve, i.e., they  essentially give the genus expansion of the Gaiotto vector $\ket{\Gamma_{\Lambda}}$. These correlators are different from the topological recursion correlators of the  Seiberg--Witten curve (see \cref{rem:SW}), whose gauge-theoretic interpretation is unknown. On the other hand, the generating function of the free energies $F_g$ constructed by topological recursion on the Seiberg--Witten curve is expected to coincide with the (instanton part of the) Nekrasov instanton partition function $Z_{\text{Nek}}$ \cite{HK10} (see also \cref{sec:gauge}).
\end{rem}

The proof is given in Section~\ref{S3}. The strategy is to start from the Whittaker constraints \eqref{WhitG2} for the Gaiotto vector $\ket{\Gamma_{\Lambda}}$ that we translate into certain constraints on the $\phi_{g,n}$. We then show by induction on $2g - 2 + n \geq -1$ that these constraints have a unique solution, and thus $\phi_{g,n} $ can be upgraded to a meromorphic $n$-differential on the half Seiberg--Witten  curve $S_{\Lambda}$, and we locate all possible poles. Then, we show that the Whittaker constraints after analytic continuation away from $\infty$ imply the linear and quadratic loop equations of \cite{BEO,BSblob}. As $S_{\Lambda}$ has genus $0$, the loop equations have a unique solution given by the topological recursion \eqref{TRome}.

The Whittaker constraints \eqref{WhitG2} where $\Lambda^r = O(\hslash) $ or $\Lambda$ is a formal parameter yields (shifted) Airy structures  to which the Kontsevich--Soibelman theorem can be applied to establish existence and uniqueness of the solution and its reconstruction by topological recursion \cite{KSTR,ABCD,BBCCN18}. In contrast, the formalism of Airy structures cannot be applied to understand whether this solution can be upgraded to an analytic function of $\Lambda$. In other words, if we consider the constraints \eqref{WhitG2} with $\Lambda \in \mathbb C^*$, the solution cannot be constructed using Airy structures. The proof of \cref{th:main} shows how the analytic behaviour in such degree $0$ terms can be understood, and this involves analytic continuation of the correlators away from the formal neighbourhood where they were initially defined (see \cref{S2ana} for more details).

\vspace{0.5cm}

\subsection{Generalisation: the CDO vector}

\medskip

The method we develop to prove \cref{th:main} is flexible enough to be applicable to other shifted Airy structures (i.e., ones with degree $0$ terms). As a demonstration of this principle, in \cref{S4} we obtain an analogue of \cref{th:main} for the Whittaker vector constructed in \cite{CDO24} by Chidambaram, Do{\l}{\k{e}}ga and Osuga that encodes $b$-Hurwitz numbers.

We restrict to the $b = 0$ case, which corresponds to the self-dual level for the $\mathcal{W}$-algebra. Consider two sets of parameters $\mathbf P = \{P_1,\ldots, P_r\}$ and $\mathbf Q  = \{Q_1,\ldots, Q_{r-1}\}$, and a certain representation $\widetilde{\mathcal{V}}_{\Lambda} = \mathbb{C}(\!(\mathbf P, \mathbf Q, \Lambda^r)\!)\big[\!\big[(J_{-k}^{a})_{k \in \mathbb{Z}_{> 0}}^{a \in [r]}\big]\!\big](\!(\hslash)\!)$ of $\mathcal W(\mathfrak{gl}_r)$, that is defined using the assignment \eqref{eq:Jrep}. Then, for any pairwise disjoint $Q_1,\ldots, Q_{r-1} \in \mathbb C$, \cite{CDO24} constructed a Whittaker vector $\ket{\Gamma_{\Lambda}^{\text{CDO}}} \in \widetilde{\mathcal{V}}_{\Lambda}$  of the form
\begin{equation*}
	\ket{\Gamma_{\Lambda}^{\text{CDO}}} =\exp\left( \sum_{(g,n) \in \mathbb{Z}_{\geq 0} \times \mathbb{Z}_{> 0}}   \frac{\hslash^{g-1}}{n!} \sum_{\substack{a_1,\ldots,a_n \in [r] \\ k_1,\ldots,k_n \in \mathbb{Z}_{> 0}}} \Phi_{g,n}\left[\begin{smallmatrix} a_1 & \cdots & a_n \\ k_1 & \cdots & k_n \end{smallmatrix}\right] \prod_{j = 1}^{n} \frac{J^{a_j}_{-k_j}}{k_j} \right),  
\end{equation*}
 satisfying the constraints
\[
\forall (i,k) \in [r] \times \mathbb{Z}_{\geq 0} \qquad \widetilde{W}^i_k \ket{\Gamma_{\Lambda}^{\text{CDO}}}= (-1)^i e_i(P_1,\cdots, P_r) \delta_{k,0} \ket{\Gamma_{\Lambda}^{\text{CDO}}},
\] where $ \widetilde{W}^i_k$ denote the modes of $\mathcal{W}(\mathfrak{gl}_r)$ in the representation $\widetilde{\mathcal{V}}_{\Lambda}$ and $e_i$ denotes the $i$-th elementary symmetric polynomial in the entries.

To state our result we first describe the spectral curve. Consider the analytic family of curves $\pi : \mathcal{S} \rightarrow \mathbb{C}^*_{\Lambda}$ cut out in $\mathbb{P}^1_x \times \mathbb{P}^1_y \times \mathbb{C}_{\Lambda}^*$ by
\[
\prod_{a = 1}^{r} \left(\frac{P_a}{x} + y\right) + \frac{1}{\Lambda^r} \prod_{a = 1}^{r - 1} \left(\frac{Q_a}{x} - y\right) = 0.
\]
The fibre over a fixed $\Lambda \in \mathbb{C}^*$ is a smooth genus $0$ curve, which can be uniformised by $\zeta \in \mathbb{P}^1$:
\begin{equation}
\label{eq:CDOcurve}\left\{\begin{array}{lll} x(\zeta) = -\Lambda^r \dfrac{\prod_{a = 1}^{r} (P_a + \zeta)}{\prod_{a = 1}^{r - 1}(Q_a - \zeta)}, \\[12pt] y(\zeta) = \dfrac{\zeta}{x(\zeta)} = - \dfrac{\zeta}{\Lambda^r} \dfrac{\prod_{a = 1}^{r - 1} (Q_a - \zeta)}{\prod_{a = 1}^r (P_a + \zeta)} .\end{array}\right.
\end{equation}
\begin{rem}
	The family of curves $ \mathcal S$ defined by \eqref{eq:CDOcurve} also admits an interpretation as a half Seiberg--Witten curve for a different gauge theory. The Seiberg--Witten family of curves for $U_{r-1}$ gauge theory with  $ r $ fundamental hypermultiplets with mass parameters $P_1,\cdots,P_r$  and energy scale $\Lambda^{-r}$ is known to be (see \cite[Section 11.6]{Tac13} for the $SU_{r-1}$ version)
	\[
	(-1)^r \frac{1}{\Lambda^{r^2}x^{r-2}} +   \prod_{a=1}^{r} \left(\frac{P_a}{x} + y\right) + \frac{1}{\Lambda^r } \prod_{a=1}^{r-1} \left(\frac{Q_a}{x} - y\right) = 0.
	\] From this, we clearly see that \eqref{eq:CDOcurve} can be obtained as a degenerate limit.
\end{rem}

We have the following topological recursion result for the correlators defined from $\ket{\Gamma^{\text{CDO}}_\Lambda}$.
\begin{thm}
\label{th:main2} Assume that $P_1,\ldots,P_{r},Q_1,\ldots,Q_{r -1}$ are generic (more precisely, that they belong to the set introduced in \cref{Rset}). For any $(g,n) \in \mathbb{Z}_{\geq 0} \times \mathbb{Z}_{> 0}$, there exists $\omega_{g,n}$ which is an element of
\begin{itemize}
\item $H^0\big(K_{\pi}(x^{-1}(\infty),\mathcal{S}\big)$ if $(g,n) = (0,1)$;
\item $H^0\big(K_{\pi}^{\boxtimes 2}(2\Delta),\mathcal{S}^{[2]}\big)$ if $(g,n) = (0,2)$;
\item $H^0\big(K_{\pi}(*\operatorname{Ram}(\mathcal{S}))^{\boxtimes n},\mathcal{S}^{[n]}\big)$ if $2g - 2 + n > 0$;
\end{itemize}
such that $\phi_{g,n}$ (defined in \eqref{eq:phidef2} from $\ket{\Gamma^{\text{CDO}}_\Lambda}$ analogously to \eqref{phidef})  is the all-order series expansion of $\omega_{g,n}$ as $\Lambda \rightarrow 0$ and $z_1,\ldots,z_n \rightarrow x^{-1}(\infty) \cong L$ using $1/x(\zeta_j)$ as local coordinate. Besides, for any fixed $\Lambda \in \mathbb{C}^*$ we have
\[
\omega_{0,1}(\zeta) = y(\zeta) \dd x(\zeta),\qquad \omega_{0,2}(\zeta_1,\zeta_2) = \frac{\dd \zeta_1 \dd \zeta_2}{(\zeta_1 - \zeta_2)^2},
\]
and for $2g - 2 + n > 0$, the $\omega_{g,n} $ are constructed by the topological recursion  \eqref{TRome} on the CDO spectral curve \eqref{eq:CDOcurve}. 
\end{thm}

\vspace{0.5cm}

\subsection{Comments}

\medskip

There are several ways to look at this result in a broader context, which make the technique of the proof of potential interest beyond the example of the Gaiotto or CDO vectors.

\medskip
 
\subsubsection{Role of analytic continuation}
 
\medskip

\label{S2ana}

It may be useful to insist on the role of analytic continuation, to explain in which sense the topological recursion statements of this article are non-trivial and interesting.

On the one hand, partition functions of shifted Airy structures (i.e., with degree $0$ terms), like Whittaker vectors, are of the form \eqref{Gformalexp}. From such partition functions, one can define a system of correlators as in \eqref{phidef}, that are meromorphic multi-differentials on a (local) curve which is a finite collection of formal discs. We may ask whether these correlators are the germs of meromorphic multi-differentials defined on a (global) connected curve containing those formal discs. If this is the case, the next question is to determine the location of the poles and the behaviour at these poles. This is hopefully a step towards the reconstruction of the multi-differentials on the global curve, from which the original partition function can be retrieved by series expansion on the formal discs. Quite often, reconstruction is possible by the Chekhov--Eynard--Orantin topological recursion (or some variants of it) on a ramified spectral curve, and for this we can rely on the theory of abstract loop equations \cite{BSblob}. These are the steps we follow to prove \cref{th:main} and \cref{th:main2}.

On the other hand, correlators generated by the topological recursion $(\omega_{g,n})_{g,n}$ are systems of meromorphic multi-differentials on the spectral curve $(S,x,y,\omega_{0,2})$, that have poles at the divisor of ramification points $\text{Ram}$ of $x$. If we choose a finite set of points $R \subset S$ and local coordinates $\eta$ near these points, one can associate a partition function $Z_{R,\eta}$ of the form \eqref{Gformalexp}  that faithfully encode the germ of $\omega_{g,n}$ near $R$ using the chosen local coordinates. If we choose $R = \text{Ram} $, the $\omega_{g,n}$ are characterised by the property that $Z_{\text{Ram},\eta}$ is the partition function of an Airy structure \cite{KSTR,ABCD,BBCCN18} (for any choice of local coordinate). For instance, if all the ramification points are simple, the topological recursion is quadratic (see \eqref{TRome}) and $Z_{\text{Ram},\eta}$ is the  solution of several copies --- one for each ramification point --- of Virasoro constraints. Then, we say that topological recursion solves Virasoro constraints \emph{at} the ramification points. 

An interesting feature is that, if the spectral curve is connected, the uniqueness of analytic continuations implies that $Z_{\{p\},\eta}$ near some point $p \in S$ fully determines $\omega_{g,n}$ and thus the $Z_{R,\eta}$ for any other choices of $(R,\eta)$. It is often the case that the enumerative information in $\omega_{g,n}$ is stored away from ramification points, e.g. in $Z_{R,\eta}$ where $R$ is a subset of the poles of $x$. A natural question is then to find constraints directly on $Z_{R,\eta}$ that are implied by topological recursion for the $\omega_{g,n}$. If such constraints exist, we say that they live \emph{at} $R$. Quite often, they are again $\mathcal{W}$-constraints, but possibly for a different $\mathcal{W}$-algebra or a different representation of it. This means that $Z_{R,\eta}$ generates a $\mathcal{W}$-algebra module, which we can consider as being ``obtained by analytic continuation'' from the $\mathcal{W}$-algebra module generated by $Z_{\text{Ram},\eta}$. 

It is quite interesting to find the constraints at $R$ where the enumerative information is stored, precisely because they bear directly on the enumerative information. Yet, deriving them from topological recursion can turn out to be a non-trivial task: there is no  general theory to do so and it strongly depends on the global geometry of the spectral curve. The procedure of globalisation described in \cite{TRlimits} for algebraic singularities goes in this direction, and \cite[Section 1.4]{TRlimits} describes some of the new $\mathcal{W}$-algebras and modules that can be obtained by analytic continuation. Another example is the case of Virasoro constraints for the Gromov--Witten theory of $\mathbb{P}^1$ found in \cite{GivVir,OPVir}. \cite{BNP1} showed how these Virasoro constraints for Gromov--Witten theory $\mathbb{P}^1$ living at $\infty$ arise from topological recursion on a spectral curve with logarithmic singularities, and thus satisfy a different set of Virasoro constraints at the ramification points.

In this language, \cref{th:main} (resp. \cref{th:main2}) shows that the $\mathcal{W}(\mathfrak{gl}_r)$-constraints at $\infty$ for the Gaiotto (resp. CDO) vector imply $(r - 1)$ copies of Virasoro constraints at the ramification points, and the solution of the latter is known to be given by the topological recursion. The converse process of deriving back the $\mathcal{W}(\mathfrak{gl}_r)$-constraints at $\infty$ from the topological recursion is a priori not easy --- in particular, the limit $\Lambda \rightarrow 0$ would not be covered by the results of \cite{TRlimits}. But, it can be done by following the proof of \cref{S3} (resp. \cref{S4}) backwards.

\medskip

\subsubsection{Comparison to other models}

\medskip

Our strategy of analytic continuation exhibits similarities and differences with strategies previously employed  in the study of matrix models and Hurwitz theory.

In the $1$-hermitian matrix model, the starting point (replacing the Whittaker constraints) is the Dyson--Schwinger equations (Virasoro constraints at $\infty$). From there, analytic continuation on the spectral curve and then abstract loop equations from which topological recursion follows, has been established in \cite{ACM,E1MM,BEO,BSblob}: this is how  topological recursion was invented. But this case is easier to analyse as  these Schwinger--Dyson equations are only quadratic (instead of degree $r$). There exist matrix models satisfying $\mathcal{W}(\mathfrak{gl}_r)$-constraints for $r > 2$, but we are not aware of topological recursion being derived (directly) for them when $r \geq 4$.

In classical Hurwitz theory, several authors took the route of first proving analytic continuation on the spectral curve starting from the cut-and-join equation (constraints living at $\infty$, not necessarily quadratic), from which they deduced the abstract loop equations and concluded that topological recursion holds. This strategy was employed in \cite{EMS,BHSLM,BKLPS,DKPS,DoKarev,BDKLM20}. However, cut-and-join equations are a priori quite different from $\mathcal{W}$-constraints and do not give  (shifted) Airy structures per se. The relation between cut-and-join equations and  $\mathcal{W}$-constraints has been better understood in the recent work of \cite{CDO24}, and provides a different proof of topological recursion for weighted Hurwitz numbers, which is in line with the philosophy of this paper.

\vspace{1cm}

\section{The Gaiotto vector: proof of \texorpdfstring{\cref{th:main}}{Theorem 2.2}}
\label{S3}
\medskip

In this section we show that the correlators $\phi_{g,n}$, defined in \eqref{phidef} from the Gaiotto vector $\ket{\Gamma_{\Lambda}}$ at self-dual level $\kappa = 1$, can be analytically continued to meromorphic differential forms $\omega_{g,n}$ on the half Seiberg--Witten  curve $S_\Lambda$ and furthermore the latter satisfy the Chekhov--Eynard--Orantin topological recursion. 

\vspace{0.5cm}

\subsection{\texorpdfstring{$\mathcal W$}{W}-constraints}

\medskip

Recall that we have introduced the correlators $\phi_{g,n}$ from the genus expansion of the Whittaker vector $\ket{\Gamma_{\Lambda}}$ in \eqref{phidef}. These correlators $\phi_{g,n}$ are germs of meromorphic $n$-differentials in the $n$-th product of the formal neighbourhood of $L := \bigsqcup_{a = 1}^{r} \{\infty_a\} \subset C $, where $C$ is the unramified curve of degree $r$ defined in \cref{S123}. 

The strategy of the proof consists is to use the $\mathcal W$-constraints \eqref{WhitG2} to fix the correlators $\phi_{g,n}$  uniquely. In addition, the $\mathcal W$-constraints will give a formula for the correlators $\phi_{g,n}$ implying that they analytically continue to meromorphic differentials on the half Seiberg--Witten  curve $S_\Lambda$. Finally, by showing that these analytic continuations satisfy the abstract loop equations, we use the results of \cite{BEO,BSblob} to prove that they coincide with the topological recursion correlators $\omega_{g,n}$.

In order to understand the implication of the $\mathcal W$-constraints \eqref{WhitG2} for the correlators $\phi_{g,n}$ we  need to introduce some  operators and notation. First, let us define the operator $\operatorname{ad}_{g,n} $ following \cite[Section 5.1.2]{BBCC} which transforms  formal series into differentials on the curve $C$.
\begin{defn} Consider a formal series $f \in \mathbb C\big[\!\big[(J^a_{-k})^{a\in[r]}_{k \in \mathbb Z_{> 0}}\big]\!\big](\!(\hslash )\!)$. We define 
\begin{equation*}
	\operatorname{ad}_{g,n}(f) = [\hslash^g] \sum_{\substack{a_1,\ldots,a_n \in [r] \\ k_1, \ldots, k_n \in \mathbb{Z}_{>0}}} \left(k_1\frac{ \partial}{\partial J^{a_1}_{-k_1}} \, \cdots \, k_n \frac{\partial}{\partial J^{a_n}_{-k_n}} f \right)_{J^{a_j}_{-k_j} = 0} \frac{  \delta_{\mathfrak{c}(w_1),a_1}  \dd w_1}{w_1^{k_1+1}}  \cdots\,\,  \frac{ \delta_{\mathfrak{c}(w_n),a_n}  \dd w_n}{w_n^{k_n+1}}.
\end{equation*}
\end{defn} The operator $\operatorname{ad}_{g,n}$ picks  the terms of order $\hslash^g$ that are homogeneous of degree $n$ in the variables $J^{a_j}_{-k_j}$, and replaces these variables by the corresponding $1$-forms $ \delta_{\mathfrak{c}(w_j),a_j} \frac{k_j \dd w_j} {w_j^{k_j+1}}$ on $C$. Then, we define the following combinations of the correlators.
\begin{defn}\label{def:Omega}
	For any $g,n,i \in \mathbb Z_{\geq 0}$,  assuming that $z_j = \left(\begin{smallmatrix} a_{j} \\ x \end{smallmatrix}\right)$ for $j \in [i]$, define 
	\begin{equation}
			\Omega_{g,i;n}(z_{[i]};w_{[n]})  := \operatorname{ad}_{g,n}\left( \ket{\Gamma_{\Lambda}}^{-1} \bigg( \prod_{j=1}^i \mathcal{J}\big(\begin{smallmatrix} a_{j} \\ x \end{smallmatrix}\big) \bigg) \ket{\Gamma_{\Lambda}}  \right).
	\end{equation}
	Given the form \eqref{Gformalexp} of the vector $\ket{\Gamma_{\Lambda}}$ in the completed polynomial algebra in the negative $J$s, it admits an inverse.
\end{defn} The \textit{differentials} $\Omega_{g,i;n}(z_{[i]};w_{[n]})$ are $(n+i)$-differentials on  the $(n+i)$-th product of the formal neighbourhood of $L$ in $C$. The purpose of defining them  is to extract finite combinations of the correlators  from the  $\mathcal W$-algebra action on the Whittaker vector $\ket{\Gamma_{\Lambda}}$.
\begin{lem} \label{lem:Omegaexp}We have the following explicit expression for the $\Omega_{g,i;n}$ in terms of the correlators $\phi_{g,n}$,
	\begin{equation}
		\label{Omegagin}
		\Omega_{g,i;n}(z_{[i]};w_{[n]}) = \sum_{\substack{\mathbf{L} \vdash [i] \\ \sqcup_{L \in \mathbb{L}} N_{L} = [n] \\ i + \sum_{L} (g_{L} - 1) = g}} \prod_{L \in \mathbf{L}} \phi_{g_L,\#L+\#N_L}(z_{L},w_{N_L}).
	\end{equation} 
\end{lem}
\begin{proof}
	We omit the proof as the statement is a slight variant of \cite[Lemma 5.4]{BBCC} to include the unstable $\phi_{0,1}, \phi_{0,2}$ terms which are present in the Whittaker vector $\ket{\Gamma_{\Lambda}}$. See also \cite[Section 2]{BBCCN18} and \cite[Section 4]{BKS23}.
\end{proof} The formula \eqref{Omegagin}  shows that $\Omega_{g,i;n}$  contains precisely $i$ summands that involve the correlator $\phi_{g,1+n}$. It will be useful in the following to consider the expression  \eqref{Omegagin} where we remove these terms. More precisely, we define $\widehat{\Omega}_{g,i;n}$ as
	\begin{equation}\label{eq:Omega'def}
		\widehat{\Omega}_{g,i;n}(z_{[i]};w_{[n]}) := \sum_{\substack{\mathbf{L} \vdash [i] \\ \sqcup_{L \in \mathbb{L}} N_{L} = [n] \\ i + \sum_{L} (g_{L} - 1) = g}} \prod_{L \in \mathbf{L}} \phi_{g_L,\#L+\#N_L}(z_{L},w_{N_L}) - \sum_{j=1}^i \phi_{g,1+n}(z_j,w_{[n]})\prod_{l \neq j} \phi_{0,1}(z_l),
	\end{equation}
so that it does not involve any correlators $\phi_{g,1+n}$. The  $\mathcal{W}$-constraints \eqref{WhitG2} on $\ket{\Gamma_{\Lambda}}$ are equivalent  to the following restrictions on the $\Omega_{g,i;n}$.

\begin{lem}\label{lem:Omegaconst}
	The $\Omega_{g,i;n}$ satisfy the following condition for any $i \in [r]$:
	\begin{equation}
		\label{constraints}
		\sum_{\substack{Z \subseteq \mathfrak{f}(z) \\ \#Z = i}} \Omega_{g,i;n}(Z;w_{[n]}) = \delta_{g,0} \delta_{i,r} \delta_{n,0} \frac{(\Lambda \dd x)^{r}}{x^{r + 1}} + O\left(\frac{(\dd x)^{i}}{x^{i}}\right),
	\end{equation} where $x = x(z)$ and we recall that $\mathfrak{f}(z) = x^{-1}\left(x(z)\right)$ is the full fibre over the point $x(z)$. By $O\big(\frac{(\dd x)^i}{x^i}\big)$ we mean a quantity containing only terms of the form $x^k(\dd x)^i$ for $k \geq -i$; the $O$-notation is therefore understood as if the variable $x$ were approaching $0$.
\end{lem}
\begin{proof}
	The $\mathcal{W}$-constraints \eqref{WhitG2} that the Whittaker vector $\ket{\Gamma_{\Lambda}}$ satisfies can be written in the following  compact form using the expression \eqref{eq:Wgens} for the generators $\mathcal W^i(x)$ as elementary symmetric polynomials in the Heisenberg $1$-forms $\mathcal{J}$. For any $i \in [r]$, we have
	\begin{equation}\label{eq:Wconst}
		\ket{\Gamma_{\Lambda}}^{-1} \left(\sum_{1 \leq a_{1} < \cdots < a_{i} \leq r} \prod_{j=1}^i \mathcal{J}\left(\begin{smallmatrix} a_{j} \\ x \end{smallmatrix}\right) \right) \ket{\Gamma_{\Lambda}}  = \delta_{i,r} \frac{(\Lambda \dd x)^{r}}{x^{r + 1}} +  O\left(\frac{\left(\dd x\right)^{i}}{x^{i}}\right).
	\end{equation}  Applying the operator $\operatorname{ad}_{g,n}$ to the above equation, and using the definition of $\Omega_{g,i;n}$ from \cref{def:Omega} proves the lemma.
\end{proof}

\vspace{0.5cm}

\subsection{The spectral curve}

\medskip 
Let us treat the unstable correlators $\phi_{0,1}$ and $\phi_{0,2}$ first in order to obtain the spectral curve.
Recall the family of curves over $\Lambda \in \mathbb C^*$ that we have considered previously in \eqref{eq:Gaiottocurve}:
\begin{equation}\label{eq:Gaiottocurve2}
	\prod_{a = 1}^{r} \bigg(y - \frac{Q_a}{x}\bigg) + \frac{(-\Lambda)^{r}}{x^{r + 1}} = 0.
\end{equation} If $Q_1,\ldots, Q_r$ are pairwise distinct, the fibre  $S_\Lambda$  known as the half Seiberg--Witten  curve is a smooth curve of genus zero which admits the following explicit parametrisation with  coordinate $\zeta \in \mathbb P^1$.
\begin{equation}\label{eq:Gcurve}
x(\zeta) =  - \dfrac{\Lambda^r}{\prod_{a = 1}^{r} (Q_a - \zeta)}\,, \qquad y(\zeta)  =  -\dfrac{\zeta}{\Lambda^r} \prod_{a = 1}^{r} (Q_a - \zeta) \,.
\end{equation} Recall that to complete the description of the spectral curve, we define 
\begin{equation}
	\omega_{0,1}(\zeta) = y(\zeta)\dd x(\zeta),\qquad \omega_{0,2}(\zeta_1,\zeta_2) = \frac{\dd \zeta_1 \dd \zeta_2}{(\zeta_1 - \zeta_2)^2}.
\end{equation} We denote the ramification points of the branched covering $x : S_\Lambda \to \mathbb P^1$ excluding the one of index $r$ at $x = 0$ by $\operatorname{Ram}(S_\Lambda) \subset S_\Lambda$. These ramification points are all simple as long as  $Q_1,\ldots,Q_r$ are pairwise distinct. For each ramification point $\rho \in \operatorname{Ram}(S_\Lambda) $, we denote the associated  deck transformation (of degree two) by $\sigma_\rho $.

With this setup, we show that the unstable correlator $\phi_{0,1}$ can be analytically continued to the meromorphic differential $\omega_{0,1}$ on the half Seiberg--Witten  curve $S_\Lambda$. 
\begin{lem}\label{lem:01} Assume that $Q_1,\ldots, Q_r$ are pairwise distinct. The all-order series expansion of the meromorphic form $\omega_{0,1}(\zeta) $ on $S_{\Lambda}$ when $\zeta $ is near $ x^{-1}(\infty) \cong L $ with $1/x(\zeta)$ as a local coordinate, and then all-order series expansion as $\Lambda \rightarrow 0$, is given by $\phi_{0,1}(\zeta)$. Explicitly, we have
	\[
		\omega_{0,1}(\zeta)  \approx \phi_{0,1}\big(\begin{smallmatrix} \mathfrak{c}(\zeta) \\ x(\zeta) \end{smallmatrix}\big),
	\] where $\approx$ is our notation to indicate an identity of all-order expansions.
\end{lem}
\begin{proof}
If $g = 0$ and $n = 0$, for any $i \in [r]$, \cref{lem:Omegaconst} gives the following relation
\begin{equation}\label{eq:2.7}
\sum_{1 \leq a_{1} < \cdots < a_{i} \leq r} \prod_{j=1}^{i} \phi_{0,1}\big(\begin{smallmatrix} a_j \\ x \end{smallmatrix}\big) = \delta_{i,r}  \frac{(\Lambda\dd x)^{r}}{x^{r + 1}} + O\left(\frac{(\dd x)^{i}}{x^{i}}\right),
\end{equation} where we have used the explicit expression for $\Omega_{0,i;0}$ from \cref{lem:Omegaexp} and we recall that the $O(\cdots)$ here means terms of the form $(\dd x)^ix^{k}$ for $k \geq -i$. As $\phi_{0,1}\big(\begin{smallmatrix} a \\ x \end{smallmatrix}\big)  = \frac{Q_a \dd x }{x} + \cdots$ is a formal power series in $1/x$, we can determine the right-hand side of \eqref{eq:2.7} explicitly. Indeed, the coefficient of the term of order $\frac{(\dd x)^i}{x^{i}}$ is precisely the elementary symmetric polynomial $e_i(Q_1,\ldots, Q_r)$ and there are no terms of lower order. Thus, \eqref{eq:2.7} becomes
\begin{equation*}
	\sum_{1 \leq a_{1} < \cdots < a_{i} \leq r} \prod_{j=1}^{i} \phi_{0,1}\big(\begin{smallmatrix} a_j \\ x \end{smallmatrix}\big)  = \delta_{i,r}  \frac{(\Lambda\dd x)^{r}}{x^{r + 1}} + e_i(Q_1,\ldots, Q_r) \frac{(\dd x)^{i}}{x^{i}},
\end{equation*} which can be put into generating series form as 
\begin{equation}\label{eq:2.9}
	\prod_{a = 1}^r \left(u + \phi_{0,1}\big(\begin{smallmatrix} a \\ x \end{smallmatrix}\big)\right) = \frac{(\Lambda\dd x)^r}{x^{r + 1}} + \prod_{a = 1}^r \left(u  + \frac{Q_a \dd x}{x}\right).
\end{equation} By substituting  $ u = -   \phi_{0,1}\big(\begin{smallmatrix} b \\ x \end{smallmatrix}\big)$ for any $b \in [r]$, we obtain  the following equation
\begin{equation*}
0= \frac{(-\Lambda \dd x)^r}{x^{r + 1}} + \prod_{a = 1}^r \left(  u'-  \frac{Q_a \dd x}{x}\right),
\end{equation*} whose $r$ independent solutions are given by $u' = \phi_{0,1}\big(\begin{smallmatrix} b \\ x \end{smallmatrix}\big)$ for any $b \in [r]$. As the above equation gives a formula for the  $ \phi_{0,1}$ in terms of $x(z)$, we see that it analytically continues to a meromorphic differential on the curve $S_\Lambda$. To show that the analytic continuation of  $\phi_{0,1} $ matches  $\omega_{0,1}$, note that \eqref{eq:Gaiottocurve2}  provides a set  of $r$ independent solutions as well --- indeed, take $u' = \omega_{0,1}(\zeta)$, where $\omega_{0,1} $ is considered as an expansion in $1/x(\zeta)$ with $\zeta $ near $x^{-1}(\infty)\cong L$. These two different sets of $r$ solutions must coincide, and by analysing the leading coefficient in the expansion of $\omega_{0,1}$, we get the claim.
\end{proof}

Let us turn now to the statement for $(g,n) = (0,2)$. It will turn out to be useful to define the following projection operators.
\begin{defn}
For $i \in \mathbb{Z}$, define the \textit{projection operator}
\begin{equation*}
\begin{split}
	p_{\geq -i} : \mathbb C\left[\!\left[ x(z)^{\pm 1}\right]\!\right] (\dd  x(z))^i &\to   \mathbb C\left[\!\left[ x(z)^{\pm 1}\right]\!\right] (\dd  x(z))^i \,, \\ \sum_{n \in \mathbb Z} x(z)^{n} (\dd  x(z))^i &\mapsto  \sum_{n \geq -i} x(z)^{n} (\dd  x(z))^i.
\end{split}
\end{equation*}
It acts on a formal power series in $  x(z)^{\pm 1}$ and keeps only the terms of order $x(z)^{-i}$ and higher. Analogously, we define the  projection operator $p_{\leq -i} $ which  keeps only the terms of order $x(z)^{-i}$ and lower.
\end{defn}

\begin{lem}\label{lem:02}
	Assume that $Q_1,\ldots, Q_r$ are pairwise distinct. The series expansion of $\omega_{0,2}(\zeta_1,\zeta_2) $ as $\zeta_1,\zeta_2 $ is near $ x^{-1}(\infty) \cong L $ with $1/x(\zeta_1), 1/x(\zeta_2) $ as a local coordinate is given by $\phi_{0,2}(\zeta_1,\zeta_2)$:
	\[
	\omega_{0,2}(\zeta_1,\zeta_2) \approx \phi_{0,2}\big(\begin{smallmatrix} \mathfrak c(\zeta_1)& \mathfrak c(\zeta_1) \\ x(\zeta_1) & x(\zeta_2) \end{smallmatrix}\big),
	\] where $\omega_{0,2}(\zeta_1,\zeta_2) $ is considered as an expansion in $1/x(\zeta_i)$ near $\zeta_i  = Q_{a_i}$.
\end{lem}

\begin{proof}
	\cref{lem:Omegaconst} when $(g,n) = (0,1)$, in combination with the explicit formula for  $\Omega_{0,i;2} $ proved in \cref{lem:Omegaexp}, imposes the following restriction on the correlator $\phi_{0,2}$ for any $i \in [r]$:
\begin{equation}
\label{constrg0n1}
\sum_{\substack{Z \subseteq \mathfrak{f}(z) \\ \#Z = i \\ z' \in Z}} \phi_{0,2}(z', w) \prod_{z''  \subseteq Z \setminus \{z'\}} \phi_{0,1}(z'') = O\bigg(\frac{(\dd x(z))^{i}}{x(z)^{i}}\bigg).
\end{equation}
We consider the left-hand side of the above equation as a formal series in $x(z') $ ($ = x(z)$), where we first expand in $x(w)$ near $\infty$, and then in $x(z) $ near $\infty$. In other words, we expand in the region $|x(z)| < |x(w)|$. Our goal is to determine the right-hand side of \eqref{constrg0n1} explicitly. Applying the projection operator $p_{\geq-i} $ to the left-hand side of \eqref{constrg0n1} yields
\begin{equation}\label{eq:2.11}
 p_{\geq -i}\left( \sum_{\substack{\{z'\} \subseteq Z \subseteq \mathfrak{f}(z) \\ \#Z = i \\ \mathfrak{c}(z') = \mathfrak{c}(w)}} \frac{\dd x(z')\dd x(w)}{(x(z')-x(w))^{2}} \prod_{z''  \in Z \setminus \{z'\}} \phi_{0,1}(z'')  \right),
\end{equation}
which follows from the simple observations
\begin{equation}
\begin{split} & \quad p_{\geq{-(i-2)}} \left(\prod_{z''  \in Z \setminus \{z\}} \phi_{0,1}(z'') \right) = 0, \\
& \quad p_{\geq{-1}} \left( \phi_{0,2}(z', w) \right) = p_{\geq{-1}} \left( \delta_{\mathfrak{c}(z'),\mathfrak{c}(w)} \frac{\dd x(z')\dd x(w)}{(x(z')-x(w))^{2}} \right).
\end{split}
\end{equation}
Thus, all the other terms disappear upon applying the projection $p_{\geq -i}$. For any subset $Z \subseteq \mathfrak{f}(z)$, we denote $\mathfrak{c}(Z) = \{\mathfrak{c}(z)\,\,|\,\,z \in Z\}$. Then, given a non-empty subset $Z \subseteq \mathfrak{f}(z)$ and an element $z' \in Z$, we decompose
\[
\prod_{z'' \subseteq Z \setminus \{z'\}} \frac{\phi_{0,1}(z'')}{\dd x(z'')} = \sum_{j\geq \#Z-1} c^{\mathfrak{c}(Z),\mathfrak{c}(z')}_j x(z)^{-j}.
\]
The coefficients $c^{\mathfrak{c}(Z),\mathfrak{c}(z')}_j$ only depend on the set $\mathfrak{c}(Z)$ and the element $\mathfrak{c}(z')$, and not directly on $Z$ or $z'$.
Using this notation, we rewrite the right-hand side of \eqref{eq:2.11} as follows:
\begin{equation}
\label{eq:2.12}
	p_{\geq -i}\left[  \left(\dd x(z)\right)^i \dd_{w}\left(\sum_{\substack{Z \subseteq \mathfrak{f}(z) \\ \#Z = i }} \sum_{\substack{z' \in Z \\ \mathfrak{c}(z')=\mathfrak{c}(w)}}   \sum_{j\geq i-1}c^{\mathfrak{c}(Z),\mathfrak{c}(z')}_j \bigg(\frac{x(z)^{-j} -  x(w)^{-j}}{x(z)-x(w)} + \frac{x(w)^{-j}}{x(z) - x(w)}\bigg)\right)\right].
\end{equation}
Recall that we always expand in the region $|x(z)| < |x(w)|$, and hence the last term remains unchanged after applying the projection $p_{\geq -i}$ with $i \geq 1$. Let us compute the  result of applying the projection to the first term. We have
\begin{equation*}
\begin{split}
p_{\geq -i}  \left[ \sum_{j\geq i-1}c^{\mathfrak{c}(Z),\mathfrak{c}(z')}_j  \frac{x(z)^{-j} -  x(w)^{-j}}{x(z)-x(w)}\right] 
&= p_{\geq -i}  \left[   \sum_{j\geq i-1}\frac{c^{\mathfrak{c}(Z),\mathfrak{c}(z')}_j}{x(w)^j} \frac{1}{x(z)^j} \frac{x(w)^{j} -  x(z)^{j}}{x(z)-x(w)} \right] \\
&=    \left(\sum_{j\geq i-1}\frac{c^{\mathfrak{c}(Z),\mathfrak{c}(z')}_j}{x(w)^j} \right) \frac{1}{x(z)^i} \frac{x(w)^{i} -  x(z)^{i}}{x(z)-x(w)},
\end{split}
\end{equation*}
where in order to get the last line, we simply remove all the terms of order strictly lower than $x(z)^{-i}$. Applying this simplification to \eqref{eq:2.12} gives the following simplified expression for the right-hand side of  \eqref{constrg0n1}
\begin{equation}
 \left(\dd x(z)\right)^i \dd_{w}\left(\sum_{\substack{\{w\} \subseteq W \subseteq \mathfrak{f}(w) \\ \#W = i}}  \left( \prod_{w''  \in W \setminus \{w\}} \frac{\phi_{0,1}(w'')}{\dd x(w'')}  \right) \frac{x(w)^{i}}{x(z)^i (x(z)-x(w))} \right),
\end{equation}
where we have replaced the sum over $Z \subseteq \mathfrak{f}(z)$ by a sum over subsets $W \subseteq \mathfrak{f}(w)$, as the coefficients $c_j$ do not directly depend on the $Z$ as previously noted. Thus, we have fully determined the right-hand side of \eqref{constrg0n1}, and  equation \eqref{constrg0n1} takes the following equivalent form:
\begin{multline*}
\sum_{\substack{Z \subseteq \mathfrak{f}(z) \\ \#Z = i \\ z' \in Z}} \frac{\phi_{0,2}(z', w)}{\dd x(z')} \prod_{z''  \in Z \setminus \{z'\}} \frac{\phi_{0,1}(z'')}{\dd x(z'')} \\
	 =  \dd_{x(w)}\left(\sum_{\substack{\{w\} \subseteq W \subseteq \mathfrak{f}(w) \\ \#W = i}}  \left( \prod_{w''  \in W \setminus \{w\}} \frac{\phi_{0,1}(w'')}{\dd x(w'')}  \right) \frac{x(w)^{i}}{x(z)^i (x(z)-x(w))} \right).
\end{multline*}
Finally, just as in \cref{lem:01}, we consider the generating series by applying  $\sum_{i =1}^{r} u^{r-i}$ to the above equation and specialise the equation to $u=-\frac{\phi_{0,1}(z)}{\dd x(z)}$, which gives
\begin{multline}\label{eq:02explicit}
	 \frac{\phi_{0,2}(z, w) }{\dd x(z)} \prod_{z''  \in \mathfrak{f}'(z) } \left(\frac{\phi_{0,1}(z'')- \phi_{0,1}(z)}{\dd x(z)}\right) 
	\\ = \dd_{w}\left(  \frac{x(w)}{ x(z)^r(x(z)-x(w))} \prod_{w''  \in \mathfrak{f}'(w)}  \left(\frac{x(w)\phi_{0,1}(w'')}{\dd x(w'')}  - \frac{x(z)\phi_{0,1}(z)}{\dd x(z)} \right)  \right).
\end{multline}
This equation gives an expression for $\phi_{0,2}(z, w)$ entirely in terms of $\phi_{0,1} $ and $x$. As \cref{lem:01} states that $\phi_{0,1}$ analytically continues to the meromorphic differential $\omega_{0,1}$ on the half Seiberg--Witten  curve $S_\Lambda$, and  $x$ analytically continues to the function $x(\zeta)$ on $S_\Lambda$ by definition, we  conclude that $ \phi_{0,2}(z, w)$  analytically continues to a meromorphic bi-differential on $S_{\Lambda}$, say $\tilde{\phi}_{0,2}(\zeta_1,\zeta_2)$.

There are two ways to show that $\tilde{\phi}_{0,2}$ coincides with the bi-differential $\omega_{0,2}(\zeta_1,\zeta_2) = \frac{\dd \zeta_1 \dd \zeta_2}{(\zeta_1 - \zeta_2)^2}$. The first way is to check that the pole structure match, i.e., the only pole is a double pole on the diagonal with bi-residue $1$. Here we prefer a second way, which is a direct computation. Using $x(\zeta)y(\zeta) = \zeta$  we can rewrite
\begin{align*}
 \tilde{\phi}_{0,2}&(\zeta_1,\zeta_2)  \\
& =  \prod_{\zeta_1'  \in \mathfrak{f}'(\zeta_1)}\frac{1}{y(\zeta_1')- y(\zeta_1)} \cdot \dd_{\zeta_2}\left(  \frac{\dd x(\zeta_1) x(\zeta_2)}{ x(\zeta_1)^r(x(\zeta_1)-x(\zeta_2))} \prod_{\zeta_2'  \in \mathfrak{f}'(\zeta_2) }  \left( x(\zeta_2')y(\zeta_2') - x(\zeta_1)y(\zeta_1) \right)  \right) \\
&   = \prod_{\zeta_1'  \in \mathfrak{f}'(\zeta_1)}\frac{1}{ \zeta_1'- \zeta_1} \cdot \dd_{\zeta_2}\left(  \frac{\dd x(\zeta_1) x(\zeta_2)}{ x(\zeta_1)(x(\zeta_1)-x(\zeta_2))} \prod_{\zeta_2'  \in \mathfrak{f}'(\zeta_2) }  \left( \zeta_2' - \zeta_1 \right)  \right),
\end{align*} where  the notation $\mathfrak f'(\zeta)$ now means $x\circ x^{-1}(\zeta) \setminus \{\zeta\}$ and $x$ is the meromorphic function on $S_\Lambda $ from  \eqref{eq:Gaiottocurve2}. Setting $D(\zeta) = \prod_{a = 1}^{r} (Q_a - \zeta)$ we have $x(\zeta) = -\Lambda^r/D(\zeta)$, and thus
\begin{align*}
x(\zeta_1) - x(\zeta_2) & = \frac{\Lambda^r(D(\zeta_2) - D(\zeta_1))}{D(\zeta_1)D(\zeta_2)} = -\frac{\Lambda^r}{D(\zeta_1)D(\zeta_2)} \prod_{\zeta_2' \in \mathfrak{f}(\zeta_2)} (\zeta_2' - \zeta_1) \\
\dd x(\zeta_1) & = \dd \zeta_1 \lim_{\zeta_2 \rightarrow \zeta_1} \frac{x(\zeta_1) - x(\zeta_2)}{\zeta_1 - \zeta_2} = \frac{\Lambda^r \dd \zeta_1}{D(\zeta_1)^2} \prod_{\zeta_1' \in \mathfrak{f}'(\zeta_1)} (\zeta_1' - \zeta_1). 
\end{align*}
Then:
\[
\tilde{\phi}_{0,2}(\zeta_1,\zeta_2) = \frac{\Lambda^r  \dd \zeta_1}{D(\zeta_1)^2\dd x(\zeta_1)} \dd_{\zeta_2}\left(\frac{\dd x(\zeta_1) D(\zeta_1)}{D(\zeta_2)}\,\frac{-D(\zeta_1)D(\zeta_2)}{\Lambda^r (\zeta_2 - \zeta_1)}\right) = \frac{\dd \zeta_1\dd \zeta_2}{(\zeta_1 - \zeta_2)^2}.
\qedhere \]
\end{proof}
This completes the proof of \cref{th:main} for the unstable correlators.

\vspace{0.5cm}

\subsection{Stable correlators and topological recursion}

\medskip

Finally, we turn to the stable correlators $\phi_{g,n}$ where $2g-2+n>0$. We proceed in two steps. First, we show that the correlators $\phi_{g,n}$ analytically continue to meromorphic differentials on the half Seiberg--Witten  curve $S_\Lambda$ with poles only at the ramification points $\operatorname{Ram}(S_\Lambda)$. Second, we show that these analytically continued correlators satisfy the linear and quadratic abstract loop equations. Then, as the system of correlators $\omega_{g,n}$ constructed by the topological recursion is the unique solution to the abstract loop equations \cite{BEO, BSblob,TRlimits}, we  conclude that the analytic continuation of the stable $\phi_{g,n}$ coincides with $\omega_{g,n}$.

\begin{prop} 
\label{prop:stablecorrelators}
	Assume that $Q_1,\ldots, Q_r$ are pairwise distinct. If $2g-2+n >0$, then the correlators $\phi_{g,n} $ which are $n$-differentials on a formal neighbourhood of $ L \cong x^{-1}(\infty)$ in $C$  admit an analytic continuation as meromorphic $n$-differentials on the half Seiberg--Witten  curve $S_\Lambda$. Moreover, these analytic continuations, denoted $\tilde{\phi}_{g,n}$, only have poles at the ramification points $\operatorname{Ram}(S_\Lambda)$ of the spectral curve.
\end{prop}

\begin{proof}
	Let us start by  extracting the terms containing $\phi_{g,1+n}$ from  the $\Omega_{g,i;n}$ using the  $\widehat{\Omega}_{g,i;n} $ defined in equation \eqref{eq:Omega'def}. For any $2g-2+(1+n)>0$, and $i \in [r]$, we have for a fixed $z$
	\begin{equation*}
		\sum_{\substack{Z \subseteq \mathfrak{f}(z) \\ \#Z = i}}  \Omega_{g,i;n}(Z;w_{[n]}) = \sum_{\substack{\{z\} \subseteq Z \subseteq \mathfrak{f}(z) \\ \#Z = i}} \phi_{g,1+n}(z,w_{[n]})\prod_{\substack{z' \in Z \setminus\{z\}}} \phi_{0,1}(z')+ \sum_{\substack{\\ Z \subseteq \mathfrak{f}(z) \\ \#Z = i}} \widehat{\Omega}_{g,i;n}(Z;w_{[n]}).
	\end{equation*} The first sum on the right-hand side involving the $ \phi_{g,1+n}$ vanishes upon applying  the projection $p_{\geq -i}$. This means that the constraints of \cref{lem:Omegaconst} can be written explicitly as 
	\begin{equation}\label{eq:projconst}
			\sum_{\substack{\{z\} \subseteq Z \subseteq \mathfrak{f}(z) \\ \#Z = i}} \phi_{g,1+n}(z,w_{[n]})\prod_{z' \in Z \setminus\{z\}} \phi_{0,1}(z') = - p_{\leq -i-1} \left[ \sum_{\substack{Z \subseteq \mathfrak{f}(z) \\ \#Z = i}} \widehat{\Omega}_{g,i;n}(Z;w_{[n]})\right].
	\end{equation}
	The usual trick of applying $\sum_{i=1}^{r}u^{r-i}$ to  the above equation to get the generating series, and then specialising to $u=- \phi_{0,1}(z)$ gives the following equation
	 \begin{multline}\label{eq:extractphi}
		 \quad \phi_{g,1+n}(z,w_{[n]}) \\
		 = \frac{-1}{\prod_{z'\in \mathfrak f'(z)} \left( \phi_{0,1}(z') - \phi_{0,1}(z)\right)} \sum_{i=1}^r \left(-\phi_{0,1}(z)\right)^{r-i} p_{\leq -i-1} \left[ \sum_{\substack{\\ Z \subseteq \mathfrak{f}(z) \\ \#Z = i}} \widehat{\Omega}_{g,i;n}(Z;w_{[n]})\right].
	 \end{multline}
	 We can rewrite the projection appearing on the right-hand side as 
	\[
	 p_{\leq -i-1} \left[ \sum_{\substack{\\ Z \subseteq \mathfrak{f}(z) \\ \#Z = i}} \widehat{\Omega}_{g,i;n}(Z;w_{[n]})\right] = (\dd x(z))^i \oint_{\gamma} \frac{\dd x(v)}{x(z)-x(v)} \left(\frac{x(v)}{x(z)} \right)^i \sum_{\substack{\\ V \subseteq \mathfrak{f}(v) \\ \#V = i}} \frac{\widehat{\Omega}_{g,i;n}(V;w_{[n]})}{(\dd x(v))^i},
	\]
	where we choose the contour $\gamma$ in the  $x(v)$-plane to be centred at $x(v) = \infty$ such that   $|x(v)| < |x(z)|$, i.e., the point $x(z)$ lies inside the contour $\gamma$. Also recall our standing assumption on $ \widehat{\Omega}_{g,i;n}(V;w_{[n]}) $ that  $ |x(v)| < |x(w_j)| $ for any $ j \in [n]$, which implies that  $x(w_j)$ are inside this contour $\gamma$. From this rewriting as a contour integral, we see from \eqref{eq:extractphi} that $\phi_{g,1+n}$ can be expressed in terms of  $x$ and $\phi_{g',1+n'}$ with $2g'-2+(1 + n') < 2g-2+(1+n)$. We know from \cref{lem:01} and \cref{lem:02} that the unstable correlators $\phi_{0,1} , \phi_{0,2} $ analytically continue to the meromorphic differentials $\omega_{0,1}, \omega_{0,2}$ respectively.  Thus, by induction on $(2g-2+n)$, we see that $\phi_{g,1+n}(z,w_{[n]})$ analytically continues to a meromorphic n-differential on the curve $S_\Lambda$. Let us denote these analytically continued differentials by $\tilde{\phi}_{g,1+n}(\zeta_0,\zeta_{[n]})$,
	\begin{multline}
	\label{eq:phitilde}
	 \quad \tilde{\phi}_{g,1+n}(\zeta_0,\zeta_{[n]})  = \\ \frac{y(\zeta_0)^r \dd x(\zeta_0)}{\prod_{\zeta_0'\in \mathfrak f'(\zeta_0)} \left( y(\zeta_0) - y(\zeta'_0)\right)} \sum_{i=1}^r \oint_{\gamma} \frac{\dd x(\zeta)}{x(\zeta_0)-x(\zeta)} \left(\frac{-x(\zeta)}{\zeta_0}\right)^i \sum_{\substack{Z \subseteq \mathfrak{f}(\zeta) \\ \#Z = i}} \frac{\widehat{\Omega}_{g,i;n}(Z;\zeta_{[n]})}{(\dd x(\zeta))^i},
	\end{multline} where we have used the fact that $x(\zeta)y(\zeta)=\zeta$. In this formula $\widehat{\Omega}_{g,i;n}$ is a combination of the $\tilde{\phi}_{g',1+n'}$ with $2g'-2+ (1 + n') < 2g-2+(1+n)$, and the contour $\gamma $ is now a contour in the $x(\zeta)$-plane centred at $x(\zeta) = \infty $ such that $|x(\zeta)|<|x(\zeta_0)|$. 
	
	To understand the poles of the $\tilde{\phi}_{g,1+n}$, we flip the contour to evaluate the residues outside $\gamma$. By the induction hypothesis, the only possible poles in the integrand of \eqref{eq:phitilde} in the variable $\zeta$ that are outside the contour $\gamma$ are at the ramification points -- i.e., the points $\operatorname{Ram}(S_\Lambda)$ and the index $r$ ramification point at $x(\zeta) = 0$. The other possible poles are at $x(\zeta) = x(\zeta_j)$ for $j \in 0 \cup [n] $, all of which lie inside the contour $\gamma$. Moreover, by induction  the integrand does not have a pole at the point $x(\zeta) = 0$ where $\zeta = \infty$. Indeed, as the $\tilde{\phi}_{g',1+n'}$ with $2g'-2+ (1 + n') < 2g-2+(1+n)$ are assumed to be holomorphic at $\zeta = \infty$, near $\zeta = \infty$ (where $x = 0$) the integrand behaves as 
	\[
		\frac{d\zeta}{\zeta^{r+1} } \times \frac{1}{\zeta^{r i}} \times \frac{\zeta^{i(r+1)}}{ \zeta^{2i}} = \frac{d \zeta}{\zeta^{r+i+1}},
	\]  as $x(\zeta)$ behaves as $\frac{1}{\zeta^r}$. Thus the integrand  is holomorphic at the ramification point $x = 0$. Finally, by evaluating the contour integral, we see that the differential $\tilde{\phi}_{g,1+n}$ only has poles at  the ramification points $\operatorname{Ram}(S_\Lambda)$ in the variables $\zeta_0$ and $\zeta_{[n]}$.
\end{proof}

Let us now show that the analytically continued correlators $\tilde{\phi}_{g,n}$ on the spectral curve $S_\Lambda$ satisfy the abstract loop equations, as considered in \cite{BSblob, BEO}.
\begin{prop}\label{prop:ALE}
For any $g,n \geq 0$, and any ramification point $\rho \in \operatorname{Ram}(S_\Lambda)$ of the spectral curve $S_\Lambda$, the analytically continued correlators $\tilde{\phi}_{g,1+n}$ satisfy the linear and quadratic loop equations, i.e., the two expressions
\begin{equation*}
\begin{split}
  &   \tilde{\phi}_{g,1+n}(\zeta_0, \zeta_{[n]}) + \tilde{\phi}_{g,1+n}(\sigma_{\rho}(\zeta_0), \zeta_{[n]}), \\
  & \tilde{\phi}_{g-1,2+n}(\zeta_0, \sigma_{\rho}(\zeta_0), \zeta_{[n]}) + \sum_{\substack{g_1+g_2=g \\ J_1 \sqcup J_2= [n]}} \tilde{\phi}_{g_1, 1+\#J_1}(\zeta_0, \zeta_{J_1}) \tilde{\phi}_{g_2, 1+\#J_2}\left(\sigma_{\rho}(\zeta_0), \zeta_{J_2}\right),\end{split}
  \end{equation*}
are holomorphic as $\zeta_0$ approaches $\operatorname{Ram}(S_\Lambda)$.
\end{prop}

\begin{proof}
	The  $ i =1$ case of the constraints on $\Omega_{g,i;n}$ obtained in \eqref{eq:projconst} gives
	\[
	\sum_{z' \in \mathfrak{f}(z)} \phi_{g,1+n}(z',w_{[n]}) = 0\,.
	\] Passing to the analytic continuation, we can write for each ramification point $\rho \in \operatorname{Ram}(S_\Lambda)$:
	\[
	\tilde{\phi}_{g,1+n}(\zeta_0,\zeta_{[n]}) + 	\tilde{\phi}_{g,1+n}(\sigma_\rho(\zeta_0),\zeta_{[n]})  = -\left( \sum_{\zeta'_0 \in \mathfrak f'(\zeta_0) \setminus \{\sigma_\rho(\zeta_0)\}}  \tilde{\phi}_{g,1+n}(\zeta_0',\zeta_{[n]})  \right).
	\]
	As the $\tilde{\phi}_{g,n}$ only have poles at the ramification points, the right-hand side clearly has no poles at $\zeta  = \rho$, and this proves the linear loop equation. 
	
	Let us turn to the quadratic loop equation now. Consider the analytic continuation of the statement of \cref{lem:Omegaconst} for $i = 2$:
		  \begin{multline}\label{eq:i=2}\frac{1}{2} \sum_{\substack{\zeta_0',\zeta_0''\in\mathfrak{f}(\zeta_0)\\ \zeta'_0 \neq \zeta_0''}} \tilde{\phi}_{g-1,2+n}(\zeta'_0, \zeta''_0, \zeta_{[n]}) + \frac{1}{2} \sum_{\substack{\zeta_0',\zeta_0''\in\mathfrak{f}(\zeta_0)\\ \zeta_0' \neq \zeta_0''}}  \sum_{\substack{g_1+g_2=g \\ J_1 \sqcup J_2 = [n]}} \tilde{\phi}_{g_1, 1+\#J_1}(\zeta_0', \zeta_{J_1}) \tilde{\phi}_{g_2,1+ \#J_2}(\zeta''_0, \zeta_{J_2}) \\
		  = \delta_{g,0} \delta_{r,2} \delta_{n,0} \frac{\Lambda^r (\dd x(\zeta_0))^{2}}{x(\zeta_0)^{3}} + O\left(\frac{(\dd x(\zeta_0))^{2}}{x(\zeta_0)^{2}}\right).
\end{multline}
The factor of $2$ is due to the fact that we sum over ordered pairs $\zeta_0',\zeta_0''$ instead of pairs $\{\zeta_0',\zeta_0''\}$. We claim that the right-hand side is regular at any ramification point $\rho \in \operatorname{Ram}(S_\Lambda)$, as there are no terms of order $ \frac{(\dd x(\zeta))^{2}}{x(\zeta)^{k}}  $ for any $k \geq 4$ . Indeed, the series expansion at $x(\zeta) = \infty$ of a pole at any ramification point (note that $x = \infty$ is not a branch point), would create holomorphic terms of the form $\frac{(\dd x(\zeta))^{2}}{x(\zeta)^{k}} $ for arbitrarily large $k \gg 0$. Now, we consider the terms of interest for the quadratic loop equations appearing in the left-hand side of equation \eqref{eq:i=2}. The first sum can be split into
		  	\begin{multline*}
\tilde{\phi}_{g-1,2+n}(\zeta_0, \sigma_\rho(\zeta_0), \zeta_{[n]}) + \frac{1}{2} \sum_{\substack{\zeta_0',\zeta_0'' \in\mathfrak{f}'(\zeta_0) \setminus \{\sigma_{\rho}(\zeta_0)\} \\ \zeta_0' \neq \zeta_0'' } } \tilde{\phi}_{g-1,2+n}(\zeta_0',\zeta_0'' , \zeta_{[n]}) \\ 
 + \sum_{\zeta_0' \in\mathfrak{f}'(\zeta_0) \setminus \{\sigma_{\rho}(\zeta_0)\}} \left(\tilde{\phi}_{g-1,2+n}(\zeta_0, \zeta'_0, \zeta_{[n]}) +\tilde{\phi}_{g-1,2+n}(\sigma_{\rho}(\zeta_0), \zeta_0', \zeta_{[n]})\right),
\end{multline*}
		   where the second line is regular at the ramification point $\rho$ as the $ \tilde{\phi}_{g-1,2+n} $ only have poles at the ramification points, and the third line is regular at $ \rho$ thanks to the linear loop equation. The second sum in \eqref{eq:i=2} now, which can be rewritten in a similar fashion:
		   \begin{multline*}
		     \sum_{\substack{g_1+g_2=g \\ J_1 \sqcup J_2 = [n]}} \bigg(  \tilde{\phi}_{g_1, 1+\#J_1}(\zeta_0, \zeta_{J_1}) \tilde{\phi}_{g_2, 1+\#J_2}(\sigma_\rho(\zeta_0), \zeta_{J_2}) \\
		      + \frac{1}{2} \sum_{\substack{\zeta'_0,\zeta_0'' \in\mathfrak{f}'(\zeta_0) \setminus \{\sigma_{\rho}(\zeta_0)\} \\ \zeta_0' \neq \zeta_0''}} \tilde{\phi}_{g_1, 1+\#J_1}(\zeta'_0, \zeta_{J_1}) \tilde{\phi}_{g_2, 1+\#J_2}(\zeta_0'', \zeta_{J_2})  \\   + \sum_{\zeta_0' \in\mathfrak{f}'(\zeta_0) \setminus \{\sigma_{\rho}(\zeta_0)\}}  \phi_{g_1,1+\#J_1}(\zeta_0',\zeta_{J_1})\left( \tilde{\phi}_{g_2, 1+\#J_2}(\zeta_0,\zeta_{J_2}) + \tilde{\phi}_{g_2,1+\# J_2}(\sigma_\rho(\zeta_0), \zeta_{J_2})\right) \bigg).
\end{multline*}
			Again, the second line is clearly regular at $\rho$, while the third line is regular thanks to the linear loop equation. Putting  this together proves the quadratic loop equation. 
\end{proof}

\begin{proof}[Proof of \cref{th:main}]
	We are now in position to finish the proof. For the unstable correlators, i.e., $(g,n) = (0,1)$ and $(g,n) = (0,2)$, we have already proved the theorem in \cref{lem:01} and \cref{lem:02} respectively. As for the stable correlators $\phi_{g,1+n}$ with $2g - 2 + (1 + n) > 0$, we proved in \cref{prop:stablecorrelators} that they admit analytic continuations $\tilde{\phi}_{g,1+n}$ with poles only at the ramification points $\operatorname{Ram}(S_\Lambda)$. In \cref{prop:ALE} we showed that the correlators $\tilde{\phi}_{g,1+n}$ satisfy the abstract loop equations. Besides, the linear loop equation implies that for any ramification point $\rho \in \operatorname{Ram}(S_\Lambda)$
	\begin{equation}
	\label{Res0}
	\Res_{\zeta = \rho} \tilde{\phi}_{g,1+n}(\zeta,\zeta_{[n]}) = 0
	\end{equation}
	for reasons of parity with respect to $\sigma_{\rho}$. Since $S_{\Lambda} \simeq \mathbb{P}^1$, we  have by the Cauchy residue formula
\begin{equation}
\label{normphig}
\begin{split}
	\tilde{\phi}_{g,1+n}(\zeta_0,\zeta_{[n]}) & = \sum_{\rho \in \operatorname{Ram}(S_\Lambda)} \Res_{\zeta = \rho} \dd \zeta_0 \left(\frac{1}{\zeta_0 - \zeta} - \frac{1}{\zeta_0 - \rho}\right) \tilde{\phi}_{g,1+n}(\zeta,\zeta_{[n]}) \\
	& = \sum_{\rho \in \operatorname{Ram}(S_\Lambda)} \Res_{\zeta = \rho}  \left(\int_{\rho}^{\zeta} \omega_{0,2}(\cdot,\zeta_0)\right) \tilde{\phi}_{g,1+n}(\zeta,\zeta_{[n]}). 
	\end{split}
	\end{equation} 
	where the term $\frac{1}{\zeta_0 - \rho}$ does not contribute to the residue due to \eqref{Res0}. According to \cite{BEO,TRlimits}, meromorphic differentials with poles only at the ramification points that satisfy the abstract loop equations and the normalisation property \eqref{normphig} are uniquely reconstructed by the topological recursion. Thus, the differentials $\tilde{\phi}_{g,1+n}$ coincide with the topological recursion correlators $\omega_{g,1+n}$.
\end{proof}

\vspace{1cm}

\section{The CDO vector: proof of \texorpdfstring{\cref{th:main2}}{Theorem 2.3}}
\label{S4}

\medskip

In this section, we study a different set $\mathcal W$-constraints that are relevant for Hurwitz theory. These constraints can be viewed as a generalisation of the ones characterising the Gaiotto vector that we studied previously in Section~\ref{S3}, and we will prove that the associated partition function can be computed via the topological recursion on a ramified spectral curve.  The half Seiberg--Witten  curve \eqref{eq:Gaiottocurve} of pure gauge theory  can be recovered as a limiting case.

\vspace{0.5cm}

\subsection{The Airy structure}

\medskip

We describe the Airy structure with degree zero terms constructed and studied in \cite{CDO24} by Chidambaram, Do{\l}{\k{e}}ga and Osuga, based on the modes of the $\mathcal{W}$-algebra $\mathcal{W}(\mathfrak{gl}_r)$. While the Airy structure exists at arbitrary shifted level $\kappa$ --- and this is crucial for the study of $b$-Hurwitz numbers in \cite{CDO24} --- we restrict here to the case of self-dual level $\kappa = 1$.

Consider two tuples of parameters $\mathbf P = (P_1,\cdots, P_r)$ and $\mathbf Q = (Q_1,\ldots, Q_{r - 1})$ and introduce the automorphism of the Heisenberg VOA which sends the modes $J^a_k$ to $\widetilde{J}^a_k$ with
\begin{equation}
\label{eq:Jrep}
\begin{split}
\widetilde{J}^a_{k} = \left\{\begin{array}{lr}  J^a_k & k > 0; \\[8pt]
						J^a_k  +(-\Lambda)^{-r}\delta_{a,r} \delta_{k,-1} & k < 0; \\[8pt]
						Q_{a} & k = 0 \,\,\text{and}\,\,a \in [r-1]; \\[8pt]  -\big(|\mathbf{P}| + |\mathbf{Q}|\big) & k = 0\,\,\text{and}\,\,a = r.
		\end{array}
		\right.
	\end{split}
	\end{equation}
	where we denote $|\mathbf{P}| = e_1(\mathbf{P}) = \sum_{a = 1}^{r} P_a$ and likewise for $|\mathbf{Q}|$.
We consider a variant of the Verma module \eqref{VLambda} where we allow $\Lambda$ to be inverted
\[
\widetilde{\mathcal{V}}_{\Lambda} = \mathbb{C}(\!(\mathbf P, \mathbf Q, \Lambda^r)\!)\big[\!\big[(J_{-k}^{a})_{k \in \mathbb{Z}_{> 0}}^{a \in [r]}\big]\!\big](\!(\hslash)\!).
\]
We let the Heisenberg algebra act on $\widetilde{\mathcal{V}}_{\Lambda}$ with the modified modes \eqref{eq:Jrep}, and this restricts to a representation of the $\mathcal{W}(\mathfrak{gl}_r)$-algebra by the modes
\[
\widetilde{W}^i_k = \sum_{1 \leq a_1<\cdots < a_i \leq r} \sum_{\substack{k_1,\ldots, k_i \in \mathbb Z\\ k_1 + \cdots + k_i = k }}  \widetilde{J}^{a_1}_{k_1} \cdots \widetilde{J}^{a_i}_{k_i} \,.
\] 
From this representation, after shifting the modes $\widetilde{W}^i_0$ to $\widetilde{W}^i_0 +(-1)^{i+1} e_i(\mathbf{P}) $, the authors of \cite{CDO24} construct a shifted Airy structure. More precisely, applying the fundamental theorem of Airy structures, they obtain the following result.
\begin{thm}\label{thm:genZ}
	Assume that $Q_1,\ldots, Q_{r-1} $ are pairwise distinct. There exists a unique  $\ket{\Gamma_{\Lambda}^{\text{CDO}}} \in \widetilde{\mathcal{V}}_{\Lambda}$ of the form
	\begin{equation}
	\label{CDOfor}
	\ket{\Gamma_{\Lambda}^{\text{CDO}}} =\exp\left( \sum_{(g,n) \in \mathbb{Z}_{\geq 0} \times \mathbb{Z}_{> 0}}   \frac{\hslash^{g-1}}{n!} \sum_{\substack{a_1,\ldots,a_n \in [r] \\ k_1,\ldots,k_n \in \mathbb{Z}_{> 0}}} \Phi_{g,n}\left[\begin{smallmatrix} a_1 & \cdots & a_n \\ k_1 & \cdots & k_n \end{smallmatrix}\right] \prod_{j = 1}^{n} \frac{J^{a_j}_{-k_j}}{k_j} \right),  
	\end{equation}
	satisfying the constraints
	\begin{equation}\label{eq:CDOconst}
		\forall (i,k) \in [r] \times \mathbb{Z}_{\geq 0} \qquad \widetilde{W}^i_k \ket{\Gamma_{\Lambda}^{\text{CDO}}}= (-1)^i e_i(\mathbf{P}) \delta_{k,0} \ket{\Gamma_{\Lambda}^{\text{CDO}}}.
	\end{equation}
	Moreover, the coefficients $ \Phi_{g,n}\left[\begin{smallmatrix} a_1 & \cdots & a_n \\ k_1 & \cdots & k_n \end{smallmatrix}\right]  $ belong to the ring $\mathbb C (\mathbf Q) [\mathbf P][\![\Lambda^r]\!]$. In particular, the $\Phi_{g,n}$ are formal \emph{power} series in $\Lambda^r$ (not just Laurent series).
\end{thm} 
\begin{proof} This is  \cite[Theorem 3.10 and Corollary 3.12]{CDO24} specialised to $L_i = (-1)^i e_i(\mathbf{P})$. 
\end{proof}

A few remarks are in order. The vector $\ket{\Gamma_{\Lambda}^{\text{CDO}}}$ is not the highest weight vector in the representation $\widetilde{\mathcal{V}}_\Lambda$, but rather a Whittaker-type vector. These Whittaker constraints \eqref{eq:CDOconst} bear on the action of all non-negative modes of the $\mathcal{W}$-algebra, while \eqref{WhitG2} involved only the positive modes. The $(i,k) = (1,0)$ constraint is trivially satisfied since
\[
\widetilde{W}_0^i =  \sum_{a = 1}^{r} \widetilde{J}^a_0 = - \sum_{a = 1}^{r} P_a = -e_1(\mathbf{P}),
\]
but the other ones are non-trivial and determine $\ket{\Gamma_{\Lambda}^{\text{CDO}}}$.  Besides, $\ket{\Gamma_{\Lambda}^{\text{CDO}}}$ has unstable terms as well, i.e., non-zero terms $\Phi_{0,1}$ and $\Phi_{0,2}$. Then, we can construct correlators
\begin{equation}\label{eq:phidef2}
\begin{split}
	\Phi_{g,n}(z_1,\ldots,z_n) & = \sum_{k_1,\ldots,k_n \in \mathbb{Z}_{> 0}} \Phi_{g,n}\big[\begin{smallmatrix} \mathfrak{c}(z_1) & \cdots & \mathfrak{c}(z_n) \\ k_1 & \cdots & k_n \end{smallmatrix}\big] \prod_{j = 1}^{n} \frac{\dd x(z_j)}{x(z_j)^{k_j + 1}} \\
	& \quad + \delta_{g,0}\delta_{n,1} \left( \widetilde{J_0}^{\mathfrak{c}(z_1)} \frac{\dd x(z_1)}{x(z_1)} + (-\Lambda)^{-r} \dd x(z_1) \right)+ \delta_{g,0}\delta_{n,2} \delta_{\mathfrak{c}(z_1),\mathfrak{c}(z_2)} \frac{\dd x(z_1)\dd x(z_2)}{(x(z_1) - x(z_2))^2}.
\end{split}
\end{equation}
As in \eqref{phidef}, they are germs of meromorphic $n$-differentials in the $n$-th product of the formal neighbourhood of $L \cong x^{-1}(\infty) = \bigsqcup_{a = 1}^{r} \{\infty_a\}$ in $C$, where $C$ is the same unramified curve of degree $r$ that appears in \cref{S123} in the context of the Gaiotto vector. More precisely, the $\phi_{g,n}$ for $2g - 2 + n > 0$ are germs of holomorphic $n$-differentials, $\phi_{0,1}(z)$ is the germ of  a meromorphic differential having a simple pole with residue given by the scalar $\widetilde{J}^a_0$ at $z = \infty_a$, and $\phi_{0,2}$ is the germ of a meromorphic bi-differential with a double pole on the diagonal.

\vspace{0.5cm}

\subsection{The spectral curve}
\label{sec:Hsc}

\medskip

As in the case of the Gaiotto vector $\ket{\Gamma_\Lambda}$, we can prove that the correlators $\phi_{g,n}$ defined in \eqref{eq:phidef2} from $\ket{\Gamma^{\text{CDO}}_\Lambda}$ analytically continue to meromorphic multi-differentials on a certain family of spectral curves and that the latter are computed by topological recursion on this spectral curve. We first describe the relevant family of spectral curves.

Assuming that the $Q_1,\ldots, Q_{r-1} $ are pairwise disjoint, we look at the  locus $\mathcal{S} \subset \mathbb P^1_x \times \mathbb P^1_y \times \mathbb C^*_\Lambda $ cut out by the equation
\begin{equation}\label{eq:curve}
	\prod_{a = 1}^{r} \left(\frac{P_a}{x} + y\right) + \frac{1}{\Lambda^r} \prod_{a = 1}^{r -1} \left(\frac{Q_a}{x} - y\right) = 0.
\end{equation} 
The map $\mathcal{S} \to \mathbb C^*_\Lambda $ defines an analytic family of algebraic curves. In particular, the fibre $\mathcal S_\Lambda$ over any fixed $\Lambda \in \mathbb C^*$ is a smooth genus $0$ curve admitting the following uniformisation by $\zeta \in \mathbb{P}^1$
\begin{equation}\label{eq:algcurve}
\left\{\begin{array}{lll} x(\zeta) & = & - \Lambda^r\, \dfrac{\prod_{a=1}^r (P_a +\zeta)}{\prod_{a=1}^{r-1} (Q_a-\zeta)}, \\[12pt] y(\zeta) & =& \dfrac{\zeta}{x(\zeta)} = - \dfrac{\zeta}{\Lambda^r}\,\dfrac{\prod_{a=1}^{r-1} (Q_a-\zeta)}{\prod_{a=1}^r (P_a +\zeta)}.
\end{array}\right.
\end{equation}

\begin{defn} 
\label{Rset}	Let $\mathcal{R}$ be the set of tuples $(P_1,\ldots,P_r,Q_1,\ldots,Q_{r - 1}) \subset \mathbb C^{2r-1}$ such that
	\begin{itemize}
		\item $Q_1,\ldots,Q_{r - 1} $ are pairwise distinct;
		\item $Q_a \neq P_b $ for any $a \in [r-1]$ and $b \in [r]$;
		\item The branched covering defined by $x : \mathcal S_\Lambda \to \mathbb P^1$ has only simple ramification points.
	\end{itemize}
\end{defn}
If $(\mathbf{P},\mathbf{Q}) \in \mathcal{R}$, then $x$ always has degree $r$ on $\mathcal S_{\Lambda}$ and $x=\infty$ is  not a  branch point. We complete the description of the spectral curve by defining as usual $ \omega_{0,1} $ and $\omega_{0,2}$ on a fibre $\mathcal S_\Lambda $ as 
\[
\omega_{0,1}(\zeta) = y(\zeta)\dd x(\zeta), \qquad \omega_{0,2}(\zeta_1,\zeta_2) = \frac{\dd \zeta_1 \dd \zeta_2}{(\zeta_1-\zeta_2)^2}\,.
\]
Then, we have the following lemma, which shows that the correlator $\phi_{0,1}$ analytically continues to the meromorphic differential $\omega_{0,1}$ on the curve $\mathcal S_\Lambda$.
\begin{lem}\label{lem:01new}
Assume $(\mathbf{P},\mathbf{Q}) \in \mathcal{R}$. The  all-order series expansion of the meromorphic $1$-form $\omega_{0,1}(\zeta)$ on $\mathcal S_{\Lambda}$ as $\zeta $ is near $ x^{-1}(\infty) \cong L$, with $1/x(\zeta)$ used as local coordinate, and then the all-order series expansion as $\Lambda \rightarrow 0$, is given by $\phi_{0,1}$. In other words:
\[
\omega_{0,1}(\zeta) \approx \phi_{0,1} \big(\begin{smallmatrix} \mathfrak{c}(\zeta) \\ x(\zeta) \end{smallmatrix}\big).
\]
\end{lem}
\begin{proof} As the proof closely follows the proof of \cref{lem:01} for the half Seiberg--Witten  curve, we will be brief. By extracting the $(0,1)$-terms from the constraints of \cref{thm:genZ} we get the following expression:
\[
	\prod_{a=1}^r \left(u + \frac{\phi_{0,1} \big(\begin{smallmatrix} a \\ x \end{smallmatrix}\big)}{\dd x} \right) = \sum_{i=0}^r u^{r-i} \left((-1)^i \frac{e_i(\mathbf{P})}{x^i} + O\left(\frac{1}{x^{i-1}}\right) \right).
\] 
The terms of order $O\left(\frac{1}{x^{i-1}}\right)$ can only come from the first two terms appearing in the definition \eqref{eq:phidef2} of $\phi_{0,1}$, i.e., the terms of order $x^{0}\dd x$ and $x^{-1}\dd x$. This gives 
\[
\prod_{a=1}^r \left(u + \frac{ \phi_{0,1} \big(\begin{smallmatrix} a \\ x \end{smallmatrix}\big)}{\dd x} \right) = \sum_{i=0}^r u^{r-i} \left( (-1)^i\frac{e_i(\mathbf{P})}{x^i}  + (-1)^{r}\frac{e_{i-1}(\mathbf{Q})}{\Lambda^r x^{i-1}}\right).
\]
If we substitute $u = -  \frac{\phi_{0,1}\big(\begin{smallmatrix} a \\ x \end{smallmatrix}\big)  }{\dd x}$ for any $a\in[r]$, we get the following algebraic equation for $\phi_{0,1}$
\begin{equation}
\sum_{i=0}^r \left(\phi_{0,1}\big(\begin{smallmatrix} a \\ x \end{smallmatrix}\big)  \right)^{r-i}\left( (-1)^r \frac{e_i(\mathbf{P})}{x^i}  +(-1)^i \frac{e_{i-1}(\mathbf{Q})}{\Lambda^r x^{i-1}}\right)(\dd x)^i = 0, 
\end{equation}
Doing the sum over $i$ we get
\[
(-1)^{r} \prod_{b = 1}^{r} \left(\phi_{0,1}\big(\begin{smallmatrix} a \\ x \end{smallmatrix}\big) + P_b \frac{\dd x}{x}\right) - \frac{1}{\Lambda^r} \prod_{b = 1}^{r - 1} \left(\phi_{0,1}\big(\begin{smallmatrix} a \\ x \end{smallmatrix}\big)  - Q_b \frac{\dd x}{x}\right) = 0,
\]
which matches the equation for $\omega_{0,1}$ from the definition of the family \eqref{eq:curve}. Thus, $\omega_{0,1}$ is the analytic continuation of $\phi_{0,1}$ to $\mathcal S_\Lambda$.
\end{proof}

Let us turn to the other unstable case where $(g,n) = (0,2)$.

\begin{lem}\label{lem:02new}
Assume $(\mathbf{P},\mathbf{Q}) \in \mathcal{R}$. The  all-order series expansion of the bi-differential $\omega_{0,2}(\zeta_1,\zeta_2)$ on $\mathcal S_{\Lambda}$ when $\zeta_1,\zeta_2 $ is near $ x^{-1}(\infty) \cong L$, with $1/x(\zeta_i)$ used as local coordinate, and then the all-order series expansion as $\Lambda \rightarrow 0$, is given by $\phi_{0,2}$. In other words:
\[
\omega_{0,2}(\zeta_1,\zeta_2) \approx \phi_{0,2} \big(\begin{smallmatrix} \mathfrak{c}(\zeta_1) & \mathfrak{c}(\zeta_2) \\ x(\zeta_1) & x(\zeta_2) \end{smallmatrix}\big).
\]
\end{lem}
\begin{proof}
	The proof closely follows the proof for the half Seiberg--Witten  curve in \cref{lem:02}. In the first half of the proof of  \cref{lem:02} we showed that the $\phi_{0,2}$ admits an analytic continuation by finding an explicit formula for the $\phi_{0,2}$ in terms of the $\phi_{0,1}$ and $x(z)$. An analogue of this explicit formula \eqref{eq:02explicit} holds in this case as well, although the proof needs to be slightly modified. As the analogue of the constraints \eqref{constrg0n1} has $O\left(\frac{\left(dx(z)\right)^{i}}{x(z)^{i-1}}\right)$ on the right-hand side (instead of $O\left(\frac{\left(dx(z)\right)^{i}}{x(z)^{i}}\right)$), we need to take the projection $p_{\geq -(i-1)}$ in \eqref{eq:2.11}. The rest of the proof goes through with minor changes, and we get the analytic continuation $\tilde{\phi}_{0,2}$, defined as  
	\begin{equation}\label{eq:02explicitnew}
		\tilde{\phi}_{0,2}(\zeta_1,\zeta_2) = \prod_{\zeta_1'  \in \mathfrak{f}'(\zeta_1)}\frac{1}{ \zeta_1'- \zeta_1} \cdot \dd_{\zeta_2}\left(  \frac{\dd x(\zeta_1) }{(x(\zeta_1)-x(\zeta_2))} \prod_{\zeta_2'  \in \mathfrak{f}'(\zeta_2) }  \left( \zeta_2' - \zeta_1 \right)  \right).
	\end{equation}
	Finally, let us evaluate $\tilde{\phi}_{0,2} $ explicitly. Using the notation $D(\zeta) = \prod_{a = 1}^{r - 1} (Q_a - \zeta)$ and $N(\zeta) = \prod_{a = 1}^{r} (P_a + \zeta)$, we have $x(\zeta) = -\Lambda^r N(\zeta)/D(\zeta)$. As $N(\zeta)$ is a polynomial of degree $r$ with leading term $\zeta^r$ and $D$ is a polynomial of degree $r - 1$ with leading term $(-1)^{r - 1}\zeta^{r - 1}$, we find
	\begin{equation*}
		\begin{split}
			x(\zeta_1) - x(\zeta_2) & = \frac{\Lambda^r \big(\frac{N(\zeta_2)}{D(\zeta_2)} D(\zeta_1) - N(\zeta_1)\big)}{D(\zeta_1)} = \frac{-(-\Lambda)^{r}}{D(\zeta_1)} \prod_{\zeta_2' \in \mathfrak{f}(\zeta_2)} (\zeta_2' - \zeta_1), \\ 
			\dd x(\zeta_1) & = \dd \zeta_1 \lim_{\zeta_2 \rightarrow \zeta_1} \frac{x(\zeta_1) - x(\zeta_2)}{\zeta_1 - \zeta_2} = \dd \zeta_1 \frac{(-\Lambda)^{r}}{D(\zeta_1)} \prod_{\zeta_1' \in \mathfrak{f}'(\zeta_1)} (\zeta_1' - \zeta_1). 
		\end{split}
	\end{equation*}
	Thus, we can evaluate \eqref{eq:02explicitnew} as 
	\[
	\tilde{\phi}_{0,2}(\zeta_1,\zeta_2) =  \frac{ (-\Lambda)^r \dd \zeta_1 } {D(\zeta_1) \dd x(\zeta_1)} \dd_{\zeta_2}\left(  \dd x(\zeta_1) \frac{D(\zeta_1) }{-(-\Lambda)^r (\zeta_2-\zeta_1)}   \right) = \frac{\dd \zeta_1  \dd \zeta_2}{(\zeta_1 - \zeta_2)^2}. \qedhere
	\]
\end{proof}

\vspace{0.5cm}

\subsection{Stable correlators and topological recursion}

\medskip

\noindent \textit{Proof of \cref{th:main2}.} In view of \cref{lem:01new} and \cref{lem:02new}, it remains to show that the stable correlators $\phi_{g,n}$ of \eqref{eq:phidef2} admit an analytic continuation on the CDO spectral curve \eqref{eq:curve}, which has poles only at the ramification points and satisfies the abstract loop equations. This is done exactly as in the proof \cref{prop:stablecorrelators} for the analytic continuation, and \cref{prop:ALE} for the location of the poles and the abstract loop equations: the proofs indeed only used the general structure of the Whittaker constraints and is not  affected by the form of the spectral curve. \hfill $\Box$

\begin{rem}\label{rem:limits}
	Note that the CDO spectral curve \eqref{eq:algcurve} appearing in this section is a generalisation of the half Seiberg--Witten  curve \eqref{eq:Gcurve}. Indeed, we can take the limit $ P_1,\ldots, P_r  \rightarrow
	 \infty$ and $\Lambda \rightarrow 0$ such that $\Lambda^r \cdot P_1 \cdots P_r \rightarrow (\Lambda')^r $ in the CDO curve to recover the half Seiberg--Witten  curve at energy scale $\Lambda'$. In fact, this limit falls into the class of allowed limits based on the results of \cite{TRlimits}. This means that the limit of  the correlators $\omega_{g,n}$ constructed by topological recursion on the curve \eqref{eq:algcurve}  matches the correlators $\omega_{g,n}$ constructed by topological recursion on the  \eqref{eq:Gcurve}. However, it is not clear how to take the above limit  of the $\mathcal{W}(\mathfrak{gl}_r)$-module defined by \eqref{eq:Jrep} and the Whittaker conditions  in \cref{thm:genZ} characterizing $\ket{\Gamma^{\text{CDO}}_\Lambda}$  directly, in order to get the Whittaker conditions \eqref{WhitG2} defining the Gaiotto vector.
	 \end{rem}
 
 \vspace{1cm}
 
\section{Consequences}
\label{S5}

\medskip

Starting from \cref{th:main} and \cref{th:main2}, we can exploit the theory of the topological recursion on ramified spectral curves to derive either properties and new interpretations of the correlators, or even, directly of the Whittaker vectors under consideration. Here we focus on:  two expressions in terms of intersection numbers on $\overline{\mathcal{M}}_{g,n}$, a relation to weighted Hurwitz numbers, quantum curves, and a discussion on free energies.

\vspace{0.5cm}

\subsection{Relation to intersection theory on \texorpdfstring{$\overline{\mathcal{M}}_{g,n}$}{Mgnbar}}
\label{S51}
\medskip

\subsubsection{Intersection theory and topological recursion}
\label{S511}
\medskip

The correlators $(\omega_{g,n})_{g,n}$ produced by the topological recursion on a spectral curve specified as $(S,x,\omega_{0,1},\omega_{0,2})$ can be quite generally expressed in terms of intersection theory on $\overline{\mathcal{M}}_{g,n}$ \cite{Einter,DOSS,BKS23}. For our purposes, it is sufficient to summarise the theory for spectral curves for which $\dd x$ is meromorphic on a compact Riemann surface and only has simple zeros at which $\dd y$ has neither a zero nor a pole\footnote{We also assume that the ramification points taken into account in \eqref{TRome} occur at finite values of $x$. One can reduce to this case by using a twist $(x,y) \mapsto (x^{-1},-yx^2)$, see paragraph D.}. The master formula reads
\begin{equation}
\label{finalformin}
\omega_{g,n}(z_1,\ldots,z_n) = \sum_{\substack{\substack{\rho_1,\ldots,\rho_n \in \text{Ram}(S) \\ \lambda_1,\ldots,\lambda_n \in \operatorname{Ram}(S) \\ m_1,\ldots,m_n \geq 0}}} \left(\int_{\overline{\mathcal{M}}_{g,n}} \Omega_{g,n;\rho_1,\ldots,\rho_n} \prod_{i = 1}^{n} \psi_i^{m_i} R_{\rho_i,\lambda_i}(\psi_i)\right) \prod_{i = 1}^{n} \dd \Xi_{\lambda_i,m_i}(z_i).
\end{equation} 
It involves two ingredients constructed from the spectral curve: a basis of meromorphic $1$-forms $\dd\Xi_{\lambda,m}$ on which we decompose the correlators, and a collection of tautological classes $\Omega_{g,n;\rho_1,\ldots,\rho_n} \in H^{\bullet}(\overline{\mathcal{M}}_{g,n})$ that we call the \emph{TR class} (the $R_{\rho,\lambda}$ factor will be defined along with them). For $i \in [n]$, $\psi_i \in  H^{2}(\overline{\mathcal{M}}_{g,n})$ denotes the standard $\psi$-class which is defined as the first Chern class of the  cotangent line bundle at the $i$-th marked point.  In this setting, the TR class is known to form a semi-simple (perhaps without unit) cohomological field theory  as it can be constructed by a certain Givental action on the trivial cohomological field theory.

\medskip

\noindent \textbf{A. The basis of differentials}

Let us fix a choice of square root $\eta_{\lambda}(z) = \sqrt{2(x(z) - x(\lambda))}$, giving a local coordinate near $\lambda \in \operatorname{Ram}(S)$. The meromorphic $1$-forms are defined by induction on $m \in \mathbb{Z}_{\geq 0}$. We set
\begin{equation}
\label{descdiff}
\dd\Xi_{\lambda,0}(z) = \Res_{z' = \lambda} \frac{\omega_{0,2}(z,z')}{\eta_{\lambda}(z)}, \qquad \dd\Xi_{\lambda,m + 1} = -\dd\left(\frac{\Xi_{\lambda,m}}{\dd x}\right).
\end{equation}
The primary differential $\dd \Xi_{\lambda,0}$ has a double pole at $\lambda$ only, the descendent differential $\dd \Xi_{\lambda,k}$ has a pole of order $(2k + 2)$ at $\lambda$ and poles at $\lambda' \in \operatorname{Ram}(S) \setminus \{\lambda\}$ of order at most $2k$.

\medskip

\noindent \textbf{B. The TR class}

The tautological classes are constructed from formal Laplace transforms using the spectral curve data. Recalling that $\omega_{0,1} = y \dd x$, we introduce for $\rho,\lambda \in \operatorname{Ram}(S)$
\begin{equation}
\begin{split}
T_{\rho}(u) & = \frac{1}{\sqrt{2\pi\,u}} \int_{\gamma_{\rho}} e^{(x(\rho) - x(z))u^{-1}} \dd y(z), \\
R_{\rho,\lambda}(u) & =  \frac{1}{\sqrt{2\pi\,u^{-1}}} \int_{\gamma_{\rho}} e^{(x(\rho) - x(z))u^{-1}} \dd \Xi_{\lambda,0}(z),
\end{split}
\end{equation}
where $\gamma_{\rho}$ are steepest descent paths for $\text{Re}(x/u)$. We are only interested in the definition of $T_{\rho}(u)$ and $R_{\rho,\lambda}(u)$ as formal power series in $u$: this is only sensitive to the germ of $\gamma_{\rho}$ around $\rho$, and for this we can take $u > 0$ and take $x(z) - x(\rho) \in \mathbb{R}_{\geq 0}$ along $\gamma_{\rho}$ (for $R_{\rho,\rho}$ we slightly push this contour off the pole at $\rho$). The orientation is chosen consistently  with the choice of square root so that we have
\begin{equation}
\label{orcon}
\forall k \in \mathbb{Z}_{\geq 0}\qquad \frac{1}{\sqrt{2\pi u}} \int_{\gamma_{\rho}} e^{(x(\rho) - x(z))u^{-1}} (\eta_{\rho}(z))^{2k} \dd \eta_{\rho}(z) = -(2k - 1)!! u^k,
\end{equation}
with the convention $(-1)!! = 1$. The assumptions on the spectral curves are known to imply that 
\begin{equation}
\label{adminRT} T_{\rho}(0) \neq 0 \qquad \text{and}\qquad \sum_{\lambda \in \operatorname{Ram}(S)} R_{\rho_1,\lambda}(u) R_{\rho_2,\lambda}(-u) = \delta_{\rho_1,\rho_2}.
\end{equation}
Then, it makes sense to introduce the formal power series
\begin{equation}
\label{eq53}
\begin{split}
\sum_{m \geq 0} t_{\rho,m} u^m & = -\ln T_{\rho}(u), \\
\qquad B_{\rho_1,\rho_2}(u_1,u_2) & = \frac{\delta_{\rho_1,\rho_2} - \sum_{\lambda \in \operatorname{Ram}(S)} R_{\rho_1,\lambda}(u_1) R_{\rho_2,\lambda}(u_2)}{u_1 + u_2}.
\end{split}
\end{equation}

The tautological class appearing in \eqref{finalformin} is then obtained in two steps. First, one constructs from $T_\rho(u)$ a cohomology class indexed by a single $\rho \in \operatorname{Ram}(S)$
\[
\Upsilon_{g,n;\rho} = \exp\left(\sum_{m \geq 0} t_{\rho,m} \kappa_m\right) \in H^{\bullet}(\overline{\mathcal{M}}_{g,n}).
\]
Since $\kappa_0 = (2g - 2 + n) \in H^0(\overline{\mathcal{M}}_{g,n})$, the determination of the logarithm is irrelevant in \eqref{eq53}.

For the second step, we recall that for each stable graph $G$ of genus $g$ with $n$ labelled leaves, we have an inclusion of boundary strata
\[
\jmath_{G} : \prod_{\text{vertex}\,\,v} \overline{\mathcal{M}}_{g(v),n(v)} \hookrightarrow \overline{\mathcal{M}}_{g,n}.
\]
For each half-edge $h$ incident at a vertex $v$, we have a corresponding $\psi_h \in H^2(\overline{\mathcal{M}}_{g(v),n(v)})$. We consider the set $\text{Stab}_{g,n}(\rho_1,\ldots,\rho_n)$ of stable graphs of genus $g$, whose half-edges $h$ are decorated by $\rho(h) \in \operatorname{Ram}(S)$ such that all half-edges incident to the same vertex $v$ have the same decoration (denoted $\rho(v)$), and with $n$ labelled leaves carrying the decorations $\rho_1,\ldots,\rho_n$.
\begin{equation}
\label{TRclasscon}
\Omega_{g,n;\rho_1,\ldots,\rho_n} = \sum_{G \in \text{Stab}_{g,n}} \frac{1}{\# \text{Aut}(G)} (\jmath_{G})_*\left( \prod_{\text{vertex}\,\,v} \Upsilon_{g(v),n(v);\rho(v)} \prod_{\text{edge}\,\,\{h,h'\}} B_{\rho(h),\rho(h')}(\psi_h,\psi_{h'})\right). 
\end{equation}
This definition makes sense for any power series $T_{\rho}(u)$ and $R_{\rho_1,\rho_2}(u)$ with indices $\rho$ in a given set and satisfying the admissibility condition \eqref{adminRT}, even if they do not come from the Laplace transform of a spectral curve data. In fact, this formula can be viewed as a special case of the Givental action on CohFTs \cite{GivVir}.
\medskip

\noindent \textbf{C. Expansion coefficients}

From the master formula \eqref{finalformin}, we can easily extract the coefficients of the series expansion of the $\omega_{g,n}$ near any point: it suffices to know the series expansion for the basis of $1$-forms. Following \cite[Section 7.5.5]{BKS23}, we discuss the case of expansion near a pole of order $d_p$ (or zero of order $-d_p$) of $\dd x$, using a local coordinate $X$ centred near $p$ such that
\begin{equation}
\label{Xcoo}
\dd x = \left\{\begin{array}{lll} c_p \dfrac{\dd X}{X} && \text{if}\,\,d_p = 1; \\[12pt] \dfrac{\dd X}{X^{d_p}} && \text{if}\,\,d_p \neq 1. \end{array}\right.
\end{equation}
The only task is in fact to compute the series expansion of the primary differentials as $z \rightarrow p$
\[
\dd \Xi_{\lambda,0}(z) \approx \dd\left(\sum_{k \geq 1} \mathsf{S}_{\lambda,0}\big[\begin{smallmatrix} p \\ k \end{smallmatrix}\big] \frac{X^k}{k}\right).
\]
Then, by construction of the descendent differentials in \eqref{descdiff}, we have  the following expansion  \cite[Lemma 7.31]{BKS23}:
\[
\dd \Xi_{\lambda,m}(z) \approx \dd \left(\sum_{k \geq 1} \mathsf{S}_{\lambda,m}\big[\begin{smallmatrix} p \\ k \end{smallmatrix}\big] \frac{X^k}{k}\right).
\]
If $d_p = 1$, we have a simple expression
\begin{equation}
\label{the55} \mathsf{S}_{\lambda,m}\big[\begin{smallmatrix} p \\ k \end{smallmatrix}\big] = (-k/c_p)^m \cdot \mathsf{S}_{\lambda,0}\big(\begin{smallmatrix} p \\ k \end{smallmatrix}\big).
\end{equation}
while if $d_p \neq 1$, we have
\begin{equation}
\label{Shigher}
\mathsf{S}_{\lambda,m}\big[\begin{smallmatrix} p \\ k \end{smallmatrix}\big] = k(k - (d_p - 1)) \cdots (k - (m - 1)(d_p - 1))  \cdot \mathsf{S}_{\lambda,0}\big[\begin{smallmatrix} p \\ k - m(d_p - 1) \end{smallmatrix}\big],
\end{equation}
with the convention that $\mathsf{S}_{\lambda,m}\big[\begin{smallmatrix} p \\ k \end{smallmatrix}\big] = 0$ for $k \leq 0$.

\medskip

\noindent \textbf{D. Twisted TR class}

The topological recursion \eqref{TRome} only depends on the data of $x$ through the ramification points and the local involution.
As a result, we can find transformations of the spectral curve that have no effect on the $(\omega_{g,n})_{g,n}$ computed by the topological recursion: we call them twists. Given a spectral curve $(S,x,\omega_{0,1},\omega_{0,2})$ with $\omega_{0,1} = y \dd x$, examples of twists are $(\tilde{x},\tilde{y}) = (f(x),y/f'(x))$, where $f(x)$ is such that $f'(x)$ is a rational function that does not vanish at branch points of $x$ --- observe that $\omega_{0,1} = y \dd x = \tilde{y} \dd \tilde{x}$. Since they affect $x$ and therefore $T(u)$ and $R(u)$, twists can radically affect the basis of differentials and the TR class, and this leads to many different intersection-theoretic representations of the same correlators.

The twist $(\tilde{x},\tilde{y}) = (\ln x,xy)$ has the interesting property that it converts poles of $\dd x$ into simple poles of $\dd \tilde{x} = \frac{\dd x}{x}$. Then,  we are in the simplest case to express the expansion coefficients of $\omega_{g,n}$ near poles of $x$ in terms of  intersection indices. Indeed, consider a spectral curve $(S,x,\omega_{0,1},\omega_{0,2})$ such that $x$ is meromorphic on a compact Riemann surface, and denote $o_p$ the order of a pole $p$ of $x$. The local coordinate realising \eqref{Xcoo}  near this pole is $X = (-o_p x)^{-1/o_p}$. Let  $(\omega_{g,n})_{g,n}$ be the correlators computed by topological recursion. Then, let $\check{T}(u)$,$\check{R}(u)$ the formal series, $\check{\Omega}_{g,n;\rho_1,\ldots,\rho_n}$ the TR class, and $\check{\mathsf{S}}_{\lambda,0}$ the expansion coefficients of the primary differentials for the spectral curve $(S,\ln x,\omega_{0,1},\omega_{0,2})$. With the help of \eqref{the55}, we get the following expansion for the correlators as $z_i$ approaches a pole $p_i$ of $x$:
\begin{multline}
\label{omint2}
 \omega_{g,n}(z_1,\ldots,z_n) \\
 \approx \dd_1 \cdots \dd_n \left[ \sum_{\substack{\rho_1,\ldots,\rho_n \in \operatorname{Ram}(S) \\ \lambda_1,\ldots,\lambda_n \in \operatorname{Ram}(S) \\ k_1,\ldots,k_n \geq 1}} \left(\int_{\overline{\mathcal{M}}_{g,n}}  \check{\Omega}_{g,n;\rho_1,\ldots,\rho_n} \prod_{i = 1}^{n} \frac{\check{R}_{\rho_i,\lambda_i}(\psi_i)}{1 - o_{p_i}^{-1} k_i\psi_i}\right) \prod_{i = 1}^{n} \check{\mathsf{S}}_{\lambda_i,0}\big[\begin{smallmatrix} p_i \\ k_i\end{smallmatrix}\big] \frac{X_i^{k_i}}{k_i}\right]. 
\end{multline}

\medskip

\subsubsection{Properties and remarkable TR classes}

\medskip

In this section we consider situations where $R$ and $T$ have no indices, i.e., ones corresponding to spectral curves with a single ramification point or  to CohFTs of rank $1$ . We mention properties of the construction of \cref{S511} and examples of TR classes that have a geometric meaning and that will be used to study Gaiotto vectors in \cref{S513}.

\medskip

\noindent \textbf{A. Multiplicativity}

Here we consider situations without indices. Consider formal power series $T(u)$ and $R(u)$ such that
\begin{equation}
\label{compati} T(0) \neq 0 \qquad \text{and} \qquad  R(u)R(-u) = 1.
\end{equation}
Then, we can write
\[
T(u) = \exp\left(-\sum_{m \geq 0} t_m u^m\right) \qquad R(u) = \exp\big(\textsf{r}(u)\big),
\]
where $\mathsf{r}(u)$ is an odd formal power series. The combinatorics of self-intersections of boundary strata in $\overline{\mathcal{M}}_{g,n}$ imply that the TR class \eqref{TRclasscon} can be written
\begin{equation}
\label{MultiTRclass} \Omega_{g,n} \cdot \prod_{i = 1}^{n} R(\psi_i) = \exp\left(\sum_{m \geq 0} t_m \kappa_m + \sum_{i = 1}^{n} \mathsf{r}(\psi_i) - \frac{1}{2} \sum_{\Delta} (\jmath_{\Delta})_* \frac{\mathsf{r}(\psi') + \mathsf{r}(\psi'')}{\psi' + \psi''}\right),
\end{equation}
where the sum ranges over boundary divisors of $\overline{\mathcal{M}}_{g,n}$, $\jmath_{\Delta}$ is the natural inclusion map  and $\psi',\psi''$ are the $\psi$-classes on the two sides of the node --- see e.g. \cite[Lemma 3.10]{the7}.

This leads to an interesting multiplicativity property of TR classes. Imagine that we have formal power series $T^{(i)}(u),R^{(i)}(u)$ for $i = 1,2$ satisfying \eqref{compati}, and we have constructed the corresponding TR classes $\Omega_{g,n}^{(i)}$. Then, the TR class associated to the products $T(u) = T^{(1)}(u)\cdot T^{(2)}(u)$ and $R(u) = R^{(1)}(u) \cdot R^{(2)}(u)$ is the product TR class
\[
\Omega_{g,n} = \Omega_{g,n}^{(1)} \cdot \Omega_{g,n}^{(2)} \in H^{\bullet}(\overline{\mathcal{M}}_{g,n}).
\]

\medskip

\noindent \textbf{B. The deformed $\Theta$ class}

The \textit{Theta class} $\Theta_{g,n}$,  introduced  by Norbury in \cite{Nor23} based on the work of Chiodo \cite{Chi08}, is constructed using the Euler class of a vector bundle on  the moduli space of twisted $2$-spin curves $\overline{\mathcal M}^{(2)}_{g,n}$. We do not dwell on the details of the construction here, but we note the following key properties: 
	\begin{itemize}
		\item $\Theta_{g,n} \in H^{2(2g-2+n)}(\overline{\mathcal M}_{g,n})$ for any $2g - 2 + n > 0$;
		\item The family $(\Theta_{g,n})_{g,n}$ forms a non-semisimple CohFT without a flat unit.
	\end{itemize} These properties, and an explanation of what they mean, can be found in \cite{Nor23, CGG22}. A specific deformation of the Theta class $\Theta_{g,n}[\epsilon]$, depending on a parameter $\epsilon^2 \in \mathbb C$ was constructed and studied in \cite{CGG22}. The deformed Theta class satisfies the following properties:
		\begin{itemize}
		\item $\Theta_{g,n}[\epsilon = 0] = \Theta_{g,n}$.
		\item $\Theta_{g,n}[\epsilon]  \in H^{\bullet}(\overline{\mathcal M}_{g,n})$ is a polynomial of degree $2g - 2 + n$ in $\epsilon^2$, and the coefficient of $\epsilon^{2m}$ belongs to $H^{2(2g - 2 + n - m)}(\overline{\mathcal{M}}_{g,n})$.
		\item The family $(\Theta_{g,n}[\epsilon])_{g,n}$ forms a semisimple CohFT without a flat unit, for any $\epsilon \neq 0$.
	\end{itemize}
If we introduce
	\[
	\exp\left( -\sum_{m \geq 1} s_m u^m \right) = \sum_{m\geq 0} (-1)^m (2m+1)!!\, u^m,
	\] 
an explicit formula for the deformed Theta class, found in \cite[Corollary 3.25]{CGG22} is
\[
\Theta_{g,n}[\epsilon] = (-\epsilon^2)^{2g-2+n} \exp\left( \sum_{m>0} s_m (-\epsilon^2)^{-m} \kappa_m \right) \in H^{\bullet}(\overline{\mathcal{M}}_{g,n}).
\]
This formula has the required properties, in particular it is polynomial in $\epsilon$ for each $g,n$ due to certain tautological relations between $\kappa$ classes, anticipated in \cite{KN24}. In other words, $\Theta^{\epsilon}_{g,n}$ is the TR class associated with
\begin{equation}
\label{TRtheta} T(u) = \sum_{m \geq 0} \frac{(2m + 1)!!}{(-\epsilon^2)^{m + 1}}\,u^m \qquad \text{and} \qquad R(u) = 1.
\end{equation}

\medskip

\noindent \textbf{C. The Hodge class}

\medskip 

The Hodge class $\boldsymbol{\Lambda}[\epsilon] = \sum_{i = 0}^{g} \epsilon^i \boldsymbol{\lambda}_i$ is the Chern polynomial of the Hodge bundle of holomorphic $1$-forms (we use bold symbols to avoid confusion with the letter $\Lambda$ used for the energy scale). It has the property that $\boldsymbol{\Lambda}[\epsilon]\boldsymbol{\Lambda}[-\epsilon] = 1 \in H^{0}(\overline{\mathcal{M}}_{g,n})$. Mumford \cite{Mumford} expressed it as
\begin{equation}
\label{Lambdae} \boldsymbol{\Lambda}[\epsilon] = \exp\left(\sum_{m \geq 1} \frac{B_{m + 1} \epsilon^m}{m(m + 1)}\left(\kappa_m -  \sum_{i = 1}^{n} \psi_i^m + \sum_{\Delta} \frac{1}{2} (\jmath_{\Delta})_* \frac{(\psi')^{2m + 1} + (\psi'')^{2m + 1}}{\psi' + \psi''}\right)\right),
\end{equation}
where $(B_m)_{m \geq 1}$ are the Bernoulli numbers defined by the power series expansion
\[
\frac{t}{e^{t} - 1} = \sum_{m \geq 0} \frac{B_k}{k!} t^k.
\]
Only odd $m$ appear in \eqref{Lambdae} as $B_{k} = 0$ for odd $k \geq 3$. Recalling the Stirling expansion as $u \rightarrow 0$
\[
\Gamma(u^{-1}) \approx \sqrt{2\pi} \,e^{(u^{-1} - 1/2)\ln(u^{-1})  - u^{-1}} \Gamma_{\text{reg}}(u^{-1}) \qquad \text{with} \qquad \Gamma_{\text{reg}}(u^{-1}) = \exp\left(\sum_{m \geq 1} \frac{B_{m + 1} u^m}{m(m + 1)}\right).
\]
and comparing with \eqref{MultiTRclass}, we see that the Hodge class can be realised from the TR class as follows.
\[
T(u) = R(u) = \frac{1}{\Gamma_{\text{reg}}(\epsilon^{-1}u^{-1})} \qquad \mathop{\Longrightarrow}^{\text{TR class}} \qquad \Omega_{g,n} \cdot \prod_{i = 1}^{n} R(\psi_i) = \boldsymbol{\Lambda}[\epsilon].
\]
This was first observed in \cite{Einter}.

\medskip

\subsubsection{Gaiotto vector for $r = 2$}
\label{S513}
\medskip

If the spectral curve has a single ramification point, the intersection-theoretic representation are simpler as all indices $\rho,\lambda$ can be dropped (we keep them in $\dd \Xi$ and $\mathsf{S}$ but drop them from the other notations). This happens for the $r = 2$ half Seiberg--Witten curve\footnote{The half Seiberg--Witten curve has also a ramification point at $\zeta = \infty$, but it does not appear in the topological recursion of \cref{th:main} and can be ignored.}, and we now focus on this case. 

\begin{prop}
\label{pr:intGai2} For rank $r = 2$, the coefficients of the Gaiotto vector \eqref{Gformalexp} satisfy for any $(g,n) \in \mathbb{Z}_{\geq 0} \times \mathbb{Z}_{> 0}$ such that $2g-2+n > 0$ and $a_1,\ldots,a_n \in \{1,2\}$ and $k_1,\ldots,k_n \in \mathbb{Z}_{> 0}$ 
\begin{equation} 
\begin{split}
\label{theform1Phi}  \Phi_{g,n}\big[\begin{smallmatrix} a_1 & \cdots & a_n \\ k_1 & \cdots & k_n \end{smallmatrix}\big] &  \\
& = 2^{2g - 2 + n}  (Q_1 - Q_2)^{2 - 2g - n - 2(k_1 + \cdots + k_n)} \Lambda^{r(k_1 + \cdots + k_n)}  (-1)^{a_1 + \cdots + a_n} \\ 
&  \times \left(\,\,\sum_{m_1,\ldots,m_n \geq 0}  \int_{\overline{\mathcal{M}}_{g,n}} \Theta_{g,n}[1] \prod_{i = 1}^n  \frac{(-1)^{m_i}(2k_i + 2m_i)!}{2^{m_i}\,(k_i + m_i)!(k_i - 1)!} \psi_i^{m_i}\right) \\
& = 2^{3g - 3 + n} (Q_1 - Q_2)^{2 - 2g - n - 2(k_1 + \cdots + k_n)} \Lambda^{r(k_1 + \cdots + k_n)} (-1)^{a_1 + \cdots + a_n} \\
& \quad \times \int_{\overline{\mathcal{M}}_{g,n}} \exp\left(\sum_{m \geq 1} \frac{(-1)^{m + 1}}{m\,2^m} \kappa_m\right) \boldsymbol{\Lambda}[-1]^2  \boldsymbol{\Lambda}\big[\tfrac{1}{2}\big]  \prod_{i = 1}^n \frac{k_i {2k_i \choose k_i}}{1 - k_i \psi_i}.
\end{split}  
\end{equation}

\end{prop}

\begin{proof} Recall the half Seiberg--Witten  curve for $r = 2$:
\[
x(\zeta) = -\frac{\Lambda^2}{(Q_1 - \zeta)(Q_2 - \zeta)},\qquad y(\zeta) = \frac{\zeta}{x(\zeta)},\qquad \omega_{0,2}(\zeta_1,\zeta_2) = \frac{\dd \zeta_1 \dd \zeta_2}{(\zeta_1 - \zeta_2)^2}.
\]
The unique ramification point that is relevant in the topological recursion \cref{th:main} is located at $\rho = \frac{1}{2}(Q_1 + Q_2)$.

We first apply the twist $(\tilde{x},\tilde{y}) = (x^{-1},-x^2y)$. The advantage in doing so is that the Laplace transforms defining $R$ and $T$ reduce to Gaussian integrals with respect to $\zeta$. The associated primary differential is
\begin{equation}
\label{primartil}
\dd\tilde{\Xi}_{\rho,0}(\zeta) = \Res_{\zeta' = \rho} \frac{\dd \zeta \dd \zeta'}{(\zeta - \zeta')^2}\,\frac{1}{\sqrt{\tilde{x}''(\rho)}\,(\zeta' - \rho) + o(\zeta' - \rho)} = \frac{\sqrt{-\Lambda^r/2}\, \dd \zeta}{(\zeta - \rho)^2}.
\end{equation}
The function $\tilde{x}$ has a simple zero at $\zeta = Q_a$ (that is, $\tilde{d}_{Q_a} = 0$), so the local coordinate of \eqref{Xcoo} is $\tilde{X} = \tilde{x} = x^{-1}$. The expansion coefficient of the primary differentials using this coordinate are
\begin{equation*}
\begin{split}
\tilde{\mathsf{S}}_{\rho,0}\big[\begin{smallmatrix} Q_a \\ k \end{smallmatrix}\big] & = \Res_{\zeta = Q_a} \frac{\sqrt{-\Lambda^r/2} \,\dd \zeta}{(\zeta - \rho)^2} (\tilde{x}(\zeta))^{-k} \\
& = -\sqrt{-\Lambda^r/2} \Res_{\zeta = Q_a} \frac{k \tilde{x}'(\zeta)}{\zeta - \rho}\,(\tilde{x}(\zeta))^{-(k + 1)} \\
& = \sqrt{-2\Lambda^r} \Res_{\zeta = Q_a} \frac{k\,(-\Lambda^r)^{k}\,\dd \zeta}{((\zeta - Q_1)(\zeta - Q_2))^{k + 1}}  \\
& = (-1)^{a - 1}\,\frac{\sqrt{-2\Lambda^r}}{Q_1 - Q_2} \left(\frac{\Lambda^r}{(Q_{1} - Q_{2})^2}\right)^{k} \frac{(2k)!}{k!(k - 1)!},
\end{split}
\end{equation*}
From \eqref{Shigher} we deduce
\begin{equation*}
\begin{split}
\tilde{\mathsf{S}}_{\rho,m}\big[\begin{smallmatrix} Q_a \\ k \end{smallmatrix}\big] & = \frac{(k + m - 1)!}{(k - 1)!} \cdot \tilde{\mathsf{S}}_{\rho,0}\big[\begin{smallmatrix} Q_a \\ k + m \end{smallmatrix}\big] \\
& =  (-1)^{a - 1} \cdot \frac{\sqrt{-2\Lambda^r}}{Q_1 - Q_2} \left(\frac{\Lambda^r}{(Q_1 - Q_2)^2}\right)^{k + m} \frac{(2k + 2m)!}{(k + m)!(k - 1)!}.
\end{split}
\end{equation*}
Now we turn to the associated formal series $\tilde{T}(u)$ and $\tilde{R}(u)$. We compute
\begin{equation*}
\begin{split} 
\tilde{T}(u) & = \frac{1}{\sqrt{2\pi u}} \int_{\gamma_{\rho}} e^{(x^{-1}(\rho) - x^{-1}(\zeta))u^{-1}} \dd \tilde{y} \\
& = \frac{1}{u\sqrt{2\pi u}} \int_{\gamma_{\rho}} \exp\left(\frac{(\zeta - \rho)^2}{\Lambda^r u}\right) \tilde{y} \dd \tilde{x} \\ 
& = \frac{1}{u\sqrt{2\pi u} }\int_{\gamma_{\rho}} \exp\left(\frac{(\zeta - \rho)^2}{\Lambda^r u}\right) \frac{2 \zeta (\rho - \zeta) \dd \zeta}{(Q_1 - \zeta)(Q_2 - \zeta)} \\ 
& = - \frac{1}{u\sqrt{2\pi u}} \int_{\gamma_{\rho}} \exp\left(\frac{(\zeta - \rho)^2}{\Lambda^r u}\right)  \sum_{m \geq 0} \left(\frac{2}{Q_1 - Q_2}\right)^{2m + 2} (\zeta - \rho)^{2(m + 1)} \dd \zeta  \\ 
& = \sqrt{-\Lambda^r/2} \sum_{m \geq 0} \left(\frac{-2\Lambda^r}{(Q_1 - Q_2)^2}\right)^{m + 1} (2m + 1)!!\,u^m,
\end{split}  
\end{equation*}
where we used \eqref{orcon}. After integration by parts, $\tilde{R}(u)$ is precisely a  Gaussian integral, and since we know that $\tilde{R}(u) = 1 + O(u)$, we must have $\tilde{R}(u) = 1$ (this can be checked by direct computation). Comparing with \eqref{TRtheta},  the associated TR class is
\[
\tilde{\Omega}_{g,n} = (-\Lambda^r/2)^{1 - g - n/2} \cdot \Theta_{g,n}\big[\tfrac{Q_1 - Q_2}{\sqrt{2\Lambda^r}}\big].
\]
Putting all ingredients in \eqref{finalformin} and comparing its series expansion as $\zeta_i \rightarrow Q_{a_i}$ in the coordinate $x^{-1}$ with the definition of the correlators from the Gaiotto vector in \eqref{phidef}, we arrive at
\begin{multline*}
\Phi_{g,n}\big[\begin{smallmatrix} a_1 & \cdots & a_n \\ k_1 & \cdots & k_n \end{smallmatrix}\big]  =  \sum_{m_1,\ldots,m_n \geq 0} \prod_{i = 1}^n \left( \frac{(-1)^{a_i - 1} \sqrt{-2\Lambda^r}}{Q_1 - Q_2}\,\left(\frac{\Lambda^r}{(Q_1 - Q_2)^2}\right)^{k_i + m_i} \frac{(2k_i + 2m_i)!}{(k_i + m_i)!(k_i - 1)!}\right) \\
  \quad \times \big(-\tfrac{\Lambda^r}{2}\big)^{1 - g - n/2} \int_{\overline{\mathcal{M}}_{g,n}} \Theta_{g,n}\big[\tfrac{Q_1 - Q_2}{\sqrt{2\Lambda^r}}\big] \prod_{i = 1}^{n} \psi_i^{m_i} .
\end{multline*}
The coefficient of $\epsilon^{2m}$ in the deformed Theta class $\Theta_{g,n}[\epsilon]$ has cohomological degree $2(2g - 2 + n - m)$, hence must be completed by a total of $\sum_{i = 1}^{n} m_i = g - 1 + m$ classes $\psi$ in order to give a non-zero contribution. Therefore, the total power of $\Lambda^r$ is $(k_1 + \cdots + k_n)$, and the power of $(Q_1 - Q_2)$ is $2 - 2g - n - 2(k_1 + \cdots + k_n)$. Collecting the powers of $(-1)$ and $2$ as well, we can replace the argument of the deformed Theta class by $\epsilon = 1$ to arrive at \eqref{theform1Phi}.

The second representation comes from the TR class after logarithmic twist. Let $\check{x}(\zeta) = \ln x(\zeta)$ and $\check{y}(\zeta) = x(\zeta)y(\zeta) = \zeta$.  As $\zeta \rightarrow \rho$ we have
\begin{equation*}
\begin{split} \check{x}(\zeta) & = -\ln \tilde{x}(\zeta) = -\ln\left(\tilde{x}(\rho) + \frac{1}{2} \tilde{x}''(\rho) (\zeta - \rho)^2 + o(\zeta - \rho)^2\right) \\
& = -\ln \tilde{x}(\rho) - \frac{\tilde{x}''(\rho)}{2\tilde{x}(\rho)} (\zeta - \rho)^2 + o(\zeta - \rho)^2,
\end{split}
\end{equation*}
This implies that primary differential after the logarithmic twist is a simple rescaling of \eqref{primartil}:
\[
\dd \check{\Xi}_{\rho,0}(\zeta) = \sqrt{-\tilde{x}(\rho)}\,\dd \tilde{\Xi}_{\rho,0}(\zeta) = \frac{Q_1 - Q_2}{2\sqrt{2}}\,\frac{\dd \zeta}{(\zeta - \rho)^2}.
\]
As we are looking at simple poles of $x$, we are in the case $\check{d}_{Q_a} = 1$ and $\check{c}_{Q_a} = -1$ in \eqref{Xcoo}. This leads to
\begin{equation}
\label{Shig}
\check{\mathsf{S}}_{\rho,m}\big[\begin{smallmatrix} Q_a \\ k \end{smallmatrix}\big] = \sqrt{-\tilde{x}(\rho)}\, k^m \,\tilde{\mathsf{S}}_{\rho,0}\big[\begin{smallmatrix} Q_a \\ k \end{smallmatrix}\big] = k^m \frac{(-1)^{a-1}}{\sqrt{2}} \left(\frac{\Lambda^r}{(Q_1 - Q_2)^2}\right)^{k} \frac{(2k)!}{k!(k - 1)!}.
\end{equation}
To get the TR class, we compute 
\begin{equation*}
\begin{split}
\check{T}(u) & = \frac{1}{\sqrt{2\pi u}} \int_{\gamma_{\rho}} \left(\frac{-4(\zeta - Q_1)(\zeta - Q_2)}{(Q_1 - Q_2)^2}\right)^{u^{-1}} \dd \zeta \\
& = -\frac{(Q_1 - Q_2)4^{u^{-1}}}{\sqrt{2\pi u}} \int_{\gamma_{\rho}} \tilde{\zeta}^{u^{-1}} (1 - \tilde{\zeta})^{u^{-1}} \dd \tilde{\zeta} \\
& = -\frac{(Q_1 - Q_2)4^{u^{-1}}}{\sqrt{2\pi u}} \frac{\big(\Gamma(u^{-1} + 1)\big)^2}{\Gamma(2u^{-1} + 2)} = -\frac{Q_1 - Q_2}{2\sqrt{2}} \,\frac{1}{1 + u/2}\,\frac{\big(\Gamma_{\text{reg}}(u^{-1})\big)^2}{\Gamma_{\text{reg}}(2u^{-1})},
\end{split} 
\end{equation*}
In the line before the last, we use the integral representation of the beta function.\footnote{The orientation of the contour should be chosen such that $\check{R}(u) = 1 + O(u)$ in the next computation.} Finally, using integration by parts and $\frac{x'(\zeta)}{x(\zeta)} = \frac{-2(\zeta - Q_1)}{(\zeta - Q_1)(\zeta - Q_2)}$, we compute
\begin{equation*}
\begin{split}
\check{R}(u) & = \frac{Q_1 - Q_2}{2\sqrt{2}}\,\frac{1}{\sqrt{2\pi u^{-1}}} \int_{\gamma_{\rho}} \left(\frac{x(\rho)}{x(\zeta)}\right)^{u^{-1}} \frac{\dd \zeta}{(\zeta - \rho)^2} \\
& = \frac{Q_1 - Q_2}{\sqrt{2}}\,\frac{1}{\sqrt{2\pi u}} \int_{\gamma_{\rho}} \left(\frac{x(\rho)}{x(\zeta)}\right)^{u^{-1}} \frac{\dd \zeta}{(\zeta - Q_1)(\zeta - Q_2)} \\
& = \frac{Q_1 - Q_2}{\sqrt{2}} \left(\frac{-4}{(Q_1 - Q_2)^2}\right)^{u^{-1}} \frac{1}{\sqrt{2\pi u}} \int_{\gamma_{\rho}} \big((\zeta - Q_1)(\zeta - Q_2)\big)^{u^{-1} - 1}\dd \zeta \\ 
& = \frac{4^{u^{-1}}}{2\sqrt{\pi u}}\,\frac{\big(\Gamma(u^{-1})\big)^2}{\Gamma(2u^{-1})} = \frac{\big(\Gamma_{\text{reg}}(u^{-1})\big)^2}{\Gamma_{\text{reg}}(2u^{-1})}.
\end{split}
\end{equation*}  
Except for the extra factor $(1 + u/2)^{-1}$ in $\check{T}(u)$, we recognise products of the $R$ and $T$ series that appeared for the Hodge class. Using the multiplicativity properties of TR classes, we get
\[
\check{\Omega}_{g,n} \prod_{i = 1}^{n} \check{R}(\psi_i) = \left(-\frac{Q_1 - Q_2}{2\sqrt{2}}\right)^{2 - 2g - n}  \exp\left(\sum_{m \geq 1} \frac{(-1)^{m + 1}}{m\,2^m}  \kappa_m \right)  \boldsymbol{\Lambda}[-1]^2 \boldsymbol{\Lambda}\big[\tfrac{1}{2}\big] .
\]
Putting this together with \eqref{Shig} in the master formula yields the second claim.
\end{proof}

\begin{rem}
	\cref{pr:intGai2} can also be derived by combining \cref{th:main} with the results of \cite{YZ24,AN24,CGG22}, and applying the techniques described in part \textbf{C.} of \cref{S511}.  First, the correlators of the generalized BGW $\tau$-function  introduced by \cite{Ale18} can be obtained by expanding the topological correlators $\omega_{g,n}$  on the $r = 2$ half Seiberg--Witten curve at the point $\zeta = \infty$ with the identification of parameters $s = Q_1 - Q_2$, as proved in \cite[Section 5]{AN24}. (It would be possible to give a proof of the aforementioned statement using the techniques developed in this paper, as the generalized BGW $\tau$-function satisfies a full set of Virasoro constraints.) Second, the generalized BGW correlators are proved to admit an interpretation in terms of triple Hodge integrals $\boldsymbol{\Lambda}[-1]^2 \boldsymbol{\Lambda}\big[\tfrac{1}{2}\big]$ in \cite{YZ24}. Third, the expansion coefficients of the $\omega_{g,n}$ near the ramification point $\zeta = \frac{Q_1 + Q_2}{2}$ encode descendant integrals of the deformed Theta class \cite{CGG22}. Combining all these results, \cref{pr:intGai2} can be proved by relating the expansion coefficients in different bases of expansion using the techniques described in part \textbf{C.} of \cref{S511}.
\end{rem}

\begin{rem}
	As alluded to in the previous remark, it is worth stressing that the very same $\omega_{g,n}$ on the $r =2$ half Seiberg--Witten curve (up to twisting) calculates the Gaiotto vector correlators, triple Hodge integrals, deformed Theta integrals or generalized BGW correlators depending on whether one expands at $\zeta = Q_a$, $\zeta = \frac{Q_1+Q_2}{2} $, $\zeta = \frac{Q_1+Q_2}{2}$ or $\zeta = \infty$ respectively, in an appropriate basis of differentials. 
\end{rem}

\vspace{0.5cm}

\subsection{Relation to Hurwitz theory}\label{sec:Hurwitz} 
\label{S5H}

\medskip

\subsubsection{Hurwitz theory and topological recursion}

\medskip

Let us briefly review the formalism of weighted Hurwitz numbers, referring to \cite{GPH17} for the details. Consider a formal power series of the form 
\[
	\mathcal{G}(\zeta) =  \sum_{m \geq 0} \mathcal{G}_m \zeta^m, \qquad \mathcal{G}_0 \neq 0.
\] 
Weighted single Hurwitz numbers $H_{\mathcal{G};g}(\mu_1,\ldots, \mu_n)$ of genus $g \geq 0$ with $\mu_1,\dots, \mu_n \in \mathbb Z_{> 0}$ are weighted sums of ramified coverings of $\mathbb P^1 $ by a smooth genus $ g$ curve of degree $d = \sum_{i=1}^n \mu_i$ with prescribed ramification profile $\{\mu_1,\ldots, \mu_n\}$ over $\infty \in \mathbb P^1$. The weight depends on the profile of the other ramification points in a way specified by $\mathcal{G}$.

To define them, we start by defining the disconnected weighted single Hurwitz numbers via the following calculation in the center of the group algebra of $\mathfrak{S}_d$
\[
H_{\mathcal{G};\chi}^{\bullet}(\mu_1,\ldots,\mu_n) = \frac{1}{d!} [\beta^{\chi + n - d}.\text{id}]\, C_{\mu} \prod_{i = 1}^{d - 1} \mathcal{G}(\beta\mathsf{J}_i)
\]
where $\mathsf{J}_i = \sum_{j < i} (j\,i)$ are the Jucys--Murphy elements, $C_{\mu}$ the indicator of the conjugacy class of a permutation with cycles of lengths $\mu_1,\ldots,\mu_n$, and $[\beta^{b}.\text{id}]$ extracts the coefficient of $\beta^b.\text{id}$ from the expression to its right. The interpretation as enumeration of branched covers of $\mathbb{P}^1$ of Euler characteristic $\chi$ is well-known.
The (connected) weighted single Hurwitz numbers $H_{\mathcal{G};g}(\mu_1,\ldots,\mu_n)$ are then defined by inclusion-exclusion from the disconnected ones. Classical choices of weights are
\begin{itemize} 
\item $\mathcal{G}(\zeta) = \exp(\zeta)$: simple Hurwitz numbers with simple ramification away from $\infty$;
\item $\mathcal{G}(\zeta) = \frac{1}{1-\zeta}$: strictly monotone Hurwitz numbers;
\item $\mathcal{G}(\zeta) = (1+\zeta)$: weakly monotone Hurwitz numbers, closely related to dessins d’enfants or bipartite maps.
\end{itemize}

Weighted Hurwitz numbers are governed by topological recursion: this has been proved in \cite{ACEH3} for polynomial $\mathcal{G}$ and in \cite{BDBKS} in largest possible generality, including rational-exponential $\mathcal{G}$.  For rational $\mathcal{G}$, an alternative proof that relates the cut-and-join equation for rationally weighted Hurwitz numbers with $\mathcal W$-constraints also appeared recently in \cite{CDO24}. The precise statement is the following.

\begin{thm}[\cite{ACEH3, BDBKS}]\label{thm:weightedH}
	Assume that $\mathcal{G}$ is an exponential times a rational function, and let $(\hat{\omega}_{g,n})_{g,n}$ be the correlators constructed by running topological recursion on the spectral curve $\mathbb{P}^1$ parametrised as
\begin{equation}
\label{spcurveH}
\hat{x}(\zeta) = \frac{\zeta}{\gamma \mathcal{G}(\zeta)},\qquad \hat{y}(\zeta) = \frac{\zeta}{\hat{x}(\zeta)},
\end{equation}
and equipped with $\hat{\omega}_{0,1} = \hat{y}\dd \hat{x}$ and $\hat{\omega}_{0,2}(\zeta_1,\zeta_2) = \frac{\dd \zeta_1\dd \zeta_2}{(\zeta_1 - \zeta_2)^2}$, assuming that $\hat{x}$ has simple ramification points. Then, we have  for any $g,n$ the all-order series expansion as $\zeta_i \rightarrow 0$ in the variable $x_i = x(\zeta_i)$:
\begin{multline*}
\hat{\omega}_{g,n}(\zeta_1,\dots, \zeta_n) - \frac{\delta_{g,0}\delta_{n,2} \dd \hat{x}_1\dd \hat{x}_2}{(\hat{x}_1 - \hat{x}_2)^2} \\
	  \approx	\sum_{\mu_1,\dots,\mu_n \in \mathbb Z_{> 0}} \#\Aut(\mu) \cdot \gamma^{\mu_1+\cdots+\mu_n}\cdot H_{\mathcal{G};g}(\mu_1,\dots,\mu_n)  \prod_{i=1}^n \dd(\hat{x}(z_i)^{\mu_i})\,,
	\end{multline*}	
\end{thm}
The factor $\# \text{Aut}( \mu) = \prod_{i \geq 1} i^{m_i} m_i!$ with $m_i = \#\{j \,\,|\,\,\mu_j = i\}$ is necessary for the comparison with the normalisation of $H_{\mathcal{G};g}$ in \cite{GPH17}. The assumption that the spectral curve has simple ramification points can be waived using the limit arguments of \cite{TRlimits}.

\medskip

\subsubsection{Application to \texorpdfstring{$\ket{\Gamma_{\Lambda}}$}{GammaLambda} and \texorpdfstring{$\ket{\Gamma_{\Lambda}^{\text{CDO}}}$}{GammaLambdaCDO}}

\label{sec:H}

\medskip

Combining \cref{thm:weightedH} and our main results \cref{th:main} and \cref{th:main2}, we can give a Hurwitz theory interpretation to a part of the two Whittaker vectors, after specialisation to the parameter $Q_1 = 0 $. Let us start with the Gaiotto vector $\ket{\Gamma_{\Lambda}}$, whose expansion coefficients were denoted $\Phi_{g,n} $ in \eqref{Gformalexp}. We have the following  corollary of \cref{th:main} --- its assumption at $Q_1 = 0$ imposes that the other $Q$s are non-zero and pairwise distinct.
\begin{cor}\label{cor:HforG}
	Consider the Gaiotto vector $\ket{\Gamma_{\Lambda}}$, and substitute $Q_1 = 0$. Then, for any pairwise distinct $Q_2,\ldots, Q_{r} \in \mathbb C^*$, any $(g,n) \in \mathbb Z_{\geq 0} \times \mathbb Z_{>0}$ and $\mu \in \mathbb{Z}_{> 0}^n$, we have
\[
	 \Phi_{g,n}\big[\begin{smallmatrix} 1 & \cdots & 1\\ \mu_1 & \cdots & \mu_n \end{smallmatrix}\big]  = \mu_1 \cdots \mu_n \cdot  \#\Aut(\mu) \cdot \Lambda^{r(\mu_1 + \cdots + \mu_n)} \cdot H_{\mathcal{G};g}(\mu_1,\dots,\mu_n),
\]
 with the weight generating series
\[
\mathcal{G}(\zeta) = \frac{1}{\prod_{a=2}^r(Q_a-\zeta)}.
\]
\end{cor}
\begin{cor}\label{cor:HforZ}
	Consider the Whittaker vector $\ket{\Gamma_{\Lambda}^{\text{CDO}}}$ with $r \geq 2$ and substitute $Q_1 = 0$. Then, for any $(P_1,\ldots,P_{r},0,Q_2,\ldots,Q_{r - 1})$ in the set $\mathcal{R}$ of \cref{Rset}, and $(g,n) \in \mathbb{Z}_{\geq 0} \times \mathbb{Z}_{> 0}$ and $\mu \in \mathbb{Z}_{> 0}^{n}$, we have
\[
		\Phi_{g,n}\big[\begin{smallmatrix} 1 & \cdots & 1\\ \mu_1 & \cdots & \mu_n \end{smallmatrix}\big]  = \mu_1 \cdots \mu_n \cdot \# \Aut(\mu) \cdot \Lambda^{r(\mu_1+\cdots+\mu_n)} \cdot H_{\mathcal{G};g}(\mu_1,\dots,\mu_n),
	\]
with the weight generating series
\[
\mathcal{G}(\zeta) = \frac{\prod_{a=1}^r(P_a  +\zeta)}{\prod_{a=2}^{r-1}(Q_a-\zeta)}.
\]
\end{cor}
\begin{proof} 
	For \cref{cor:HforG}, consider the correlators $\phi_{g,n}$ defined in \eqref{phidef} using the coefficients  $\Phi_{g,n}$ of $\ket{\Gamma_{\Lambda}}$. \cref{th:main} states that these correlators can be analytically continued to the curve defined by 
	 \[
	x(\zeta) = - \dfrac{\Lambda^r}{\prod_{a = 1}^{r} (Q_a - \zeta)}, \qquad y(\zeta) =  \dfrac{\zeta}{x(\zeta)} = -\dfrac{\zeta}{\Lambda^r} \prod_{a = 1}^{r} (Q_a - \zeta), \qquad \zeta \in \mathbb{P}^1.
	 \]
	and their analytic continuations coincide with the correlators $(\omega_{g,n})_{g,n}$ of the topological recursion on this curve, considered as a spectral curve with $\omega_{0,1} = y \dd x$ and $\omega_{0,2}(\zeta_1,\zeta_2) = \frac{\dd \zeta_1 \dd \zeta_2}{(\zeta_1 - \zeta_2)^2}$. Setting $Q_1 = 0$ and  define $ \hat{x} = x^{-1}$ and $\hat{y} = x^2y$ recovers the curve \eqref{spcurveH} associated to weighted Hurwitz numbers with $\mathcal{G}(\zeta) =\frac{1}{\prod_{a=2}^r(Q_a-\zeta)}$ and $\gamma = \Lambda^r$. Let $\hat{\omega}_{g,n}$ be the topological recursion correlators on this curve.
	
	Since $y\dd x = -\hat{y}\dd \hat{x}$ and the ramification points of $x$ are those of $\hat{x}$, the recursive definition \eqref{TRome} implies $\omega_{g,n} = (-1)^{2g - 2 + n} \hat{\omega}_{g,n} = (-1)^n \omega_{g,n}$. 
Then, for any $g,n$ and $\mu_1,\ldots,\mu_n > 0$ we have
\begin{equation*}
\begin{split}
 \Phi_{g,n}\big[\begin{smallmatrix} 1 & \cdots & 1\\ \mu_1 & \cdots & \mu_n \end{smallmatrix}\big] & \mathop{=}_{\text{Thm.}\,\, \ref{th:main}} \quad  (-1)^n \Res_{\zeta_1 = 0} \cdots \Res_{\zeta_n = 0} \omega_{g,n}(\zeta_1,\ldots,\zeta_n) \prod_{i = 1}^{n} x(\zeta_i)^{\mu_i} \\ 
 & \mathop{=}_{\phantom{Thm. 1.2}} \quad \Res_{\zeta_1 = 0} \cdots \Res_{\zeta_n = 0} \hat{\omega}_{g,n}(\zeta_1,\ldots,\zeta_n) \prod_{i = 1}^{n} \hat{x}(\zeta_i)^{-\mu_i} \\
 & \mathop{=}_{\text{Thm.}\,\,\ref{thm:weightedH}} \quad \mu_1 \cdots \mu_n \cdot \# \text{Aut}(\mu) \cdot \Lambda^{r(\mu_1 + \cdots + \mu_n)} \cdot H_{\mathcal{G};g}(\mu_1,\ldots,\mu_n).
 \end{split}
\end{equation*} 
The proof of Corollary~\ref{cor:HforZ} is similar and omitted.
 \end{proof}

\begin{rem}
The presence of the prefactor $\mu_1 \cdots \mu_n$ is due to the choice of normalisation in the definition of the coefficients of the Whittaker vectors, see e.g. \eqref{Gformalexp}. The choice of setting $Q_1 = 0$ in the above  \cref{cor:HforG} and \cref{cor:HforZ} is arbitrary. Indeed, if one were to set one of the other $Q_a $ to $0$ instead, the only difference in the corollaries would be that the Hurwitz numbers would appear as $\Phi_{g,n}\big[\begin{smallmatrix} a & \cdots & a\\ \mu_1 & \cdots & \mu_n \end{smallmatrix}\big] $ instead of 	$\Phi_{g,n}\big[\begin{smallmatrix} 1 & \cdots & 1\\ \mu_1 & \cdots & \mu_n \end{smallmatrix}\big] $. 
\end{rem}
 
\medskip
 
 \subsubsection{Comments on gauge/Hurwitz correspondences}
 
 \label{SGT}
 
 \medskip
 
In this section, we comment on the relation between gauge theory and Hurwitz theory that has appeared previously in the physics literature \cite{GrossTaylor, GrossTaylor2}.  Gross and Taylor consider the large $N$ expansion of two-dimensional $\text{U}_N$ Yang--Mills theory on a target Riemann surface, and interpret the expansion coefficients as certain  Hurwitz numbers counting branched coverings of the target Riemann surface. A precise mathematical statement of the former is proved in the recent paper \cite{GTNovak}. On the other hand, the large $N$ limit of  two-dimensional $\text{U}_N$ Yang--Mills theory on $\mathbb S_2$ can  be interpreted as the instanton partition function for four-dimensional $\mathcal N =2$ supersymmetric gauge theories \cite{MMO04}. Putting these two facts together gives a Hurwitz-theoretic interpretation of the Nekrasov instanton partition function. 

The connection between gauge theory and Hurwitz theory that we proved in the previous section \cref{sec:H} is different from the aforementioned one derived from two-dimensional Yang--Mills theory. Indeed, \cref{cor:HforG} shows that the expansion coefficients of the Gaiotto vector themselves coincide with certain rationally-weighted Hurwitz numbers, as opposed to the instanton partition function (which is the squared-norm of the Gaiotto vector). In the specific case of two-dimensional Yang--Mills theory on $\mathbb S_2$, which corresponds to our case of interest, the associated Hurwitz numbers are (simple) double Hurwitz numbers, see \cite[Theorem 4.4]{GTNovak}. It may be possible to recover this result from our \cref{cor:HforG}, but we are not aware of a derivation of \cref{cor:HforG} from the observations of Gross and Taylor.

\vspace{0.5cm}

\subsection{Quantum curves} 
\label{S53}
\medskip

\subsubsection{Background and definitions}

\medskip
The Gaiotto and the CDO curves belong to the class of spectral curves for which Bouchard and Eynard have constructed associated quantum curves in \cite{BE17}, so we can directly apply their results. We first review the context.

\medskip

\noindent \textbf{A. Wave functions and quantum curves}

Given the correlators $(\omega_{g,n})_{g,n}$ constructed by topological recursion on a genus zero spectral curve, say $S$, uniformised by the coordinate $\zeta \in \mathbb P^1$, we can define the associated \textit{wave function}  (sign $+$) and the \textit{dual wave function} (sign $-$) 
\begin{equation}
\label{def:wave}
	\psi_{\beta}^{\pm} := \exp \left(\sum_{\substack{g \in \mathbb Z_{\geq 0}  \\ n \in \mathbb Z_{>0}}} \frac{(-1)^n\hslash^{2g-2+n}}{n!}\int^\zeta_{\beta} \cdots \int^\zeta_{\beta} \left(\omega_{g,n} (\zeta_1,\cdots, \zeta_n)- \delta_{g,0}\delta_{n,2} \frac{\dd x(\zeta_1) \dd x(\zeta_2)}{(x(\zeta_1)-x(\zeta_2))^2} \right) \right)\,,
\end{equation}
It depends on the choice of a base point $\beta \in S$, which is typically chosen in $x^{-1}(\infty)$. The integral $\int_{\beta}^\zeta \omega_{0,1}$ may need to be regularised (see \cref{01regintdef}). The dual wave function is obtained from the wave function by replacing $\hslash$ with $-\hslash$.

A quantisation of the spectral curve whose underlying defining polynomial is $P(x,y) = 0$ is a differential operator $\widehat{P}(\hslash,x,\hslash \frac{\dd}{\dd x})$ which is polynomial in $\hslash$, and whose symbol --- obtained by replacing $\hslash \frac{\dd}{\dd x} $ by $ y $ and then setting $\hslash $ to $ 0$ --- is the polynomial $P(x,y)$.  There is no unique way to upgrade $P(x,y) $ to a quantum curve as there are always issues of ordering --- in particular $xy - yx = 0$ gets quantised to $x \hslash  \frac{\dd}{\dd x} - \hslash   \frac{\dd}{\dd x} x \neq 0 $. The main theorem of \cite{BE17} is that, as long as the curve is of genus $0$ and satisfies a key admissibility condition (which is valid for the Gaiotto and CDO curves), there exists a differential operator which is a quantisation of the original spectral curve and whose solution --- considered as a formal expansion as $\hslash \rightarrow 0$ --- is the wave function $\psi_{\beta}^{\pm}$. Such a quantisation is usually called a \textit{quantum curve}, and \cite{BE17} provides an explicit algorithm to compute it.

\medskip

\noindent \textbf{B. Wave functions for Whittaker vectors}

Recall that the Gaiotto vector takes the form
\[
\ket{\Gamma_{\Lambda}} = \exp\left(\sum_{(g,n) \in \mathbb{Z}_{\geq 0} \times \mathbb{Z}_{> 0}} \frac{\hslash^{g - 1}}{n!} \sum_{\substack{a_1,\ldots,a_n \in [r] \\ k_1,\ldots,k_n \in \mathbb{Z}_{> 0}}} \Phi_{g,n}\big[\begin{smallmatrix} a_1 & \cdots & a_n \\ k_1 & \cdots & k_n \end{smallmatrix}\big] \prod_{i = 1}^{n} \frac{J_{-k_j}^{a_j}}{k_j}\right).
\] 
Consider as well the correlators $\omega_{g,n}$ of the associated spectral curve, related to the Gaiotto vectors by \eqref{phidef}. If we introduce the algebra homomorphism called  \emph{principal specialisation}
\begin{equation}
\label{evaop}
\text{ev}^\pm_{a} : \begin{array}{cll} J_{-k}^{b} &  \longmapsto & \mp \hslash \delta_{a,b} kx^{-k} \\
\hslash & \longmapsto & \hslash^2 \end{array},
\end{equation}
we then have as $\zeta \rightarrow Q_a$ for each $a \in [r]$
\begin{equation}
\label{principsi} 
\psi_a^{\pm}(\zeta) := \psi^\pm_{Q_a}(\zeta)  \mathop{\approx}  \text{ev}^{\pm}_{a} \left( \ket{\Gamma_{\Lambda}}\right).
\end{equation}
The same holds for the CDO vector, provided we take the convention $Q_r = \infty$ in \eqref{principsi} for the remaining pole of $x$.

\medskip
 
 \noindent \textbf{C. The $(0,1)$ and $(0,2)$ terms}
 
The $(0,1)$ and $(0,2)$ terms in the wave functions \eqref{def:wave} must be carefully examined due to the singularities in the integrand. We explain the regularisation for $(0,1)$, that should be part of the definition \eqref{def:wave} for \eqref{principsi} to hold, and show that the $(0,2)$ term is well-defined. This will also be useful to evaluate the asymptotics of the wave functions as $\zeta$ approaches the $r$ poles of $x$ (Corollary~\ref{leadingpsi1} and \ref{leadingpsi2}).

In the Gaiotto case, we have $x^{-1}(\infty) = \{Q_1,\ldots,Q_r\}$ and $\omega_{0,1} = y \dd x$ has a simple pole at $\zeta = Q_a$, with residue $-Q_a$. For the CDO case, we have $x^{-1}(\infty) = \{Q_1,\ldots,Q_{r - 1},\infty\}$ and $\omega_{0,1}$ has a simple pole at $\zeta = Q_a$ with residue $-Q_a$ for $a \in [r-1]$. At $\zeta = \infty$ we rather have a double pole with behaviour
$$
\omega_{0,1}(\zeta) = \left((-\Lambda)^{-r} - \frac{|\mathbf{P}| + |\mathbf{Q}|}{x(\zeta)} + O\big(x(\zeta)^{-2}\big)\right)\dd x(\zeta).
$$ 
Therefore, we take the following definition of the regularised integrals.
\begin{defn}
\label{01regintdef} For the $(0,1)$ term in the wave function \eqref{def:wave} for the half Seiberg--Witten  curve, we take 
\begin{equation}
\label{01reg} \int_{Q_a}^{\zeta} \omega_{0,1} :=  Q_a\ln(x(\zeta))  +  \int_{Q_a}^{\zeta}  \left(\omega_{0,1} - \frac{Q_a \dd x}{x}\right),
\end{equation}
and in the CDO case
 \begin{equation*}
\int_{\infty}^{\zeta} \omega_{0,1}  := (-\Lambda)^{-r}x(\zeta) - \left(|\mathbf{P}| + |\mathbf{Q}|\right)\ln x(\zeta)  
+ \int_{\infty}^{\zeta} \left(\omega_{0,1} - \bigg((-\Lambda)^{-r} - \frac{|\mathbf{P}| + |\mathbf{Q}|}{x}\bigg)\dd x\right).
\end{equation*}
\end{defn}

Since $\omega_{g,n}$ for $2g - 2 + n > 0$ only has poles at ramification points, those multiple integrals in \eqref{def:wave} are well-defined. The double integral of $\omega_{0,2}$ is also well-defined and explicitly computable.
\begin{lem}
\label{02intlem} We have
\[
\int_{Q_a}^{\zeta} \int_{Q_a}^{\zeta} \left( \omega_{0,2}(\zeta',\zeta'') - \frac{\dd x(\zeta')\dd x(\zeta'')}{(x(\zeta') - x(\zeta''))^2}\right)  = \ln\left(\frac{A_a}{x'(\zeta)(\zeta - Q_a)^2}\right).
\]
with
\begin{equation}
\label{AAq}
A_a =  -\frac{\Lambda^r}{\prod_{b \neq a} (Q_b - Q_a)}\qquad \text{or} \qquad A_a = - \Lambda^r \frac{\prod_{b = 1}^{r} (P_b + Q_a)}{\prod_{b \neq a} (Q_b - Q_a)}.
\end{equation}
And, in the CDO case:
\[
\int_{\infty}^{\zeta} \int_{\infty}^{\zeta} \left( \omega_{0,2}(\zeta',\zeta'') - \frac{\dd x(\zeta')\dd x(\zeta'')}{(x(\zeta') - x(\zeta''))^2}\right)  = \ln\left(\frac{(-\Lambda)^r}{x'(\zeta)}\right).
\]
\end{lem} 
\begin{proof}
Since $\omega_{0,2}(\zeta_1,\zeta_2) = \frac{\dd \zeta_1 \dd \zeta_2}{(\zeta_1 - \zeta_2)^2} = \dd_{\zeta_1}\dd_{\zeta_2} \ln(\zeta_1 - \zeta_2)$, we have
\begin{equation*}
\begin{split}
& \quad \int_{\zeta_2}^{\zeta_1} \int_{\zeta_4}^{\zeta_3} \left( \omega_{0,2}(\zeta',\zeta'') - \frac{\dd x(\zeta')\dd x(\zeta'')}{(x(\zeta') - x(\zeta''))^2}\right) \\
& = \ln\left(\frac{\zeta_1 - \zeta_3}{x(\zeta_1) - x(\zeta_3)}\,\frac{\zeta_2 - \zeta_4}{x(\zeta_2) - x(\zeta_4)}\,\frac{x(\zeta_1) - x(\zeta_4)}{\zeta_1 - \zeta_4}\,\frac{x(\zeta_2) - x(\zeta_3)}{\zeta_2 - \zeta_3}\right).
\end{split}
\end{equation*}
Taking $\zeta_1 = \zeta_3 = \zeta$ and $\zeta_2 = \zeta_4 = \beta$ yields
\[
\int_{\beta}^{\zeta} \int_{\beta}^{\zeta} \left( \omega_{0,2}(\zeta',\zeta'') - \frac{\dd x(\zeta')\dd x(\zeta'')}{(x(\zeta') - x(\zeta''))^2}\right)  = \ln\left(\frac{1}{x'(\beta)x'(\zeta)}\,\frac{(x(\zeta) - x(\beta))^2}{(\zeta - \beta)^2}\right).
\]
As $\beta \rightarrow Q_a$ we have $x(\beta) \sim \frac{-A_a}{\beta - Q_a}$ for  $A_a \in \mathbb{C}^*$ as given in \eqref{AAq}, and thus
\[
\int_{Q_a}^{\zeta} \int_{Q_a}^{\zeta} \left( \omega_{0,2}(\zeta',\zeta'') - \frac{\dd x(\zeta')\dd x(\zeta'')}{(x(\zeta') - x(\zeta''))^2}\right)  = \ln\left(\frac{A_a}{x'(\zeta)(\zeta - Q_a)^2}\right).
\]
In the CDO case, as $\beta \rightarrow \infty$ we have $x(\beta) \sim (-\Lambda)^{r}\beta$ and this leads to the second claim.
\end{proof}

\medskip

 \noindent \textbf{D. Asymptotics of wave functions near poles of $x$}

\begin{defn}
\label{stablepsi}The stable part of the wave function is
\[
\widetilde{\psi}_{a}^{\pm}(\zeta) = \exp\left(\sum_{\substack{(g,n) \in \mathbb{Z}_{\geq 0}\times\mathbb{Z}_{>0} \\ 2g - 2 + n > 0}} \frac{(\pm\hslash)^{2g - 2 + n}}{n!} \int_{Q_a}^{\zeta} \cdots \int_{Q_a}^{\zeta} \omega_{g,n}\right)
\]
\end{defn}
In other words, it is \eqref{def:wave} where we omit the $(0,1)$ and the $(0,2)$ terms. The properties of the $\omega_{g,n}$ for $2g - 2 + n > 0$ imply that the expression inside is a formal power series in $\hslash$ of meromorphic functions of $\zeta$, with poles only at the ramification points. By construction, $\widetilde{\psi}_{a}^{\pm}(Q_a) = 1$ for any $a \in [r]$.

\begin{cor}
\label{leadingpsi1}
For the Gaiotto wave functions, for any $a,b \in [r]$ we have
\[
\psi_{a}^{+}(\zeta)\,\,\mathop{\sim}_{\zeta \rightarrow Q_b}\,\, B_{a,b} \cdot \widetilde{\psi}_{a}^+(Q_b) \cdot (x(\zeta))^{\hslash^{-1}Q_b - \delta_{b \neq a}},
\]
where
\begin{equation*}
\begin{split}
 B_{a,b} - \delta_{a,b} & = \pm \delta_{a \neq b} \cdot {\rm i} (\Lambda^r)^{\hslash^{-1}(Q_a - Q_b) + 1} \cdot e^{\hslash^{-1}r(Q_a - Q_b)} \cdot (Q_a - Q_b)^{2\hslash^{-1}(Q_b - Q_a) - 2} \\
 & \quad   \times \prod_{c \neq a,b} (Q_c - Q_a)^{\hslash^{-1}(Q_c - Q_a) - \frac{1}{2}}(Q_c - Q_b)^{\hslash^{-1}(Q_b - Q_c) - \frac{1}{2}}
 \end{split}
\end{equation*}
\end{cor}

\begin{cor}
\label{leadingpsi2}
In the CDO case, for any $a,b \in [r - 1]$ we have
\[
\psi_{a}^+(\zeta) \,\,\mathop{\sim}_{\zeta \rightarrow Q_b}\,\, B_{a,b}\,x^{\hslash^{-1}Q_b - \delta_{b \neq a}},
\]
where
\begin{multline*}
B_{a,b} - \delta_{a,b} = \pm \delta_{a \neq b}\cdot {\rm i} e^{\hslash^{-1}(Q_b - Q_a)} \cdot (\Lambda^r)^{\hslash^{-1}(Q_a - Q_b) + 1} \cdot (Q_a - Q_b)^{2\hslash^{-1}(Q_b - Q_a) - 2}
 \\ 
  \times \prod_{c \neq a,b} (Q_c - Q_a)^{\hslash^{-1}(Q_c - Q_a) - \frac{1}{2}} (Q_c - Q_b)^{\hslash^{-1}(Q_b - Q_c) - \frac{1}{2}}
    \\
 \times \prod_{c = 1}^{r} (Q_b + P_c)^{-\hslash^{-1}(Q_b + P_c) + 1} (Q_a + P_c)^{\hslash^{-1}(Q_a + P_c)} \times \widetilde{\psi}^+_a(Q_b) 
\end{multline*}
For $a \in [r - 1]$ and $b = r$, we have
\[
\psi_{a}^+(\zeta) \,\,\mathop{\sim}_{\zeta \rightarrow \infty} B_{a,r}\,e^{\hslash^{-1}(-\Lambda)^{-r}x} \,x^{-\hslash^{-1}(|\mathbf{P}| + |\mathbf{Q}|) - 1} ,
\]
where
\begin{multline*}
B_{a,r}  = \pm {\rm i}  (-\Lambda^r)^{\hslash^{-1}(Q_a + |\mathbf{P}| +|\mathbf{Q}|) + 1} e^{-\hslash^{-1}Q_a} 
\\
 \times \prod_{c \neq a} (Q_a - Q_c)^{\hslash^{-1}(Q_c - Q_a) -\frac{1}{2}} \prod_{c = 1}^{r} (P_c + Q_a)^{\hslash^{-1}(P_c + Q_a) + \frac{1}{2}} \cdot  \widetilde{\psi}_{a}^+(\infty).
\end{multline*}
For $a = r$ and $b \in [r - 1]$, we have
\[
\psi_r^+(\zeta)\,\,\mathop{\sim}_{\zeta \rightarrow Q_b} \,\,B_{r,b} \,x^{\hslash^{-1} Q_b - 1},
\]
where
\begin{multline*}
B_{r,b} = \pm {\rm i} e^{\hslash^{-1}(|\mathbf{P}| + |\mathbf{Q}| + Q_b)}\,((-\Lambda)^r)^{-\hslash^{-1}(|\mathbf{P}| + |\mathbf{Q}| + Q_b) + 1} \cdot  
\\
 \times \prod_{c \neq b} (Q_b - Q_c)^{\hslash^{-1}(Q_b - Q_c) - \frac{1}{2}} \cdot \prod_{c = 1}^r (Q_b+P_c)^{-\hslash^{-1}(Q_b + P_c) + \frac{1}{2}} \cdot \widetilde{\psi}_r^+(Q_b).
\end{multline*}
For $a = b = r$, we have
\[
\psi_r^+(\zeta)\,\,\mathop{\sim}_{\zeta \rightarrow \infty} \,\, \,B_{r,r} e^{\hslash^{-1}(-\Lambda)^{-r}x}\,x^{-\hslash^{-1}(|\mathbf{P}| + |\mathbf{Q}|)}.
\]
with $B_{r,r} = 1$.
\end{cor} 
\begin{proof}
	Both corollaries follow immediately from \cref{01regintdef} and \cref{02intlem}.
\end{proof}

\subsubsection{Application to Gaiotto and CDO vectors}

\medskip

The half Seiberg--Witten and CDO spectral curves are of genus $0$ and all the poles of $x$ are simple. Therefore, we can apply \cite[Lemma 5.14]{BE17} to obtain the quantum curves. It turns out that these are generalised hypergeometric differential equations. We denote $D = \hslash x\partial_x$ and introduce the Pochhammer symbol
\[
[x]_k = x(x + 1) \cdots (x + k - 1) = \frac{\Gamma(x + k)}{\Gamma(x)}.
\]
We also note that
\[
x(x-1) \cdots (x - k + 1) = (-1)^k [-x]_{k}.
\] 

The following propositions give quantum curves for the Gaiotto and CDO wave-functions respectively.
\begin{prop}
\label{pr:QGai}
The Gaiotto wave functions satisfy the differential equations\footnote{A different quantum curve (without the term $\hbar \delta_{c,a}$ in \eqref{QcurveGai}) in the Gaiotto case  appears  in \cite{DHS09} where the authors use the corresponding wave function to build the dual Nekrasov instanton partition function for pure gauge theory.}, for any $a \in [r]$
\begin{equation}
\label{QcurveGai}\left(\prod_{c = 1}^{r} (Q_c + \hslash \delta_{c,a} - D) + \frac{\Lambda^r}{x}\right) x\psi^+_{a} = 0.
\end{equation}
For $b \in [r]$, this function has the explicit series representation as $\zeta \rightarrow Q_b$ 
\begin{equation}
\label{basispsi}
\psi^{+}_a \approx \left\{\begin{array}{lll} x^{\hslash^{-1}Q_a}\left(1 + \displaystyle\sum_{k \geq 1} \dfrac{(-\Lambda^{r})^k}{k!\,(\hslash^r x)^k}  \dfrac{1}{\prod_{c \neq a} \big[\frac{Q_c - Q_a}{\hslash}\big]_k}\right) && \text{if} \,\,b = a \\[10pt] B_{a,b}\,x^{\hslash^{-1}Q_b - 1}\left(1 + \displaystyle\sum_{k \geq 1} \dfrac{(-\Lambda^r)^k}{k!\,(\hslash^r x)^k} \dfrac{1}{\big[\frac{Q_a - Q_b}{\hslash} + 2\big]_k} \dfrac{1}{\prod_{c \neq a,b}\big[\frac{Q_c - Q_b}{\hslash} + 1\big]_k}\right) & & \text{if} \,\,b \neq a \end{array}\right. 
\end{equation} 
with the constants $B_{a,b}$ appearing in Corollary~\ref{leadingpsi1}.
\end{prop} 
The series in brackets in \eqref{basispsi} define entire functions of $x^{-1}$, namely the following generalised hypergeometric functions:
\begin{equation*}
\begin{split}
{}_{r - 1}F_0\Big[\begin{smallmatrix} \emptyset \\ \big(\hslash^{-1}(Q_c - Q_a)\big)_{c \neq a}\end{smallmatrix}\Big]\big(-\hslash^{-r}\Lambda^r x^{-1}\big)  & \quad \text{for}\,\, b = a, \\
{}_{r - 1}F_0\Big[\begin{smallmatrix} \emptyset \\ \big(\hslash^{-1}(Q_b - Q_a) + 1 + \delta_{c,b}\big)_{c \neq a} \end{smallmatrix}\Big]\big(-\hslash^{-r}\Lambda^r x^{-1}\big) & \quad \text{for}\,\, b \neq a.
\end{split}
\end{equation*}
Therefore, the right-hand sides of \eqref{basispsi} --- indexed by $b \in [r]$ --- provide a basis of analytic solutions of \eqref{QcurveGai} in the domain $\mathbb{C} \setminus \gamma$ for any fixed choice of path $\gamma$ between $0$ and $\infty$. The presence of the branch cut $\gamma$ is solely due to the power of $x$ in the prefactor.

\begin{prop}
\label{pr:CDOQ} 
The CDO wave functions satisfy the differential equations for any $a \in [r - 1]$
\begin{equation}
\label{Qrdiff0}\left(\prod_{c = 1}^{r - 1} (Q_c + \hslash \delta_{c,a}  - D)\cdot x + \Lambda^r \prod_{c = 1}^{r} (P_c + D)\right)\psi_{a}^{+} = 0\,.
\end{equation}  We have the explicit series representation as $\zeta \rightarrow Q_b$ for any $b \in [r-1]$, or as $\zeta \to \infty$:
\begin{equation}\label{eqpsiCDO1}
\psi_{a}^+ \approx \left\{\begin{array}{ll} x^{\hslash^{-1}Q_a }\left(1 + \displaystyle\sum_{k \geq 1} \dfrac{(-\hslash (-\Lambda)^r)^k}{k!x^k}\,\dfrac{\prod_{c = 1}^{r} \big[-\frac{P_c + Q_a}{\hslash} \big]_{k}}{\prod_{c \neq a} \big[\frac{Q_c - Q_a}{\hslash} \big]_{k}}\right) & \text{if}\,\,b = a, \\
	[15pt] B_{a,b}\,x^{\hslash^{-1}Q_b - 1}\left(1 + \displaystyle\sum_{k \geq 1} \dfrac{(-\hslash(- \Lambda)^r)^k}{k! x^k}\,\dfrac{\prod_{c = 1}^{r} \big[-\frac{P_c + Q_b}{\hslash} + 1\big]_{k}}{\prod_{c \neq b} \big[\frac{Q_c - Q_b}{\hslash} + 1 + \delta_{c,a}\big]_k}\right)  & \text{if}\,\,b \neq a, \\
	 B_{a,r} \,e^{\hslash^{-1} (-\Lambda)^{-r} x} x^{-\hslash^{-1}\left( |\mathbf P| + |\mathbf Q|\right) -1 } \left(1 + \displaystyle\sum_{k \geq 1} c_{k,a} \left(\dfrac{\hslash (-\Lambda)^r}{x}\right)^k  \right) & \text{if}\, b = r,
\end{array}\right.
\end{equation} with 
\begin{equation}\label{eq:ckdef}
	c_{k,a} = \sum_{\substack{ k_1,\ldots, k_{r-1} \geq 0 \\ k_1+\cdots+k_{r-1}= k}}  \dfrac{\prod_{c = 1}^{r-1} \big[\frac{P_{c+2} + Q_c}{\hslash}+ \delta_{a,c} \big]_{k_c}  \big[k_1 + \cdots + k_{c-1} + \sum_{d=1}^c \big(\frac{P_{d+1} + Q_d}{\hslash} + \delta_{d,a}\big) \big]_{k_c}}{k_1!\cdots k_{r-1}!},
\end{equation} where we take $P_{r+1} := P_1$ by convention.

For $a = r$, we rather have the differential equation:
\begin{equation}
\label{Qrdiff} \left(\prod_{c = 1}^{r - 1} (Q_c - D) \cdot x + \Lambda^r  \prod_{c = 1}^{r} (P_c + D)\right) \psi_{r}^+ = 0.
\end{equation}
We have the following explicit series representation as $\zeta \rightarrow x^{-1}(\infty)$:
\begin{equation} \label{eqpsiCDO2}
	\psi_{r}^+ \approx \left\{\begin{array}{ll} B_{r,b} x^{\hslash^{-1}Q_b - 1} \left(1 + \displaystyle\sum_{k \geq 1} \dfrac{(-\hslash(- \Lambda)^r)^k}{k! x^k}\,\dfrac{\prod_{c = 1}^{r} \big[-\frac{P_c + Q_b}{\hslash} + 1\big]_k}{\prod_{c \neq b} \big[\frac{Q_c - Q_b}{\hslash} + 1\big]_k} \right), & \text{if}\, \zeta \rightarrow Q_b,\\
		e^{\hslash^{-1} (-\Lambda)^{-r} x} x^{-\hslash^{-1}\left( |\mathbf P| + |\mathbf Q|\right) } \left(1 + \displaystyle\sum_{k \geq 1} c_{k,r} \left(\dfrac{\hslash (-\Lambda)^r}{x}\right)^k  \right)& \text{if}\,  \zeta \to \infty\,,
	\end{array}\right.
\end{equation} 
 with 
\begin{equation*}
	c_{k,r} = \sum_{\substack{ k_1,\ldots, k_{r-1} \geq 0 \\ k_1+\cdots+k_{r-1}= k}}  \dfrac{ \prod_{c = 1}^{r-1} \big[\frac{P_{c+2} + Q_c}{\hslash}\big]_{k_c}  \big[k_1 + \cdots + k_{c-1} + \sum_{d=1}^c \frac{P_{d+1} + Q_d}{\hslash}\big]_{k_c}}{k_1!\cdots k_{r-1}!},
\end{equation*}where we take $P_{r+1} := P_1$ by convention in the above formula.
\end{prop}  
In contrast with the Gaiotto case, now the series in brackets in the first two lines of \eqref{eqpsiCDO1} and the first line of \eqref{eqpsiCDO2} have zero radius of convergence, and they correspond to generalised hypergeometrics ${}_{r}F_{r - 2}$. In principle, a Borel resummation would be necessary to construct analytic solutions. This is related to the fact that $\omega_{0,1}$ has a double pole at a simple pole of $x$, and thus the differential equations have an irregular singularity at $\infty$. However, the quantum curve \eqref{Qrdiff0} for $a \in [r-1]$ admits an explicit basis of solutions, indexed by $b \in [r]$
\begin{equation} 
\label{basisrm1}
 x^{-\hslash^{-1} P_b} {}_{(r-1)}F_{(r-1)} \left[ \begin{smallmatrix} \big(1 - \delta_{c,a} -\hslash^{-1}(P_b+Q_c)  \big)_{c =1}^{r - 1} \\ \big(1 -\hslash^{-1}(P_b-P_{c+\delta_{c \geq b}})  \big)_{c = 1}^{r - 1}  \end{smallmatrix} \right]  \left( \hslash^{-1} (-\Lambda)^{-r} x \right).
\end{equation}
The differential equation \eqref{Qrdiff} also admits an explicit basis of solutions, indexed by $b \in [r]$
\[
x^{-\hslash^{-1} P_d} {}_{(r-1)}F_{(r-1)} \left[ \begin{smallmatrix} \big(1  -\hslash^{-1}(P_b+Q_c)  \big)_{c = 1}^{r - 1} \\ \big(1 -\hslash^{-1}(P_b-P_{c+\delta_{c \geq b}})  \big)_{c = 1}^{r - 1}  \end{smallmatrix} \right]  \left( \hslash^{-1} (-\Lambda)^{-r} x \right).
\]
To summarise, we obtain a basis of analytic solutions for the differential equations \eqref{Qrdiff0} and \eqref{Qrdiff0} in $\mathbb C \setminus \gamma$ for a fixed choice of a  branch cut $\gamma $ between $0$ and $\infty$. As before, the branch cut $\gamma$ is only due to the power of $x$ in the prefactor. We discuss the relation between this analytic basis and the formal basis obtained from wave functions in Section~\ref{sec:r-1Fr-1}.

\begin{figure}[h!]
\begin{center}
\includegraphics[width=0.65\textwidth]{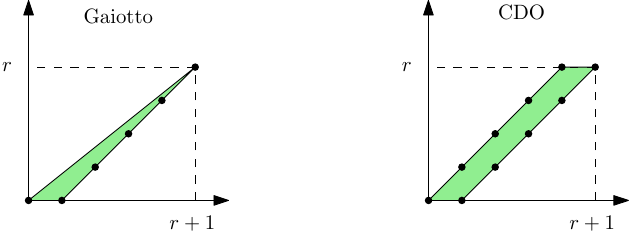}
\caption{\label{fig:Npoly} Newton polygons of the half Seiberg--Witten and CDO curves.}
\end{center}
\end{figure}

\begin{proof}[Proof of \cref{pr:QGai}]
	The equation for the half Seiberg--Witten  curve reads
	\[
	P(x,y) = \sum_{i = 0}^{r} y^{r - i}p_i(x) = 0 \qquad  \text{with} \qquad p_i(x) = (-1)^{r - i}x^{r - i + 1} e_i(\mathbf{Q}) + \delta_{i,r}\Lambda^r.
	\]
	For $a \in \{0\} \cup [r]$, we need to calculate the $\alpha_i$ defined in \cite[Section 2.2]{BE17}, which is the height of the leftmost point in the Newton polygon (see Figure~\ref{fig:Npoly}). For the half Seiberg--Witten  curve, this is $\alpha_i = (1 + 1/r)i$ and thus $\lfloor \alpha_i \rfloor = i + \delta_{i,r}$ for any $i \in \{0\} \cup [r]$. Then, \cite[Lemma 5.14]{BE17} yields
	\begin{equation}
		\label{qdiff} \left(p_r(x) + \sum_{i = 1}^{r} D^{i - 1} x^{-i} p_{r - i}(x) D - \sum_{i = 1}^{r - 1} C_{r - i}(Q_a) D^{i - 1} \hslash x\right) \psi_{Q_a}^{+} = 0,
	\end{equation}
	with the constants: 
	\begin{equation}
		\label{CQa}
		C_{r - i}(Q_a)  = \lim_{\zeta \rightarrow Q_a} x^{-i} \left( \sum_{\ell = 1}^{r - i } p_{r - i - \ell}(x) y^{\ell}\right) = \sum_{\ell = 1}^{r - i} (-1)^{\ell + i} Q_a^{\ell} e_{r - i - \ell}(\mathbf{Q}),
	\end{equation}
	recalling that $xy = \zeta$. Introducing $\mathbf{Q}^{[a]} = (Q_b)_{b \neq a}$ and writing
	\[
	e_{k}(\mathbf{Q}) = Q_a e_{k - 1}(\mathbf{Q}^{[a]}) + e_k(\mathbf{Q}^{[a]}) \qquad \text{with}\qquad e_{-1} = 0,
	\]
	we recognise in \eqref{CQa} a telescopic sum:
	\begin{equation*}
		\begin{split}
			C_{r - i}(Q_a) & = (-1)^{i - 1} \sum_{\ell = 1}^{r - i} \left((-Q_a)^{\ell + 1} e_{r - 1 - i - \ell}(\mathbf{Q}^{[a]}) - (-Q_a)^{\ell} e_{r - i - \ell}(\mathbf{Q}^{[a]})\right) \\
			& = (-1)^{i - 1}Q_a e_{r - 1 - i}(\mathbf{Q}^{[a]}).
		\end{split}
	\end{equation*}
	We insert this in \eqref{qdiff}. Then, writing $xD = Dx - \hslash x$ in the first sum, we obtain
	\[
	\left(\Lambda^r + \sum_{i = 0}^{r} (-1)^{i} e_{r - i}(\mathbf{Q}) D^{i} x -  \sum_{i = 1}^{r}  (-1)^{i}D^{i - 1}\big(e_{r - i}(\mathbf{Q}) - Q_a e_{r -1  - i}(\mathbf{Q}^{[a]})\big)\hslash x\right)\psi_{a}^{+} = 0,
	\] 
	Therefore: 
	\[
	\left(\frac{\Lambda^r}{x} + \sum_{i = 0}^{r} (-1)^i e_{r - i}(\mathbf{Q}) D^i - \hslash \sum_{i = 1}^{r} (-1)^i D^{i - 1} e_{r - i}(\mathbf{Q}^{[a]})\right)x \psi_{a}^+ = 0.
	\]
	which can be re-summed as
	\begin{equation}
		\label{diffeqn}
		\left(\frac{\Lambda^r}{x} + \prod_{b = 1}^{r} (Q_b - D) + \hslash \prod_{\substack{b = 1 \\ b \neq a}}^{r} (Q_b - D) \right) x \psi_{a}^+ = 0.
	\end{equation}
	From \cref{01regintdef}, \cref{02intlem} and \eqref{principsi}, the wave function has an expansion as $\zeta \rightarrow Q_a$ in the local coordinate $x^{-1}$, of the form
	\[
	\psi_{a}^+ \approx x^{\hslash^{-1} Q_a}\left( \sum_{k \geq 0} c_{k} x^{-k}\right) \qquad c_0 = 1
	\]
	with $c_k$ a formal Laurent series in $\hslash$. Inserting this in the differential equation \eqref{diffeqn} we get for any $k \in \mathbb{Z}_{\geq 0}$
	\[
	\Lambda^r c_k + \left(\prod_{b = 1}^{r} (Q_b - Q_a + \hslash k) + \hslash \prod_{b \neq a} (Q_b - Q_a + \hslash k)\right)c_{k + 1} = 0.
	\]
	In other words
	\[
	c_{k + 1} = \frac{-\Lambda^r c_k}{\hslash (k + 1) \prod_{b \neq a} (Q_b - Q_a + \hslash k)}.
	\] 
	With the initial condition $c_0 = 1$, we get
	\[
	c_k = \frac{(-\Lambda^{r})^{k}}{ k!\, \hslash^{r}}\,\frac{1}{\prod_{b \neq a} \prod_{\ell = 0}^{k - 1} (Q_b - Q_a - \hslash \ell)} = \frac{\Lambda^{rk}}{k! \,\hslash^{kr}} \,\frac{1}{\prod_{b \neq a} \big[\frac{Q_b - Q_a}{\hslash}\big]_k}.
	\]
	Near $\zeta \rightarrow Q_b$ for $b \neq a$, the wave function rather has an expansion in the local coordinate $x^{-1}$ of the form
	\[
	\psi_{a}^+ \approx B_{a,b} x^{\hslash^{-1}Q_b - 1} \left(\sum_{k \geq 0} c_k x^{-k}\right)\qquad c_0 = 1
	\]
	Inserting this expansion in the differential equation \eqref{diffeqn} leads again to a recursion for $c_k$ which can be solved explicitly and gives the announced result.
\end{proof}

\begin{proof}[Proof of \cref{pr:CDOQ}]
As the proof is similar to the Gaiotto case, we will be brief. The equation for the CDO curve reads
\[
P(x,y) = \sum_{i = 0}^{r} y^{r - i} p_i(x) = 0 \qquad \text{with}\qquad p_i(x) = (-1)^{r - i}x^{r + 1 - i}e_{i - 1}(\mathbf{Q}) + x^{r - i} \Lambda^r e_i(\mathbf{P}).
\]
We first observe that $\alpha_i = i$ for $i \in \{0\} \cup [r]$, hence $D_i = \hslash x \partial_x := D$. For $a \in [r - 1]$ we compute the constants
\begin{equation}\label{CforCDO}
\begin{split}
C_{r - i}(Q_a) & = \lim_{\zeta \rightarrow Q_a} \left(\sum_{\ell = 0}^{r  - 1 - i} p_{r - 1 -  i -\ell}(x)y^{r - i - \ell} x^{-(i + 1)}\right) \\
& = \sum_{\ell = 1}^{r -1 - i} (-1)^{r - \ell} Q_a^{r - i - \ell} e_{\ell - 1}(\mathbf{Q}) \\
& = \sum_{\ell = 0}^{r - i - 2} (-1)^{r - 1 - \ell} Q_a^{r - i - 1 - \ell} \big(e_\ell(\mathbf{Q}^{[a]}) + Q_a e_{\ell - 1}(\mathbf{Q}^{[a]})\big) \\
& = (-1)^{i - 1} Q_a e_{r - i - 2}(\mathbf{Q}^{[a]}).
\end{split}
\end{equation}
Then, the differential equation for the wave function is \cite[Lemma 5.14]{BE17}
\begin{equation}
\label{thedffc}
\left(p_r(x) + \sum_{i = 1}^{r} D^{i - 1} x^{-i} p_{r - i}(x) D - \sum_{i = 1}^{r - 1} D^{i - 1} C_{r - i}(Q_a) \hslash x\right) \psi^+_{a} = 0.
\end{equation}
Writing $\hslash x = Dx - xD$, the last sum combines with the second sum, and easy algebraic manipulations lead to the differential equations \eqref{Qrdiff0}. For $a = r$, we rather need the constants
\begin{equation*}
\begin{split} 
C_{r - i}(\infty) & = \lim_{\zeta \rightarrow \infty} \sum_{\ell = 0}^{r - 1 -  i} p_{\ell}(x) y^{r - i - \ell} x^{-(i + 1)} \\
& = \lim_{\zeta \rightarrow \infty} \sum_{\ell = 0}^{r - 1 - i}  \zeta^{r - i - \ell} \left((-1)^{r - \ell} e_{\ell - 1}(\mathbf{Q}) + \Lambda^r x^{-1} e_{\ell}(\mathbf{P})\right) \\
& = \lim_{\zeta \rightarrow \infty} \frac{\sum_{\ell = 0}^{r - 1 - i} \big( (-1)^{r - \ell} e_{\ell - 1}(\mathbf{Q}) \prod_{a = 1}^{r} (\zeta + P_a) - e_{\ell}(\mathbf{P}) \prod_{a = 1}^{r - 1} (Q_a - \zeta)\big)}{\prod_{a = 1}^{r} (P_a + \zeta)} \\
& =  \lim_{\zeta \rightarrow \infty} \frac{1}{\prod_{a = 1}^{r} (\zeta + P_a)} \left(\sum_{\substack{0 \leq j \leq r \\ 0 \leq \ell \leq r - i - 2}} - \sum_{\substack{0 \leq j \leq r - i - 1 \\ 0 \leq \ell \leq r - 1}}\right) e_{j}(\mathbf{P})e_{\ell}(\mathbf{Q}) \zeta^{2r - 1 - i - j - \ell} \\
& = (-1)^{i - 1}e_{r - i - 1}(\mathbf{Q}).
\end{split} 
\end{equation*} 
 Inserting this into  \eqref{thedffc} and using the previous tricks yields the differential equation \eqref{Qrdiff}. 
 
 It is then straightforward to compute the basis of series solutions  for those corresponding the expansions as $\zeta \rightarrow Q_a$ with $a \neq r$. As for the expansion of $\psi_{a}^+$ when $\zeta  \rightarrow \infty$, we first notice that a solution to the differential equation \eqref{Qrdiff0} is given by the function
 \[
x^{-\hslash^{-1}P_1} {}_{(r-1)}F_{(r - 1)}\left[\begin{smallmatrix} \big(1 - \delta_{\ell,a} - \hslash^{-1}(P_1 + Q_{b})\big)_{b = 1}^{r - 1} \\ \big(1 - \hslash^{-1}(P_1 - P_{b + 1})\big)_{b = 1}^{r - 1}\end{smallmatrix}\right]\left(\hslash^{-1}(-\Lambda)^{-r}x\right)
 \] Indeed, conjugating \eqref{Qrdiff0} by $x^{-\hslash^{-1} P_1}$ yields precisely the differential equation satisfied by the generalized hypergeometric series above  which is an analytic function in $x$. Then, the asymptotic expansion of $ { }_{(r-1)} F _{(r-1)} $ as $x \to \infty $ is computed in \cite[Theorem 4.1]{VW14}, and plugging in the exponentially growing term in the asymptotic expansion yields the result. The calculation for the expansion of $\psi^+_r$ as $\zeta \to \infty$ is analogous --  the relevant analytic solution of the differential equation is
\[
x^{-\hslash^{-1} P_1} {}_{(r-1)}F_{(r-1)} \left[ \begin{smallmatrix} \big(1  -\hslash^{-1}(P_1+Q_c)  \big)_{c = 1}^{r - 1} \\ \big(1 -\hslash^{-1}(P_1-P_{c+1})  \big)_{c = 1}^{r - 1}  \end{smallmatrix} \right]  \left( \hslash^{-1} (-\Lambda)^{-r} x \right). \qedhere
\]
 \end{proof}
 
 \subsubsection{Relating the formal and analytic bases of CDO differential equations}

 \label{sec:r-1Fr-1}
 
Fix $a \in [r]$. For the $a$-th CDO differential equation of Proposition~\ref{pr:CDOQ} we have encountered two bases of solutions.
We focus on the case $a \in [r - 1]$, as the $a = r$ case is similar.
 
 The first one is a formal WKB solution given by the wave function $\psi_{a}^{+}(\zeta)$: it is a formal series of exponential type in the formal parameter $\hslash$ whose coefficients are meromorphic functions of $\zeta$ on the CDO curve, and we obtain locally in the $x$-plane $r$ solutions by choosing $\zeta$ among the $r$ preimages of $x$ in the CDO curve. In particular, doing so in a neighborhood of $x = \infty$ gives $(r - 1)$ solutions that are up to a (explicit) power of $x$ and (not explicit) constants $B_{a,b}$ in prefactor generalised hypergeometric series ${}_{r}F_{r - 2}$ in the variable $1/x$, while the last one is more complicated. These $r$ series have zero radius of convergence. Up to a change of normalisation indicated in \eqref{comparb} below, we denote $(\overline{\psi}_{a,b}(x))_{b = 1}^{r}$ this basis, given for $b \in [r - 1]$ by
 \begin{multline}
 \overline{\psi}_{a,b}(x)  =  (-1)^{\delta_{b \neq a}} \frac{\prod_{c \neq b} \Gamma\big( \frac{Q_b - Q_c}{\hslash} + \delta_{b,a}-\delta_{c,a} \big) }{\prod_{c = 1}^{r} \Gamma\big( \frac{Q_b + P_c}{\hslash} + \delta_{b,a} \big)}   \left(   \frac{x}{\hslash  (-\Lambda)^{r}} \right) ^{\hslash^{-1}Q_b -1}
  \\
  \times  {}_{r}F_{r - 2}\left[\begin{smallmatrix} \big(-\frac{Q_b + P_c}{\hslash} + 1 - \delta_{b,a}\big)_{c = 1}^{r} \\  \big(\frac{Q_c - Q_b}{\hslash} + 1 + \delta_{c,a} - \delta_{b,a}\big)_{c \neq b} \end{smallmatrix}\right]\left(-\hslash (-\Lambda)^{r} x^{-1}\right)
 \end{multline}
and for $b = r$ by
\[
\overline{\psi}_{a,r}(x) = e^{\hslash^{-1} (-\Lambda)^{-r} x}  \left(   \frac{x}{\hslash  (-\Lambda)^{r}} \right) ^{-\hslash^{-1}(|\mathbf P| + |\mathbf Q| ) - 1} \sum_{k=0}^\infty c_{k,a}  \left(   \frac{\hslash  (-\Lambda)^{r}} {x}\right) ^k
\] 
with the constants $c_{k,a}$ introduced in \eqref{eq:ckdef}.  The new normalisation was chosen so that, for $\zeta \rightarrow Q_b$ 
 \begin{equation}
 \label{comparb}
 \psi_{a}^+(\zeta) \approx \left\{\begin{array}{lll}  \dfrac{B_{a,b} \cdot \prod_{c = 1}^{r} \Gamma\big(\frac{Q_b + P_c}{\hslash} + \delta_{b,a}\big)}{\prod_{c \neq b} \Gamma\big(\frac{Q_b - Q_c}{\hslash} + \delta_{b,a} - \delta_{c,a}\big)} \left(\hslash^{-1}(-\Lambda)^{-r}\right)^{1 - \hslash^{-1}Q_b} \cdot \overline{\psi}_{a,b}(x) && \text{if} \,\,b \in [r - 1] \\[15pt]
B_{a,r} \left(\hslash^{-1}(-\Lambda)^{-r}\right)^{\hslash^{-1}(|\mathbf{P}| + |\mathbf{Q}|) + 1} \cdot \overline{\psi}_{a,r}(x) && \text{if}\,\,b = r  \end{array}\right.
\end{equation}

The second one is a basis of analytic solutions $(\chi_{a,b}(x))_{b = 1}^{r}$, which is an entire function of $x$ multiplied by a power of $x$ that creates a branch cut from $0$ to $\infty$.  The entire functions in question are generalised hypergeometric functions ${}_{(r - 1)}F_{(r - 1)}$.  Here $\hslash$ can take any value in $\mathbb{C}^*$. We choose to normalise this basis as
\begin{multline*}
\chi_{a,b}(x)  = \frac{\prod_{c = 1}^{r - 1} \Gamma\big(1 - \frac{P_b + Q_c}{\hslash} - \delta_{c,a}\big)_{c = 1}^{r - 1}}{\prod_{c \neq b} \Gamma\big(1 - \frac{P_b - P_c}{\hslash}\big)} \left( \frac{x}{\hslash  (-\Lambda)^{r}} \right)^{-\hslash^{-1} P_b} 
\\
 \times {}_{(r-1)}F_{(r-1)} \left[ \begin{smallmatrix} \big(1 - \delta_{c,a} -\hslash^{-1}(P_b+Q_c)  \big)_{c = 1}^{r -1} \\ \big(1 -\hslash^{-1}(P_b-P_{c+\delta_{c\geq b}})  \big)_{c = 1}^{r  -1}  \end{smallmatrix} \right]  \left(   \frac{x}{\hslash  (-\Lambda)^{r}} \right) 
\end{multline*}

The all-order asymptotic expansion of this basis as $x \rightarrow \infty$ must be a linear combination of the formal basis $(\overline{\psi}_{a,b}(x))_{b = 1}^{r}$. With the asymptotic expansion of the hypergeometric functions ${}_{(r-1)} F_{(r-1)}(z)$ as $z \to \infty$ found in \cite[Chapter 16.11]{DLMF}, we find as $x \rightarrow \infty$
\begin{equation}
\label{lineartoinv}
\chi_{a,b}(x) \approx \overline{\psi}_{a,r}(x) + \sum_{d = 1}^{r - 1} \frac{\pi\,e^{-{\rm i}\pi \hslash^{-1}(P_b+Q_d)}}{\sin\left(\frac{\pi(P_b + Q_d)}{\hslash}\right)}\,\overline{\psi}_{a,d}(x), 
\end{equation}
where we have used the reflection formula $\Gamma(z)\Gamma(1-z) = \frac{\pi}{\sin(\pi z)}$.  Upon inversion of this linear system, we obtain analytic functions (instead of formal series) that are asymptotic to the formal series specified by the wave functions.

\begin{prop}
\label{psichi} As  $x \rightarrow \infty$, we have for $b \in [r - 1]$
\[
\overline{\psi}_{a,b}(x) \approx \frac{1}{\pi} \sum_{d = 1}^{r} \frac{e^{-2{\rm i}\pi \hslash^{-1} Q_b} \cdot \prod_{c=1}^{r - 1} \sin\left(\frac{\pi(P_{d} + Q_c)}{\hslash}\right) \cdot \prod_{c = 1}^{r} \sin\left(\frac{\pi(P_{c} + Q_b)}{\hslash}\right)}{\sin\left(\frac{\pi(P_d + Q_b)}{\hslash}\right) \cdot \prod_{c \neq b} \sin\left(\frac{\pi(Q_b - Q_c)}{\hslash}\right) \prod_{c \neq d} \sin\left(\frac{\pi(P_c - P_d)}{\hslash}\right)} \cdot \chi_{a,d}(x),
\]
and for $b = r$:
\[
\overline{\psi}_{a,r}(x) \approx \sum_{d = 1}^{r} \frac{e^{-{\rm i}\pi \hslash^{-1}(|\mathbf{P}| + |\mathbf{Q}|)} \cdot \prod_{c = 1}^{r - 1} \sin\left(\frac{\pi(P_d + Q_c}{\hslash}\right)}{\prod_{c \neq d} \sin\left(\frac{\pi(P_d - P_c)}{\hslash}\right)}\cdot \chi_{a,d}(x).
\]
\end{prop}

\begin{proof}
Let us introduce $p_c = e^{{\rm i}\pi \hslash^{-1} P_c}$ and $q_c = e^{{\rm i}\pi \hslash^{-1} Q_c}$. The linear system \eqref{lineartoinv} takes the form
\[
\chi_{a,b}(x) \approx \sum_{d = 1}^{r} M_{b,d} \overline{\psi}_{a,d}(x),
\]
with the matrix
\[
M_{b,d} = \left\{\begin{array}{cll} \dfrac{2{\rm i}\pi}{p_b^2 - q_d^2} & &  \text{if}\,\,d  \in [r - 1], \\[10pt] 1 && \text{if}\,\,d = r. \end{array}\right.
\]
If we introduce an auxiliary variable $q_r$ and the Cauchy matrix of size $r$
\[
\widetilde{M}_{b,d} = \frac{1}{p_b^2 - q_d^{-2}},
\]
we observe that
\[
\mathbf{M} = \lim_{q_r \rightarrow 0} \widetilde{\mathbf{M}} \cdot \mathbf{K} \qquad \text{where}\,\,\mathbf{K} = \text{diag}(2{\rm i}\pi,\ldots,2{\rm i}\pi,-q_r^{-2}).
\]
Using the well-known formula for the inverse of Cauchy matrices, we get
\[
(\widetilde{\mathbf{M}} \cdot \mathbf{K})^{-1}_{d,b} = K_{d,d}^{-1} \cdot (p_b^2 - q_d^{-2}) \prod_{\substack{c = 1 \\ c \neq b}}^{r} \frac{p_b^2 - q_c^{-2}}{q_c^{-2} - q_d^{-2}} \prod_{\substack{c = 1 \\ c \neq b}}^{r} \frac{-q_b^{-2} + p_{c}^2}{p_b^2 - p_c^2}.
\]
Taking the limit $q_r \rightarrow 0$ yields
\[
M^{-1}_{d,b} = \left\{\begin{array}{lll} \dfrac{p_b^2 - q_d^{-2}}{2{\rm i}\pi} \prod_{\substack{c = 1 \\ c \neq d}}^{r - 1} \dfrac{p_b^2 - q_c^{-2}}{- q_d^{-2} + q_c^{-2}} \prod_{\substack{c = 1 \\ c \neq b}}^{r} \dfrac{-q_d^{-2} + p_c^2}{p_b^2 - p_c^2} && \text{if} \,\,d \in [r-1] \\[15pt]
\dfrac{\prod_{c = 1}^{r - 1} (p_b^2 - q_c^{-2})}{\prod_{\substack{c = 1 \\ c \neq b}}^{r} (p_b^2 - p_c^2)} && \text{if}\,\, d = r \end{array}\right.
\]
This rational expression can be simplified back to trigonometric functions and leads to the claimed formulae.
\end{proof}

\subsubsection{Lax form} 
In this section, we write down the quantum curves obtained  for the Gaiotto and CDO curves in  \cref{pr:QGai} and \cref{pr:CDOQ} in Lax form,  i.e., as first-order matrix-valued linear ordinary differential equations. For this purpose, consider the column vector 
\[
\Psi^+_a = \left( \psi^{+,1}_a(\zeta),\psi^{+,2}_a(\zeta),\ldots, \psi^{+,r}_a(\zeta)\right)^{T},
\] where $\psi^{+,r}_a(\zeta) = \psi^+_a(\zeta)$ is the wave function obtained from topological recursion as defined in  \eqref{principsi}. The other functions   $\psi^{+,1}_a(\zeta),\ldots,\psi^{+,r-1}_a(\zeta),$  are defined in terms of the $\omega_{g,n}$ following  \cite[Section 5.2]{BE17}, but we do not recall the precise form here --- it can be deduced from the proof of \cref{prop:Lax} --- as we will not use it. From the following Lax form of the quantum curve, one can easily write $\psi^{+,j}_a(\zeta)$ for $j \neq r$ in terms of the $ \psi^+_a(\zeta)$ and its derivatives. Recall the notation $\mathbf{Q}^{[a]}$ for the tuple $(Q_b)_{b \neq a}$. 

Applying the results of \cite{BE17}, we obtain the following Lax forms.

\begin{prop}\label{prop:Lax}
	In the Gaiotto case, for any $a \in [r]$ we have
	\begin{equation}
		D \Psi^+_a = 
		\begin{pmatrix}
			e_1(\mathbf Q^{[a]}) & 1 & 0& \cdots &0 \\[4pt]
			- e_2(\mathbf Q^{[a]}) & 0 &\ddots &  &\vdots \\[4pt]
			\vdots & \vdots & & 1& 0 \\[4pt]
			(-1)^r e_{r-1}(\mathbf Q^{[a]}) & 0 & \cdots &0 & - \Lambda^r \\[4pt]
			(-1)^r x^{-1} & 0   & \cdots &0 & Q_a
		\end{pmatrix}
		\Psi^+_a\,.
	\end{equation}
	In the CDO case, for  $a \in [r-1]$ we have
	\begin{equation}
		D \Psi^+_a = 
		\begin{pmatrix}
			(-\Lambda)^{-r}e_0(\mathbf{Q}^{[a]})x - e_1(\mathbf{P}) & 1 & 0& \cdots & (-1)^{r-1} Q_a e_0\left(\mathbf Q^{[a]}\right) x \\[4pt]
			-(-\Lambda)^{-r}e_1(\mathbf{Q}^{[a]})x  - e_2(\mathbf{P}) & 0 &\ddots &  & \vdots \\[4pt]
			\vdots & \vdots & & 1&  Q_a e_{r-3}\left(\mathbf Q^{[a]}\right) x  \\[4pt]
			(-1)^{r - 2}(-\Lambda)^{-r}e_{r - 1}(\mathbf{Q}^{[a]})x  - e_{r - 1}(\mathbf{P}) & 0 & \cdots & 0& - x e_{r-1}(\mathbf Q)- \Lambda^r e_r(\mathbf P) \\[4pt]
			\Lambda^{-r} & 0   & \cdots &0 & 0
		\end{pmatrix}
		\Psi^+_a\,,
	\end{equation}
	while for $a = r$ we have 
	\begin{equation}
		D \Psi^+_r = 
		\begin{pmatrix}
			-e_1(\mathbf P) & 1 & 0& \cdots & (-1)^{r-1} \left(  e_1(\mathbf Q)+e_1(\mathbf P)\right)x  \\[4pt]
			-e_2(\mathbf P) & 0 &\ddots &  & \vdots \\[4pt]
			\vdots & \vdots & & 1&  (-1)^{r-1}\big(e_{r-2}(\mathbf P) + (-1)^{r-1} e_{r-2}(\mathbf Q)\big)x  \\[4pt]
			-e_{r-1}(\mathbf P) & 0 & \cdots & 0& \big((-1)^{r-1}e_{r-1}(\mathbf P)- e_{r-1}(\mathbf Q)\big)x - \Lambda^r e_r(\mathbf P) \\[4pt]
			\Lambda^{-r} & 0   & \cdots &0 & (-\Lambda)^{-r}x
		\end{pmatrix}
		\Psi^+_r\,.
	\end{equation}
\end{prop}
\begin{proof} 
	We start the Lax form of the quantum curve obtained by Bouchard and Eynard in \cite[Theorem 5.11]{BE17} for the half Seiberg--Witten  curve with the choice of base point $\beta = Q_a$:
		\begin{equation*}
		D 	\begin{pmatrix}
		 \varphi_1 \\[4pt]
		 \varphi_2\\[4pt]
			\vdots \\[4pt]
			\varphi_{r-1}\\[4pt]
			\varphi_r
		\end{pmatrix} = 
		\begin{pmatrix}
			e_1(\mathbf Q) & 1 & 0&  \cdots & -\hslash \frac{x C_1(Q_a)}{x e_r(\mathbf Q) + \Lambda^r} \\[4pt]
			- e_2(\mathbf Q) & 0 &\ddots &  &\vdots \\[4pt]
			\vdots & \vdots & & 1& 0 \\[4pt]
			(-1)^r e_{r-1}(\mathbf Q) & 0 & \cdots &0 & 1-\hslash \frac{x C_{r-1}(Q_a)}{x e_r(\mathbf Q) + \Lambda^r} \\[4pt]
			\frac{(-1)^{r+1} x e_r(\mathbf Q) + \Lambda^r}{x} & 0   & \cdots &0 & \hslash \frac{x e_r(\mathbf Q) }{x e_r(\mathbf Q) + \Lambda^r}
		\end{pmatrix}
	\begin{pmatrix}
		\varphi_1 \\
		\varphi_2\\
		\vdots \\
		\varphi_{r-1}\\
		\varphi_r
	\end{pmatrix},
	\end{equation*} where $\varphi_r := - (x e_r(\mathbf Q) + \Lambda^r) \psi^+_a(\zeta)$. The  remaining  $\varphi_i $ for $i \in [r-1]$ are defined in \cite[Section 5.2]{BE17} --- where they are denoted as $\psi_m(x;D)$ with the divisor $D = [\zeta] - [Q_a]$. Also recall the constants $C_i (Q_a)$ for $i \in [r-1]$ calculated in \eqref{CQa}.
	By defining 
	\[
	\Psi^+_a := \begin{pNiceArray}{ccc|c}
		\Block{3-3}{\textbf{Id}_{r - 1}} & &  & \frac{C_1(Q_a) x}{xe_r(\mathbf Q) + \Lambda^r}\\[4pt]
		  & & & \vdots \\[4pt]
		  &&& \frac{C_{r-1}(Q_a) x}{xe_r(\mathbf Q) + \Lambda^r}\\[4pt]
		\hline
		\Block{1-3}{\textbf{0}} &&  & -\frac{1}{xe_r(\mathbf Q) + \Lambda^r}
	\end{pNiceArray} \begin{pmatrix}
	\varphi_1 \\[4pt]
	\vdots \\[4pt]
	\varphi_{r-1}\\[4pt]
	\varphi_r
	\end{pmatrix},
	\]
the differential equation for $\Psi^+_a$ takes the claimed Lax form, after some simplifications arising from the form of the $C_i(Q_a)$.
	
	In the CDO case, with base point $\beta  = Q_a$ and $a \in [r-1]$, we get
	\begin{equation*}
		D 	\begin{pmatrix}
			\varphi_1 \\[4pt]
			\varphi_2\\[4pt]
			\vdots \\[4pt]
			\varphi_{r-1}\\[4pt]
			\varphi_r
		\end{pmatrix} = 
		\begin{pmatrix}
			-\frac{p_1(x)}{x^{r-1}\Lambda^r}& 1 & 0&  \cdots & -\hslash \frac{x C_1(Q_a)}{p_r(x)} \\[4pt]
			-\frac{p_2(x)}{x^{r-2}\Lambda^r}& 0 &\ddots &  &\vdots \\[4pt]
			\vdots & \vdots & & 1& 0 \\[4pt]
				-\frac{p_{r-1}(x)}{x\Lambda^r} & 0 & \cdots &0 & 1-\hslash \frac{x C_{r-1}(Q_a)}{p_r(x)} \\[4pt]
				-\frac{p_r(x)}{\Lambda^r} & 0   & \cdots &0 & \hslash \frac{x e_{r-1}(\mathbf Q) }{p_r(x)} 
		\end{pmatrix}
		\begin{pmatrix}
			\varphi_1 \\[4pt]
			\varphi_2\\[4pt]
			\vdots \\[4pt]
			\varphi_{r-1}\\[4pt]
			\varphi_r
		\end{pmatrix},
	\end{equation*} where $\varphi_r := - p_r(x) \psi^+_a(\zeta)$. Again, the remaining  $\varphi_i $ for $i \in [r-1]$ are defined in \cite[Section 5.2]{BE17}. Also recall the $C_i (Q_a)$ for $i \in [r-1]$ calculated in \eqref{CforCDO} and that
	\[
	\forall i \in \{0\} \sqcup [r] \qquad p_i(x) = (-1)^{r - i}x^{r + 1 - i}e_{i - 1}(\mathbf{Q}) + x^{r - i} \Lambda^r e_i(\mathbf{P}).
	\]
		With the definition
\[
\Psi^+_a := \begin{pNiceArray}{ccc|c}
	\Block{3-3}{\mathbf{Id}_{r-1}} & &  & \frac{x C_1(Q_a) }{p_r(x)}\\[4pt]
	& & & \vdots \\[4pt]
	&&& \frac{x C_{r-1}(Q_a) }{p_r(x)}\\[4pt]
	\hline
	\Block{1-3}{\mathbf{0}} &&  & -\frac{1}{p_r(x)}
\end{pNiceArray} \begin{pmatrix}
	\varphi_1 \\[4pt]
	\vdots \\[4pt]
	\varphi_{r-1}\\[4pt]
	\varphi_r
\end{pmatrix},
\]
we obtain the claimed Lax form. We omit the details in the case with base point $\beta = \infty$, which is similar.
\end{proof}

\vspace{0.5cm}

\subsection{Determinantal formulae}
\label{sec:detform}

The partition function of topological recursion on genus $0$ spectral curves, which includes the half Seiberg--Witten and CDO spectral curves, satisfies KP integrability. More precisely, we have

\begin{prop}
\label{prop:hKP} The Gaiotto and the CDO vectors are $\hslash$-KP tau functions, separately for each $a \in [r]$ in the series of times $(J_{-k}^{a})_{k \in \mathbb{Z}_{> 0}}$.
\end{prop}
\begin{proof}
For the half Seiberg--Witten or CDO spectral curve, $x(\zeta)$ is a rational function and $x(\zeta)y(\zeta) = \zeta$. The series expansions of $\big(\omega_{g,n}\,\,:\,\,(g,n) \in \mathbb{Z}_{\geq 0} \times \mathbb{Z}_{> 0}\big)$ for these spectral curves at $\zeta = Q_a$ in the variable $x(\zeta)$ is encoded in $\ket{\Gamma_{\Lambda}}$ seen as a generating series in the formal variables $(J_{-k}^{a})_{k \in \mathbb{Z}_{> 0}}$. $\hslash$-KP integrability is then covered by \cite{ABDBKS} --- see also \cite{Zhou15}.
\end{proof}

This has several remarkable consequences which we now explain.

\subsubsection{Bispinor in terms of wave functions}

\begin{defn}
We define the bispinor
\[
K(\zeta_1,\zeta_2) = \frac{\sqrt{\dd \zeta_1 \dd \zeta_2}}{\zeta_1 - \zeta_2} \exp\left( \hslash^{-1} \int_{\zeta_2}^{\zeta_1} \omega_{0,1} + \sum_{(g,n) \in \mathbb{Z}_{\geq 0} \times \mathbb{Z}_{> 0}} \frac{\hslash^{2g - 2 + n}}{n!} \int_{\zeta_2}^{\zeta_1} \cdots \int_{\zeta_2}^{\zeta_1} \omega_{g,n}\right).
\]
\end{defn}

There is an analogue of the Christoffel--Darboux formula for this bispinor.
\begin{prop}
\label{bispiq} For the half Seiberg--Witten or the CDO spectral curve, we have
\[
K(\zeta_1,\zeta_2) = - \frac{\sum_{a = 1}^{r} \psi_a^+(\zeta_1)\psi_a^-(\zeta_2) \,\sqrt{\dd x(\zeta_1) \dd x(\zeta_2)}}{x(\zeta_1) - x(\zeta_2)}.
\]
\end{prop} 
\begin{proof}
According to \cite[Conjecture 7.4]{BE12}, the partition function of topological recursion for compact spectral curves with $x$ and $y$ meromorphic satisfies Hirota bilinear difference equations. Since $\hslash$-KP integrability is known to be equivalent to Hirota bilinear difference equations, Proposition~\ref{prop:hKP} justifies that the conjecture holds for the Gaiotto and CDO spectral curves (since these curves have genus $0$, all Theta corrections in \cite{BE12} can be ignored). Therefore, we can use \cite[Theorem 8.3]{BE12}, which shows that the bispinor can be expressed as a bilinear expression in the wave function and the dual wave function\footnote{There is a missing minus sign in \cite[Theorem 8.3]{BE12} --- compare to \cite[Proposition 3.5]{BE12} where the sign appears correctly.}
\[
K(\zeta_1,\zeta_2) = - \frac{\sum_{a = 1}^{r} \psi_a^+(\zeta_1)A_{a,b}\psi_a^-(\zeta_2) \,\sqrt{\dd x(\zeta_1) \dd x(\zeta_2)}}{x(\zeta_1) - x(\zeta_2)},
\]
where
\begin{equation}
\label{Amatrix}
A^{-1}_{a,b} = \sum_{\zeta \in x^{-1}(x_0)} \psi_{a}^+(\zeta)\psi_{b}^-(\zeta)
\end{equation}
is a matrix which is independent of $x_0 \in \mathbb{P}^1$. Recall that $\psi_a^-$ is simply $\psi_a^+$ with $\hslash$ replaced by $-\hslash$. We compute the matrix $A$ by evaluating the right-hand side of \eqref{Amatrix} as $x_0 \rightarrow \infty$ using Corollaries~\ref{leadingpsi1} and \ref{leadingpsi2}. For $b \neq a$ or for $b = a$ but $\zeta$ not approaching $Q_a$, we have $\psi_{a}^+(z)\psi_{b}^-(z) = O(1/x_0)$ as $x_0 \rightarrow \infty$. Therefore $A$ is a diagonal matrix and the only term surviving for $a = b$ is the one corresponding to $\zeta \rightarrow Q_a$: it involves the product of constants $C_{a,a} \cdot C_{a,a}|_{\hslash \mapsto -\hslash}$ which is always $1$. Thus $A = \text{Id}$.
\end{proof}

\begin{defn}
The correlators are defined at least as formal series in $\hslash$ as
\[
\omega_{n}(\zeta_1,\ldots,\zeta_n) = \sum_{g \geq 0} \hslash^{2g - 2 + n} \omega_{g,n}(\zeta_1,\ldots,\zeta_n).
\]
The disconnected correlators are then defined as formal Laurent series in $\hslash$
\[
\omega^{\bullet}_{n}(\zeta_1,\ldots,\zeta_n) = \sum_{\mathbf{L} \vdash [n]} \prod_{L \in \mathbf{L}} \omega_{|L|}(\zeta_L).
\]
\end{defn}

\begin{prop}
\label{bispiq2} For the half Seiberg--Witten or the CDO spectral curve, we have for $n \geq 2$
\[
\omega_{n}(\zeta_1,\ldots,\zeta_n) = (-1)^{n + 1} \sum_{\sigma  = n\text{-cycle}} \prod_{i = 1}^n K(\zeta_i,\zeta_{\sigma(i)})
\]
in terms of the bispinor of Proposition~\ref{bispiq}. For $n = 1$, the bispinor is singular on the diagonal so this formula would not make sense, but a regularisation of it holds
\[
\omega_{1}(\zeta) = \hslash^{-1} \omega_{0,1}(\zeta) +  \left(\lim_{\zeta' \rightarrow \zeta}  \frac{\psi(\zeta',\zeta) e^{-\hslash^{-1} \int_{\zeta}^{\zeta'} \omega_{0,1}}}{\sqrt{\dd x(\zeta)\dd x(\zeta')}} - \frac{1}{x(\zeta) - x(\zeta')}\right) \dd x(\zeta).
\] 
\end{prop}
For the disconnected correlators, this leads to the determinantal formulae
\[
\omega^{\bullet}_n(\zeta_1,\ldots,\zeta_n) = \,\,:\,\det_{1 \leq i,j \leq n} K(\zeta_i,\zeta_j)\,: \,,
\]
where the colons indicate that all factors $K(\zeta_i,\zeta_i)$ appearing in the determinant should be replaced by $\omega_{1}(\zeta_i)$.

\begin{proof}
For $n = 1$ this is e.g. \cite[Lemma 8.1]{BE12}. For $n \geq 2$ this is \cite[Theorem 8.1]{BE12} conditionally to \cite[Conjecture 7.4]{BE12}, but the conjecture holds in our case as explained in the proof of Proposition~\ref{bispiq}. Alternatively: the determinantal formulae are known by \cite[Section 3.1 and 3.2]{ABDBKS} to be equivalent to the $\hslash$-KP integrability stated in Proposition~\ref{prop:hKP}.
\end{proof}

\begin{rem} \label{connecrem} In the Gaiotto case, Propositions~\ref{bispiq} and \ref{bispiq2} together with the comment after Proposition~\ref{pr:QGai} would offer a complete description of the correlators $\omega_n$ as analytic functions of $\hslash$, if the constants $B_{a,b}$ of Corollary~\ref{leadingpsi2} could be explicitly computed. In the CDO case, one should rather use Proposition~\ref{psichi} for the analytic description, but again, the constants $B_{a,b}$ of Corollary~\ref{leadingpsi2} would have to be computed. These constants are expressed in terms of the connection coefficients $\widetilde{\psi}_{a}^+(Q_b)$ and have to do with the way the wave function computed by topological recursion is normalised. As of writing, we do not know how to obtain a description of the connection coefficients as \emph{analytic functions} of $\hslash$ or how to compute them, but we will relate them to the topological recursion free energies in Proposition~\ref{FGcon}.
\end{rem}
\vspace{0.5cm}

\subsection{Free energies}
\label{S54}
\medskip

\subsubsection{In topological recursion}

\medskip

\textbf{A. Definition.} Given a spectral curve $(S,x,y,\omega_{0,2})$, besides the correlators, the topological recursion also has a natural definition of \textit{free energies} $F_{g} = \omega_{g,0}$  \cite{EORev}:
\begin{equation}\label{eq:Fgdef}
\forall g \in \mathbb{Z}_{\geq 2} \qquad 	F_g =  \frac{1}{2-2g} \sum_{\rho \in \operatorname{Ram}(S)} \Res_{ \zeta = \rho } \left[  \left( \frac{1}{2} \int_{\sigma_{\rho}(\zeta)}^\zeta \omega_{0,1} \right) \omega_{g,1}(\zeta)\right].
\end{equation}
The free energies $F_0$ and $F_1$ are defined differently, see \cite{EORev}.

\medskip

\textbf{B. Deformations.} The free energies satisfy a functional equation that relate them to the connection matrix $\widetilde{\psi}_{a}^+(Q_b)$ that appeared in Definition~\ref{stablepsi} and Corollary~\ref{leadingpsi1} or \ref{leadingpsi2}.

\begin{prop}
\label{FGcon}Let $F_g = F_g(\mathbf{Q})$ be the topological recursion free energies of the half Seiberg--Witten or the CDO spectral curve, seen as functions of the charges $\mathbf{Q} = (Q_1,\ldots,Q_r)$. Let $\mathbf{e}_1,\ldots,\mathbf{e}_r$ denote the standard basis of $\mathbb{C}^r$. For any $a \neq b$, 
\[ 
\widetilde{\psi}_{a}^+(Q_b) = \widetilde{\psi}_{b}^-(Q_a) = \exp\left(\sum_{g \geq 2} \hslash^{2g - 2}\big(F_g(\mathbf{Q} + \hslash(\mathbf{e}_b - \mathbf{e}_a)) - F_g(\mathbf{Q})\big) \right). 
\]
\end{prop}
\begin{proof} We use the properties of  topological recursion under deformations of spectral curves. For $a \in [r]$, we introduce the operator
\[
\nabla_a f = \frac{\dd}{\dd Q_a} f \Big|_{x\,\,\text{fixed}},
\]
which acts on functions or $n$-differentials on the Gaiotto (or CDO) curve. In contrast, we will use $\partial_{Q_a}$ to denote the partial derivative at fixed $\zeta$. Since $\omega_{0,1}(\zeta) = \zeta \,\dd \ln x(\zeta)$ in both the Gaiotto or the CDO case, we compute
\[
\Omega_{a}(\zeta) := \nabla_{Q_a} \omega_{0,1}(\zeta) = -( \partial_{Q_a} \ln x) \dd z = \frac{\dd z}{Q_a - z} = \int_{Q_a}^{\infty}  \omega_{0,2}(z,\cdot),
\]
where we recall that $\omega_{0,2}(\zeta_1,\zeta_2) = \frac{\dd \zeta_1 \dd \zeta_2}{(\zeta_1 - \zeta_2)^2}$.  By \cite[Theorem 5.1]{EOFg}, we have for any $(g,m) \in \mathbb{Z}_{\geq 0}^2$
\begin{equation}
\label{nablaa}
\nabla_{Q_a} \omega_{g,m}(\zeta_1,\ldots,\zeta_m) = \int_{Q_a}^{\infty} \omega_{g,m+ 1}(\cdot,\zeta_1,\ldots,\zeta_m),
\end{equation}
where for $(g,m) = (0,0)$ one should use  the regularised integral of $\omega_{0,1}$ from Definition \ref{01regintdef} on the right hand side.
The $x$-projection of the integration path from $\zeta' = Q_a$ to $\zeta' = \infty$ is a loop based at $x = \infty$. Therefore, deforming $Q_a$ does not act on the integration contour and we can iterate \eqref{nablaa} to obtain for any $m \in \mathbb{Z}_{\geq 0}$
\[
\nabla_{Q_a}^n \omega_{g,m}(\zeta_1,\ldots,\zeta_m) = \int_{Q_a}^{\infty} \cdots \int_{Q_a}^{\infty} \omega_{g,m+n}(\cdot,\zeta_1,\ldots,\zeta_m).
\]
The same formulae hold for $\nabla_{Q_b} - \nabla_{Q_a}$ if we rather integrate from $Q_b$ to $Q_a$. We apply this to $F_g = \omega_{g,m=0}$ for $g \geq 2$. These are analytic functions of $Q_1,\ldots,Q_r$ in the domain where they are pairwise distinct, and by Taylor expansion we find
\[
F_g\big(\mathbf{Q} + \hslash(\mathbf{e}_b - \mathbf{e}_a)\big) = F_g(\mathbf{Q}) + \sum_{n = 1}^{\infty} \frac{\hslash^n}{n!} \int_{Q_b}^{Q_a}\cdots \int_{Q_b}^{Q_a} \omega_{g,n}.
\]
Multiplying by $\hslash^{2g - 2}$ and summing over $g \geq 2$ gives the claim after comparison with Definition~\ref{stablepsi}.
\end{proof}

In view of Propositions~\ref{bispiq}-\ref{bispiq2} and \ref{FGcon} it would be interesting to find the Stokes matrices, the connection coefficients $\widetilde{\psi}_{a}^+(Q_b)$ and perform the resurgence analysis for the solutions of the differential equations of Propositions~\ref{pr:QGai}-\ref{pr:CDOQ}. In particular, this would give analytic solutions that can be used for non-perturbative (with respect to $\hslash$) computations of the wave functions, the bispinor, the correlators and the free energies. Conversely, if we had closed formulae for the free energies, we could obtain information about the connection coefficients $\widetilde{\psi}_{a}^+(Q_b)$.

\subsubsection{In gauge theory} \label{sec:gauge}

The instanton (or non-perturbative) part of the Nekrasov partition function is obtained from the Gaiotto vector by computing the square-norm:
\[
Z_{\text{Nek}} = \langle \Gamma_{\Lambda}\!\ket{\Gamma_{\Lambda}}
\]
The full partition function of the underlying $\mathcal{N} = 2$ supersymmetric gauge theory also has a perturbative part:
\[
\overline{Z}_{\text{Nek}} = \exp\left(\sum_{g \geq 0} \hslash^{2g - 2} F_g^{\text{pert}}\right) \cdot Z_{\text{Nek}}
\]
The expression for the perturbative part can be found in \cite[Equation 3.5 and 3.8]{NO06} --- see also \cite{NY03}, but beware of the opposite global sign compared to  \cite{NO06}. To compare our notations and the notations between the different references: our $r$ is also $r$ in \cite{NY03} but $N$ in \cite{NO06}, our $\hslash$ is their $\hslash^2$ (the transformation was already met in \eqref{evaop}), our $(Q_1,\ldots,Q_r)$ is their $(a_1,\ldots,a_r)$, and comparing the degenerate limit of their Seiberg--Witten curve with our half Seiberg--Witten  curve \eqref{eq:Gaiottocurve}, our $\Lambda^r$ is their $(-1)^{N + 1} \Lambda^N$. With our notations, the expressions in \cite{NO06} yield
\begin{equation}
\label{nonpertZ}
\begin{split}
F_0^{\text{pert}} & =  \sum_{1 \leq a < b \leq r} (Q_a - Q_b)^2\left[\frac{3}{2} + \frac{1}{2}\ln\left(-\frac{\Lambda^r}{(Q_a - Q_b)^2}\right)\right], \\ 
F_1^{\text{pert}} & = -\sum_{1 \leq a < b \leq r} \frac{1}{12}\ln\left(-\frac{\Lambda^r}{(Q_a - Q_b)^2}\right), \\ 
F_g^{\text{pert}} & = - \sum_{1 \leq a < b \leq r} \frac{2B_{2g} (Q_a - Q_b)^{2 - 2g}}{2g(2g - 2)}.   
\end{split}
\end{equation}
Although the Whittaker vectors have been normalised in \eqref{Gformalexp} to have no constant terms, a more natural normalisation from the topological recursion perspective would be
\[
\ket{\widetilde{\Gamma}_{\Lambda}} = \exp\left(\sum_{g \geq 0} \hslash^{g - 1} F_g\right) \ket{\Gamma_{\Lambda}},
\]
where $F_g$ are the free energies of the spectral curve. This would give
\[
\langle \widetilde{\Gamma}_{\Lambda}\!\ket{\widetilde{\Gamma}_{\Lambda}} = \exp\left(\sum_{g \geq 0} \hslash^{g - 1} 2F_g\right) Z_{\text{Nek}}.
\]
We expect that this agrees with the partition function of the supersymmetric gauge theory --- up to the change $\hslash \mapsto \hslash^2$ already met in \eqref{evaop}. In other words, we expect
\begin{equation}
\label{FGpertF}
\forall g \in \mathbb{Z}_{\geq 0} \qquad F_g^{\text{pert}} = 2F_g. 
\end{equation}

\medskip

\subsubsection{Comparisons and conjectures}

\medskip

The half Seiberg--Witten and CDO curves for $r = 2$ already appeared in disguise in \cite{IKT19}, where their topological recursion free energies are computed. In the half Seiberg--Witten case, the results match the expectation from \eqref{nonpertZ}-\eqref{FGpertF}.
\begin{prop}
	For the $r = 2$ half Seiberg--Witten curve \eqref{eq:Gcurve}, we have
	\begin{equation*}
	\begin{split}
	F_0 & = \frac{3}{4}(Q_1-Q_2)^2 + \frac{(Q_1-Q_2)^2}{4} \ln \left(\frac{\Lambda^2}{(Q_1-Q_2)^2}\right),\\ 
	F_1 & = -\frac{1}{24} \ln \left(\frac{\Lambda^2}{(Q_1-Q_2)^2}\right), 
	\end{split}
	\end{equation*}
	and for $g \geq 2$
	\[
	F_{g} = -\frac{B_{2g} (Q_1 - Q_2)^{2 - 2g}}{2g(2g-2)}. 
	\]
	\end{prop}
	\begin{prop} 
	For the $r = 2$ CDO curve \eqref{eq:algcurve}, we have 
	\begin{equation*}
	\begin{split}
	F_0 & = \frac{3}{4}\big((P_1 - P_2)^2 - (Q_1 - P_1)^2 - (Q_1 - P_2)^2\big) + \frac{1}{4}\left[(P_1 - P_2)^2\ln\left(\frac{\Lambda^2}{(P_1 - P_2)^2}\right)\right. \\
	& \quad \left. - (Q_1 - P_1)^2\ln\left(\frac{\Lambda^2}{(Q_1 - P_1)^2}\right) - (Q_1 - P_2)^2\ln\left(\frac{\Lambda^2}{(Q_1 - P_2)^2}\right)\right] \\ 
	F_1 & = \frac{1}{24}\ln\left(\frac{\Lambda^2 (P_1 - P_2)^2}{(Q_1 - P_1)^2(Q_1 - P_2)^2}\right), 
	\end{split} 
	\end{equation*} 
	and for $g \geq 2$:
	\[
	F_{g} = \frac{B_{2g}}{2g(2g-2)}\left((Q_1 + P_1)^{2 -2g} + (Q_1 + P_2)^{2 - 2g} - (P_1 - P_2)^{2 - 2g}\right). 
	\] 
\end{prop}
\begin{proof}
	Up to a change of variables, the half Seiberg--Witten curve coincides with the Bessel curve with the parameter $\lambda_0 = \frac{Q_1-Q_2}{2}$, cf. \cite[Table 1.2 and Section 2.3.5]{IKT19}\footnote{Note that the Bessel curve  is defined as $y^2 = \frac{x + 4 \lambda_0^2}{4 x^2}$ in \cite[Section 2.3.5]{IKT19}, but there is a typo in \cite[Table 1.1]{IKT19} where the  curve appears as $y^2 = \frac{x +  \lambda_0^2}{4 x^2} $.}. Taking into account the extra minus sign in the definition of the $F_g$ in \eqref{eq:Fgdef} and the $\Lambda $ in the half Seiberg--Witten  curve gives the result. Likewise, up to a change of variables, the CDO curve is identified with the Kummer curve with $\lambda_0 = \frac{P_1-P_2}{2}$ and $\lambda_\infty = Q_1 - \frac{P_1+P_2}{2}$, for which the free energy appears in \cite[Table 1.1 and Table 1.2]{IKT19}.
\end{proof}

For $r > 2$, closed formulae for the topological free energy for the Gaiotto or the CDO curves are not available. Based on the expectation \eqref{nonpertZ}-\eqref{FGpertF} for the half Seiberg--Witten  curve, we are led to propose the following conjecture.

\begin{conj}
For the half Seiberg--Witten  curve, we have
\begin{equation*}
\begin{split}
F_0 & = \sum_{1 \leq a < b \leq r} \frac{3}{4}(Q_a - Q_b)^2 + \frac{1}{4}(Q_a - Q_b)^2\ln\left(\frac{\Lambda^{r}}{(Q_a - Q_b)^2}\right), \\ 
F_1 & = -\sum_{1 \leq a < b \leq r} \frac{1}{24}\ln\left(\frac{\Lambda^{r}}{(Q_a - Q_b)^2}\right), 
\end{split}
\end{equation*}
and for $g \geq 2$:
\[
F_g = -\frac{B_{2g}}{2g(2g - 2)} \sum_{1 \leq a < b \leq r - 1} (Q_a - Q_b)^{2 - 2g}. 
\]
\end{conj}

\newpage

\bibliographystyle{alpha}
\bibliography{BibliographyTRonSW}

\end{document}